\documentclass[a4paper,onecolumn,unpublished]{quantumarticle}
\pdfoutput=1
\usepackage[utf8]{inputenc}
\usepackage[english]{babel}
\usepackage[T1]{fontenc}
\usepackage{float}
\usepackage[10pt]{moresize}
\usepackage{amsmath}
\usepackage{mathtools}
\usepackage{hyperref}
\usepackage{cleveref}

\usepackage[titletoc,title]{appendix}
\usepackage{hyperref}
\hypersetup{
    colorlinks=true,
    citecolor=myblue,
    linkcolor=myblue,
    filecolor=myblue,
    urlcolor=myblue,
    breaklinks=true
}
\usepackage[marginal]{footmisc}
\usepackage[dvipsnames]{xcolor}
\usepackage{url}

\makeatletter
\g@addto@macro{\UrlBreaks}{\UrlOrds}
\makeatother

\usepackage{mathtools}
\usepackage{amsmath}
\usepackage{amssymb}
\usepackage[shortlabels]{enumitem}
\usepackage{graphicx,epic,eepic,epsfig,latexsym,verbatim}
\usepackage{dsfont}

\usepackage{subcaption}
\usepackage{framed}

\definecolor{shadecolor}{rgb}{0.9,0.9,0.9}
\definecolor{LightGray}{RGB}{220,220,220}
\definecolor{LinkColor}{RGB}{167,20,49}
\definecolor{myred}{RGB}{236, 17, 0}
\definecolor{myblue}{RGB}{10, 88, 153}
\definecolor{myorange}{RGB}{236, 137, 0}
\definecolor{mygreen}{RGB}{26, 152, 81}
\definecolor{myamber}{RGB}{239,187,26}
\definecolor{mypurple}{RGB}{123, 57, 160}
\definecolor{myteal}{RGB}{42, 206, 181}

\usepackage{float}
\usepackage{tcolorbox}
\usepackage{relsize}

\usepackage{theorem}
\newtheorem{definition}{Definition}
\newtheorem{proposition}[definition]{Proposition}

\newtheorem{theorem}[definition]{Theorem}

\def\squareforqed{\hbox{\rlap{$\sqcap$}$\sqcup$}}
\def\qed{\ifmmode\squareforqed\else{\unskip\nobreak\hfil
\penalty50\hskip1em\null\nobreak\hfil\squareforqed
\parfillskip=0pt\finalhyphendemerits=0\endgraf}\fi}
\def\endenv{\ifmmode\;\else{\unskip\nobreak\hfil
\penalty50\hskip1em\null\nobreak\hfil\;
\parfillskip=0pt\finalhyphendemerits=0\endgraf}\fi}
\newenvironment{proof}{\noindent \textbf{{Proof~} }}{\hfill $\blacksquare$ \vspace{0.2cm}}

\newcounter{remark}
\newenvironment{remark}[1][]{\refstepcounter{remark}\par\medskip\noindent%
\textbf{Remark~\theremark #1} }{\medskip}

\newcounter{example}

\mathchardef\ordinarycolon\mathcode`\:
\mathcode`\:=\string"8000
\def\vcentcolon{\mathrel{\mathop\ordinarycolon}}
\begingroup \catcode`\:=\active
  \lowercase{\endgroup
  \let :\vcentcolon
  }

\usepackage{colortbl}
\usepackage{pifont}
\definecolor{Gray}{gray}{0.92}
\definecolor{Gray2}{gray}{0.75}
\definecolor{maroon}{cmyk}{0,0.87,0.68,0.32}
\usepackage{booktabs}
\usepackage{makecell}
\usepackage{diagbox}
\usepackage{multirow}

\newcommand{\nc}{\newcommand}
\nc{\rnc}{\renewcommand}
\nc{\bra}[1]{\langle#1|}
\nc{\ket}[1]{|#1\rangle}
\nc{\<}{\langle}
\rnc{\>}{\rangle}
\nc{\ketbra}[2]{|#1\rangle\!\langle#2|}
\nc{\braket}[2]{\langle#1|#2\rangle}
\nc{\braandket}[3]{\langle #1|#2|#3\rangle}
\nc{\proj}[1]{| #1\rangle\!\langle #1 |}
\nc{\avg}[1]{\langle#1\rangle}
\nc{\Rank}{\operatorname{Rank}}
\nc{\smfrac}[2]{\mbox{$\frac{#1}{#2}$}}
\nc{\tr}{\operatorname{Tr}}
\nc{\ox}{\otimes}

\nc{\cA}{{\cal A}}
\nc{\cB}{{\cal B}}
\nc{\cC}{{\cal C}}
\nc{\cD}{{\cal D}}
\nc{\cE}{{\cal E}}
\nc{\cF}{{\cal F}}
\nc{\cG}{{\cal G}}
\nc{\cH}{{\cal H}}
\nc{\cI}{{\cal I}}
\nc{\cJ}{{\cal J}}
\nc{\cK}{{\cal K}}
\nc{\cL}{{\cal L}}
\nc{\cM}{{\cal M}}
\nc{\cN}{{\cal N}}
\nc{\cO}{{\cal O}}
\nc{\cP}{{\cal P}}
\nc{\cQ}{{\cal Q}}
\nc{\cR}{{\cal R}}
\nc{\cS}{{\cal S}}
\nc{\cT}{{\cal T}}
\nc{\cV}{{\cal V}}
\nc{\cX}{{\cal X}}
\nc{\cY}{{\cal Y}}
\nc{\cZ}{{\cal Z}}
\nc{\cW}{{\cal W}}

\nc{\RR}{{{\mathbb R}}}
\nc{\CC}{{{\mathbb C}}}
\nc{\FF}{{{\mathbb F}}}
\nc{\NN}{{{\mathbb N}}}
\nc{\ZZ}{{{\mathbb Z}}}
\nc{\PP}{{{\mathbb P}}}
\nc{\QQ}{{{\mathbb Q}}}
\nc{\UU}{{{\mathbb U}}}
\nc{\EE}{{{\mathbb E}}}
\nc{\id}{{\operatorname{id}}}

\nc{\1}{{\mathds{1}}}
\nc{\supp}{{\operatorname{supp}}}

\usepackage{tikz}
\usetikzlibrary{decorations,snakes,decorations.markings,positioning,shadows.blur,calc,quantikz2}

\usetikzlibrary{matrix}
\newcommand\x{\times}

\makeatletter
\def\grd@save@target#1{%
  \def\grd@target{#1}}
\def\grd@save@start#1{%
  \def\grd@start{#1}}
\tikzset{
  grid with coordinates/.style={
    to path={%
      \pgfextra{%
        \edef\grd@@target{(\tikztotarget)}%
        \tikz@scan@one@point\grd@save@target\grd@@target\relax
        \edef\grd@@start{(\tikztostart)}%
        \tikz@scan@one@point\grd@save@start\grd@@start\relax
        \draw[minor help lines,magenta] (\tikztostart) grid (\tikztotarget);
        \draw[major help lines] (\tikztostart) grid (\tikztotarget);
        \grd@start
        \pgfmathsetmacro{\grd@xa}{\the\pgf@x/1cm}
        \pgfmathsetmacro{\grd@ya}{\the\pgf@y/1cm}
        \grd@target
        \pgfmathsetmacro{\grd@xb}{\the\pgf@x/1cm}
        \pgfmathsetmacro{\grd@yb}{\the\pgf@y/1cm}
        \pgfmathsetmacro{\grd@xc}{\grd@xa + \pgfkeysvalueof{/tikz/grid with coordinates/major step}}
        \pgfmathsetmacro{\grd@yc}{\grd@ya + \pgfkeysvalueof{/tikz/grid with coordinates/major step}}
        \foreach \x in {\grd@xa,\grd@xc,...,\grd@xb}
        \node[anchor=north] at (\x,\grd@ya) {\pgfmathprintnumber{\x}};
        \foreach \y in {\grd@ya,\grd@yc,...,\grd@yb}
        \node[anchor=east] at (\grd@xa,\y) {\pgfmathprintnumber{\y}};
      }
    }
  },
  minor help lines/.style={
    help lines,
    step=\pgfkeysvalueof{/tikz/grid with coordinates/minor step}
  },
  major help lines/.style={
    help lines,
    line width=\pgfkeysvalueof{/tikz/grid with coordinates/major line width},
    step=\pgfkeysvalueof{/tikz/grid with coordinates/major step}
  },
  grid with coordinates/.cd,
  minor step/.initial=.2,
  major step/.initial=1,
  major line width/.initial=2pt,
}
\makeatother

\makeatletter
\newcommand*\rel@kern[1]{\kern#1\dimexpr\macc@kerna}
\newcommand*\widebar[1]{%
  \begingroup
  \def\mathaccent##1##2{%
    \rel@kern{0.8}%
    \overline{\rel@kern{-0.8}\macc@nucleus\rel@kern{0.2}}%
    \rel@kern{-0.2}%
  }%
  \macc@depth\@ne
  \let\math@bgroup\@empty \let\math@egroup\macc@set@skewchar
  \mathsurround\z@ \frozen@everymath{\mathgroup\macc@group\relax}%
  \macc@set@skewchar\relax
  \let\mathaccentV\macc@nested@a
  \macc@nested@a\relax111{#1}%
  \endgroup
}
\makeatother

\usepackage[ruled, vlined]{algorithm2e}
\usepackage{pgfplots}
\pgfplotsset{compat=1.18} 
\usepgfplotslibrary{statistics}
\usepgfplotslibrary{colorbrewer} 

\newcommand{\indegree}[1]{\delta^-(#1)}
\newcommand{\outdegree}[1]{\delta^+(#1)}

\newcommand{\Mod}[1]{\ (\mathrm{mod}\ #1)}
\newcommand{\myrightarrow}[1]{\mathrel{\raisebox{-2pt}{$\xrightarrow{#1}$}}}

\begin{document}

\title{Dynamic quantum circuit compilation}

\author{Kun Fang}
\email{fangkun02@baidu.com}
\affiliation{Institute for Quantum Computing, Baidu Research, Beijing 100193, China}
\author{Munan Zhang}
\affiliation{Institute for Quantum Computing, Baidu Research, Beijing 100193, China}
\author{Ruqi Shi}
\affiliation{Institute for Quantum Computing, Baidu Research, Beijing 100193, China}
\author{Yinan Li}
\affiliation{School of Mathematics and Statistics and Hubei Key Laboratory of Computational Science, Wuhan University, Wuhan 430072, China}

\maketitle

\fontfamily{lmr}\selectfont

\begin{abstract}
Quantum computing has shown tremendous promise in addressing complex computational problems, yet its practical realization is hindered by the limited availability of qubits for computation. Recent advancements in quantum hardware have introduced mid-circuit measurements and resets, enabling the reuse of measured qubits and significantly reducing the qubit requirements for executing quantum algorithms. In this work, we present a systematic study of dynamic quantum circuit compilation, a process that transforms static quantum circuits into their dynamic equivalents with a reduced qubit count through qubit-reuse. We establish the first general framework for optimizing the dynamic circuit compilation via graph manipulation. In particular, we completely characterize the optimal quantum circuit compilation using binary integer programming, provide efficient algorithms for determining whether a given quantum circuit can be reduced to a smaller circuit and present heuristic algorithms for devising dynamic compilation schemes in general. Furthermore, we conduct a thorough analysis of quantum circuits with practical relevance, offering optimal compilations for well-known quantum algorithms in quantum computation, ansatz circuits utilized in quantum machine learning, and measurement-based quantum computation crucial for quantum networking. We also perform a comparative analysis against state-of-the-art approaches, demonstrating the superior performance of our methods in both structured and random quantum circuits. Our framework lays a rigorous foundation for comprehending dynamic quantum circuit compilation via qubit-reuse, bridging the gap between theoretical quantum algorithms and their physical implementation on quantum computers with limited resources.
\end{abstract}

\tableofcontents

\newpage
\section{Introduction}\label{sec:introduction}

Quantum computing has emerged as a promising avenue for solving intractable problems that are beyond the reach of classical computing, such as integer factoring~\cite{shor1994algorithms}, large database search~\cite{grover1996fast}, chemistry simulation~\cite{lanyon2010towards} and machine learning~\cite{biamonte2017quantum}. Nevertheless, existing quantum computers are limited by the number of qubits available for computation. To conclusively demonstrate the computational advantage of these devices compared to their classical counterparts in a variety of practical applications, we need to make efficient use of qubits when designing and executing quantum algorithms. 

A plethora of quantum algorithms have been traditionally formulated as \emph{static quantum circuits}, where the computation is performed on an initially prepared quantum state, and all measurements are applied at the end of the circuit to extract the computational results. However, the recent advancements in quantum hardware have paved the way for a more flexible approach, allowing for measurements and qubit resets to be executed in the midst of a quantum circuit. Moreover, these dynamic circuits permit the real-time evolution of the quantum circuit based on prior measurement outcomes~\cite{corcoles2021exploiting,pino2021demonstration,acharya2022suppressing}. This new paradigm of quantum computation, characterized by the ability to dynamically adjust the circuit, is referred to as a \emph{dynamic quantum circuit}, and plays a crucial role in the study of quantum error correction~\cite{fowler2012surface}, quantum communication~\cite{bennett1993teleporting}, and also measurement-based quantum computation~\cite{raussendorf2001one}.

\begin{figure}[!htb]
    \centering
    \begin{subfigure}[b]{.42\textwidth}
        \centering
        \includegraphics[width=0.7\textwidth]{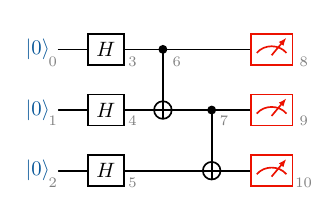}
        \caption{A static quantum circuit with $3$ qubits.}
        \label{fig:reducible example 1}
    \end{subfigure}
    \begin{subfigure}[b]{.55\textwidth}
        \centering
        \includegraphics[width=0.8\textwidth]{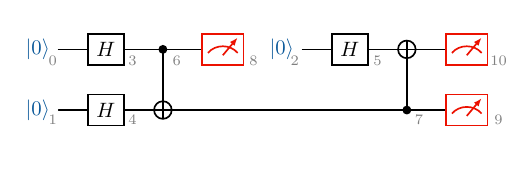}
        \caption{A dynamic quantum circuit with $2$ qubits.}
        \label{fig:reducible example 2}
    \end{subfigure}
    \caption{ An example of dynamic quantum circuit compilation from (a) to (b).}
    \label{fig:reducible example}
\end{figure}

This work investigates \emph{dynamic quantum circuit compilation}, which rewires static quantum circuits into equivalent dynamic circuits using fewer qubits through qubit-reuse strategies. An explicit example is illustrated in Figure~\ref{fig:reducible example}, where a 3-qubit static quantum circuit is compiled into a 2-qubit dynamic circuit by recycling the first qubit following its measurement instruction and subsequently reusing it for operations on the third qubit of the original static circuit. 
Dynamic quantum circuit compilation via qubit-reuse offers substantial advantages and addresses several critical facets of fault-tolerant quantum computation. Firstly, it efficiently reduces the number of qubits required for executing quantum algorithms, particularly valuable for implementing large-scale quantum circuits on near-term quantum computers and enhancing their classical simulation feasibility. Moreover, it effectively compacts the quantum circuit topology. This simplification minimizes the need for inserting swap gates and circumvents the allocation of faulty qubits when mapping logical circuits to specific quantum hardware devices, ultimately resulting in error reduction and fidelity improvement~\cite{Fei-CaQR, brandhofer2023optimal}. Furthermore, determining the minimum qubit requirement for executing specific quantum algorithms provides profound insights into the algorithms' complexity, aiding in the design of practical algorithms with quantum advantages.

The idea of compiling a quantum circuit via qubit-reuse was first introduced in~\cite{paler2016wire}, which employed the notion of wire recycling applied to predefined ancilla qubits. In~\cite{Fei-CaQR}, the authors developed a compiler-assisted tool and exploited the trade-off between qubit reuse, fidelity, gate count, and circuit duration. In~\cite{brandhofer2023optimal}, the authors introduced an SAT-based model for qubit-reuse optimization on near-term quantum devices. 
Previous studies have predominantly concentrated on superconducting quantum computers, known for their limited qubit connectivity and relatively short coherence time. But the situation is completely different for other quantum architectures, such as trapped-ion quantum computers. In these systems, qubits feature all-to-all connectivity and extremely long coherence time (on the order of tens of minutes~\cite{wang2021single}) but the scalability forms a primary bottleneck in their development. This motivates our emphasis on minimizing the number of qubits needed to execute a quantum circuit.

The work in~\cite{decross2022qubitreuse} studied the qubit-reuse compilation using the causal structure of the circuits and proposed a constraint programming optimization model capable of delivering exact solutions for small circuits, along with greedy algorithms designed to generate approximate solutions for large circuits. In this work, we conducted a thorough comparative analysis against their algorithms, demonstrating the superior performance of our methods in both structured and random quantum circuits. Notably, our framework successfully addresses an open challenge emphasized in their work, namely, the effective handling of quantum circuits with commutable structures and the ability to conduct compilation at the level of quantum algorithms, regardless of their specific quantum instruction sequences. This holds particular significance for practical applications in quantum machine learning~\cite{farhi2014quantum}, measurement-based quantum computation~\cite{briegel2009measurement}, and the pursuit of quantum supremacy~\cite{shepherd2009temporally}.

Quantum computers are expected to operate in a noisy environment with a limited number of qubits and without error correction in the near term. Therefore, circuit optimization is a crucial step to ensure the successful execution of quantum algorithms. This optimization can manifest in various forms, typically entailing modifications to the gate structure of a quantum circuit with the goal of enhancing result fidelity, such as reducing gate count by simplifying Clifford subcircuits~\cite{aaronson2004improved,bravyi2021clifford}, eliminating redundant gates~\cite{xu2022quartz}~\cite{qiskit,xu2023synthesizing}, and applying the Cartan decomposition to two-qubit gates~\cite{khaneja2000cartan}. Furthermore, circuit optimization may encompass the mapping of circuits to a specific quantum device architecture, aiming to minimize the number of inserted swap gates~\cite{qiskit,zulehner2018efficient} or the depth of the compiled circuit~\cite{Zhang-time-optimal} and schedule the braiding path between remote qubits~\cite{Fei-autobraid}. In contrast, dynamic circuit compilation is centered around the reordering of instructions and the reassignment of logical qubits, while preserving both the number and type of circuit instructions. This approach serves as a complementary strategy to other circuit optimization techniques and can be seamlessly integrated into the existing ones. For instance, it can be applied after the removal of redundant gates or before the mapping of the circuit to a specific quantum system. 

\paragraph{\textbf{Main contributions.}} In this work, we give a comprehensive investigation into dynamic quantum circuit compilation, a process that transforms a static quantum circuit into an equivalent dynamic circuit with a smaller size. In summary, we make the following contributions:
\begin{itemize}
    \item The first general framework for optimizing dynamic circuit compilation through graph manipulation and a rigorous mathematical model for optimal compilation in terms of qubit reuse. Our framework primarily targets qubit savings but is also adaptable to other scenarios, such as optimizing tradeoffs among circuit width, depth, and related factors.
    \item Efficient approaches to determine whether a given static quantum circuit can be reduced to a smaller circuit through qubit-reuse and heuristic algorithms to design dynamic circuit compilation schemes in general. Detailed time complexity analyses are also provided.
    \item Compiling quantum circuits with commutable structures, which is a significant feature in recent quantum applications. This successfully resolves an open challenge for circuit compilation in the prior study.
    \item Optimal compilations of quantum circuits with practical relevance, including well-known quantum algorithms in quantum computation, ansatz circuits applied in quantum machine learning, and measurement-based quantum computation crucial for quantum networking. These optimal compilations establish tight lower bounds and serve as benchmarks for other variants of qubit-reuse compilation methods.
    \item Numerical evaluations of our heuristic algorithms on both structured and random quantum circuits, highlighting the superior performance of our methods over existing approaches. In particular, experiments show that our approach outperforms in approximately $98.5\%$ of randomly generated quantum circuits and nearly $100\%$ of randomly generated instantaneous quantum polynomial (IQP) circuits. Noisy simulations also demonstrate a substantial enhancement in qubit reduction and the algorithm's performance for compiled circuits.
\end{itemize}

The rest of this work is structured as follows: Section~\ref{sec:preliminaries} and Section~\ref{sec: Quantum circuit and its representations} briefly introduce the notations and fundamental concepts employed throughout our study. Section~\ref{sec:Quantum circuit compilation via graph manipulation} establishes the core principles behind dynamic quantum circuit compilation through graph manipulation. Section~\ref{sec:determine reducibility} introduces various efficient approaches for determining the reducibility of a static quantum circuit. Section~\ref{sec:mathematical model} presents a rigorous mathematical model for optimal quantum circuit compilation. Section~\ref{sec:heuristic} introduces several heuristic algorithms for circuit compilation in general. Section~\ref{sec:examples} presents the analytical and numerical evaluations of the proposed methods and compares them with existing approaches. Section~\ref{sec:related work} reviews some related works. Finally, Section~\ref{sec:discussion} concludes our study and discusses related open problems. All proofs in the main text are delegated to the supplementary material. All algorithms in this work have been implemented in QNET, a quantum network toolkit developed by the Institute for Quantum Computing at Baidu. All the code referenced in this paper can be accessed online on GitHub~\footnote{\href{https://github.com/baidu/QCompute/tree/master/Extensions/QuantumNetwork}{https://github.com/baidu/QCompute/tree/master/Extensions/QuantumNetwork}}.

\section{Preliminaries}\label{sec:preliminaries}

\paragraph{Notation} Throughout this work, we denote $I_n$ as the identity matrix of size $n \times n$, $J_n$ as the all-one matrix of size $n \times n$, and $O_n$ as the zero matrix of size $n \times n$. We also use $E^k_{i,j}$ to represent a $k\times k$ matrix where the $(i,j)$-th entry is one, and zero otherwise. We use $n$ to represent the width of a quantum circuit, while $m$ signifies the number of quantum instructions. We use \textit{list}[$i$] to represent the $i$-th element of a list. Furthermore, we adopt the convention that all qubit registers in quantum circuits, along with the row and column indices of matrices, and the indices of lists in algorithms, all start from zero.

\paragraph{Boolean matrix} A Boolean matrix is a matrix whose entries are either $0$ or $1$. Let $A$ and $B$ be two Boolean matrices of size $k\times k$. Then their Boolean product is defined as $(A \odot B)_{ij} := \bigvee_{l=0}^{k-1} (A_{il} \wedge B_{lj})$ 
where $\vee$ is the logical OR and $\wedge$ is the logical AND. 

\paragraph{Directed Graph} 
A graph is an ordered pair $G = (V, E)$ comprising vertices $V$ and edges $E$. A graph is directed if its edges are ordered pairs of vertices. For any edge $(u,v)$ in a directed graph, $u$ is called the tail and $v$ is called the head of the edge. For a vertex $v$, the number of head ends adjacent to a vertex is called the indegree of the vertex which is denoted as $\indegree{v}$ and the number of tail ends adjacent to a vertex is its outdegree, which is denoted as $\outdegree{v}$. A vertex with zero indegree is called a root and a vertex with zero outdegree is called a terminal. A vertex $v$ is \emph{reachable} from another vertex $u$ if there exists a direct path from $u$ to $v$ in the graph. A \emph{directed acyclic graph} (DAG) is a directed graph with no directed cycles, which has been widely used in formulating the classical and quantum circuit model. A \emph{topological ordering} of a DAG is an ordering of its vertices into a sequence, such that for every edge the tail vertex of the edge occurs earlier in the sequence than the head vertex of the edge. This ordering can be efficiently obtained from its DAG but the result is not unique. A \emph{bipartite graph} is a graph whose vertices can be divided into two disjoint and independent sets $U$ and $V$ such that every edge connects a vertex in set $U$ to a vertex in set $V$. We write $G = (U, V, E)$ to denote a bipartite graph with parts $U$, $V$ and edges $E$. A bipartite graph is complete if every vertex in $U$ is connected to every vertex in $V$.

\paragraph{Matrix Representation of Graph}
Let $G = (V, E)$ be a directed graph where $V = \{v_0,\cdots, v_{k-1}\}$. Its \emph{adjacency matrix} is a Boolean matrix, denoted as $A(G)$, of size $k\times k$ whose $(i,j)$-th entry is one if the directed edge $(v_i, v_j) \in E$ and zero otherwise. If a bipartite graph $G = (U, V, E)$ with $U = \{u_0,\cdots,u_{k-1}\}$, $V = \{v_{0},\cdots, v_{k-1}\}$ and all edges pointing from $U$ to $V$, then its adjacency matrix can be written as $A(G) = 
    \bigl(\begin{smallmatrix}
    O_k & X\\
    O_k & O_k
    \end{smallmatrix}\bigr)$
where $X$ is a $k \times k$ Boolean matrix in which $X_{ij} = 1$ if $(u_i, v_j) \in E$. We call this submatrix the \emph{biadjacency matrix} of the bipartite graph and denote it as $B(G)$. The biadjacency matrix of a complete bipartite graph is an all-one matrix.
A matrix $A$ of size $k\times k$ is \emph{nilpotent} if there exists an integer $l$ with $1 \leq l \leq k$ such that $A^l = O_k$. A directed graph is acyclic if and only if its adjacency matrix is nilpotent (see~\cite[Theorem 9-17]{deo2017graph} or~\cite{li2022connections}).

\section{Quantum circuit and its representations}
\label{sec: Quantum circuit and its representations}

\subsection{Quantum computing basics}

\paragraph{Quantum states}
In quantum computing, a quantum bit or qubit is the fundamental unit of quantum information. A qubit can be in two \emph{computational basis states}, which are represented by the 2-dimensional vectors $\ket{0} = \bigl(\begin{smallmatrix}
1 \\
0 
\end{smallmatrix}\bigr) $ and $\ket{1} = \bigl(\begin{smallmatrix}
0 \\
1 
\end{smallmatrix}\bigr) $. Unlike classical bit, a qubit can be in a linear combination (\emph{superposition}) of the basis states, $\alpha \ket{0} + \beta \ket{1} = \bigl(\begin{smallmatrix}
\alpha \\
\beta 
\end{smallmatrix}\bigr)$, where $\alpha, \beta$ are complex numbers such that $|\alpha|^2 + |\beta|^2 = 1$. 

\paragraph{Quantum operations}
Quantum operations manipulate the state of qubits and encompass various actions such as quantum gates, quantum measurements and reset operations. Quantum gates are unitary operators that transform the state of the qubits. For example, the Hadamard gate $H$ can be used to create equal superposition: $H \ket{0} = (\ket{0} + \ket{1}) / \sqrt{2}$, while the controlled-NOT (CNOT) gate acts on $2$ qubits and maps the state $\ket{a, b}$ to $\ket{a, a \oplus b}$, where $\oplus$ denotes XOR operation. Quantum measurements extract classical information from quantum states, resulting in the collapse of the quantum state. For example, measuring a single qubit state $\alpha \ket{0} + \beta \ket{1}$ yields $\ket{0}$ with probability $|\alpha|^2$ or $\ket{1}$ with probability $|\beta|^2$. After a qubit is measured, it can be reset to a known state (typically the state $\ket{0}$) and to be used for subsequent computation.

\paragraph{Quantum circuits}
A quantum circuit is a mathematical and visual model used in quantum computing to represent and perform quantum computations. It is analogous to classical digital circuits used in classical computing. A quantum circuit consists of a series of quantum gates, each of which operates on one or more qubits. In this work, we focus on single-qubit and two-qubit gates without loss of generality and most of the results apply to multi-qubit gates with slight modification. A quantum circuit is \emph{static} if it does not contain mid-circuit reset operations and all measurements, are performed at the end of the circuit (e.g., Figure~\ref{fig:reducible example 1}). A quantum circuit is \emph{dynamic} if it contains mid-circuit measurements and reset operations (e.g., Figure~\ref{fig:reducible example 2}). A quantum circuit is \emph{reducible} if it can be written as an equivalent quantum circuit with fewer qubits. Otherwise, it is called \emph{irreducible}. A quantum circuit and its reduced circuit are considered equivalent if their measurement outcomes yield identical distributions. In the following discussion, a `quantum circuit' specifically refers to a circuit that incorporates a predetermined sequence of operations. However, for scenarios involving commutable gates or blocks, we utilize the term `quantum circuit with commutable structure'.

\subsection{Quantum circuit instructions}

Quantum circuits are conventionally depicted through circuit diagrams. Nevertheless, when considering circuit compilation, a more precise and streamlined approach involves employing quantum circuit instructions. This allows a quantum circuit to be presented as a list of instructions, structured according to the chronological order of their execution. Each entry within the instruction list corresponds to a distinct quantum operation within the circuit and is denoted by a quartet [ID, TYPE, QUBIT, PAR]. This quartet serves to uniquely identify the operation, specify its operation type, indicate the qubit(s) it acts upon, and provide any relevant parameters such as the rotation angle or the group tag in case where gates are commutable (set to `NONE' in cases where no parameters are applicable)~\cite{fang2023network}. For example, the instruction [$5$, RY, $1$, $\pi$] signifies the application of an $R_y$ rotation gate to qubit $q_1$, with the rotation angle set to $\pi$, and the corresponding instruction ID is $5$. Similarly, the instruction [$4$, CX, [$0$, $2$], NONE] represents the implementation of a controlled-NOT gate, with $q_0$ as the control qubit and $q_2$ as the target qubit, with an instruction ID of $4$. Likewise, the instruction [$7$, MEASURE, $2$, NONE] denotes the measurement of qubit $q_2$ in the computational basis, with an instruction ID of $7$. Finally, the instruction [$2$, RESET, 3, NONE] indicates the reset of qubit $q_3$ to the ground state, typically represented as the zero state $\ket{0}$, and carries an instruction ID of $2$. An example of quantum circuit instructions is given in Figure~\ref{fig:quantum circuit composition}(d).

\subsection{Graph representation of quantum circuit}

Given that a quantum circuit essentially constitutes a time-ordered sequence of quantum instructions, it is natural to employ a directed graph to represent the causal relationships among them. In this context, a quantum circuit finds its representation through a directed graph, effectively preserving the execution constraints of all quantum instructions. Specifically, each quantum instruction corresponds to a vertex within the graph. A directed edge starting from vertex $u$ and terminating at vertex $v$ signifies that the quantum instruction associated with vertex $u$ must be executed before the quantum instruction associated with vertex $v$. This directed graph is inherently acyclic, as the presence of any cycle would imply a dependency of past instructions on future ones, contravening the causal relationship inherent in the quantum circuit. Given this characteristic, we refer to this directed graph as the \emph{DAG representation} of the quantum circuit (see e.g. Figure~\ref{fig:quantum circuit compilation}).

In many instances of interest, the original circuit is not unique and can be represented with different ordering of gates. A typical example is given by the instantaneous quantum polynomial (IQP) circuits~\cite{shepherd2009temporally} in the form of $H^{\ox n} D H^{\ox n}$, where $H$ denotes the hadamard gate and $D$ constitutes a block of gates diagonal in the computational basis (e.g. constructed by randomly selecting gates from the set $\{\sqrt{CZ}, T\}$) and consequently can be applied in any temporal order. Other examples can be given by the quantum approximate optimization algorithm (QAOA) used in quantum machine learning~\cite{farhi2014quantum} and the preparation of graph states in the measurement-based quantum computation~\cite{briegel2009measurement}. These circuits do not have pre-determined structures, and imposing any dependencies among commuting gates may limit the opportunities for qubit-reuse. To handle circuits with commutable structure, we can avoid imposing dependencies between commutable gates and only establish edges between operations with pre-defined ordering when generating the DAG representation of the circuit. An illustrative example of circuits with commutable structure and the algorithm for generating the DAG representation is detailed in Appendix~\ref{app: further algorithms} of the supplemental material.

In this work, we will employ the DAG representation to investigate the dynamic quantum circuit compilation problem. To facilitate our analysis, all vertices within the DAG representation can be categorized into three distinct groups:

\begin{enumerate}
    \item Root vertices: vertices with zero indegree, which correspond to the first layer of quantum operations on each qubit (typically reset operations);
    \item Terminal vertices: vertices with zero outdegree, which correspond to the last layer of quantum operations on each qubit (typically quantum measurements);
    \item Internal vertices: vertices with nonzero indegree and outdegree, which correspond to the intermediate quantum operations (typically quantum gates).
\end{enumerate}

\begin{remark}
We assume that all static quantum circuits in this work start from reset operations and end with measurements. This implies that the circuit width is equal to the number of root vertices within the DAG representation.
\end{remark}

\BlankLine

Concerning our compilation problem, all essential information resides within the reachability from roots to terminals within the DAG representation, while the internal vertices facilitate and transmit this reachability. Consequently, we can further streamline the DAG representation, concentrating our focus on the \emph{simplified DAG representation} of the circuit. This simplified representation takes the form of a bipartite graph $(R, T, E)$, with $R$ and $T$ representing the sets of roots and terminals from the DAG representation. An edge $(r, t)$ in $E$ connects a root $r \in R$ to a terminal $t \in T$ if a directed path exists from $r$ to $t$ within the DAG representation. Figure~\ref{fig:quantum circuit compilation} showcases an example of a quantum circuit alongside its corresponding DAG and simplified DAG representations.

\subsection{Quantum circuit composition and subcircuit}

Quantum circuit composition stands as a fundamental technique for integrating modular components into complex quantum algorithms. An illustrative instance arises in the domain of quantum machine learning, exemplified by the Variational Quantum Eigensolver (VQE)~\cite{peruzzo2014variational} where specific quantum circuit patterns are joined sequentially. More specifically, given two quantum circuits of the same size, we define their \emph{circuit composition} as the sequential integration of quantum gate operations from a second circuit onto those of the first, while preserving the initialization in the first circuit and the measurement in the second. An illustration of quantum circuit composition and the sequence of circuit instructions is presented in Figure~\ref{fig:quantum circuit composition}. Conversely, we can also consider \emph{subcircuits} which are obtained by removing part of circuit instructions from the original circuit. This approach offers an alternative perspective for gaining insight into the essential structure of the circuit.

\begin{figure}[H]
    \centering
    \includegraphics[width=0.8\textwidth]{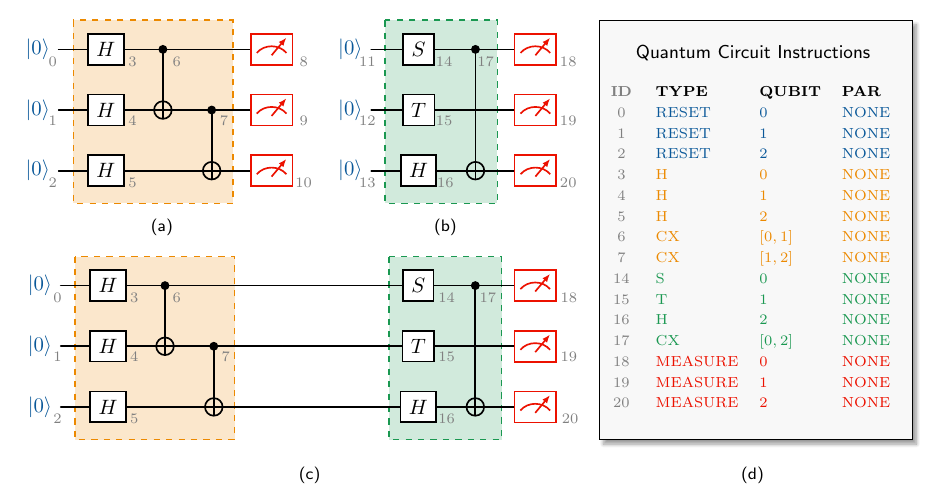}
    \vspace{-4mm}
    \caption{(Color online) An illustration of quantum circuit composition. (a), (b) quantum circuits; (c) composite quantum circuit; (d) the instructions of the composite quantum circuit.}
    \label{fig:quantum circuit composition}
\end{figure}

\section{Quantum circuit compilation via graph manipulation}
\label{sec:Quantum circuit compilation via graph manipulation}

In this section, we present the mathematical formulation of the quantum circuit compilation problem using its graph representation, which serves as a pivotal foundation for subsequent in-depth investigations. Drawing inspiration from the example illustrated in Figure~\ref{fig:reducible example}, dynamic quantum circuit compilation via qubit-reuse involves resetting a qubit after measurement, thereby deferring certain reset operations until after the measurement process. In the DAG representation, this corresponds to the addition of a directed edge from a terminal to a root, signifying that the corresponding reset operation occurs subsequent to the execution of the measurement.

For instance, Figure~\ref{fig:quantum circuit compilation 2} provides a DAG representation of the quantum circuit depicted in Figure~\ref{fig:reducible example 1} or~\ref{fig:quantum circuit compilation 1}. The qubit-reuse in Figure~\ref{fig:reducible example 2} is depicted by the addition of a new edge (marked as a dashed green line) from terminal $8$ to root $2$. With the inclusion of this new edge, the resulting DAG exactly mirrors the DAG representation of the dynamic quantum circuit depicted in Figure~\ref{fig:reducible example 2}. This new edge can also be integrated into the simplified DAG representation as shown in Figure~\ref{fig:quantum circuit compilation 3}.

\begin{figure}[!htb]
    \centering
    \begin{subfigure}[b]{.34\textwidth}
        \centering
        \includegraphics[width=0.8\textwidth]{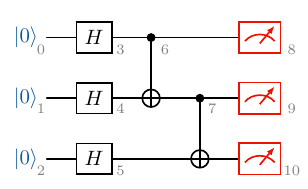}
        \caption{quantum circuit diagram}
        \label{fig:quantum circuit compilation 1}
    \end{subfigure}
    \hspace{0em}
    \begin{subfigure}[b]{.32\textwidth}
        \centering
        \includegraphics[width=0.8\textwidth]{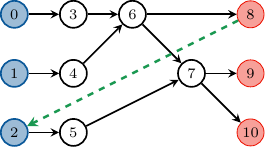}
        \caption{DAG with new edge}
        \label{fig:quantum circuit compilation 2}
    \end{subfigure}
    \hspace{-1em}
    \begin{subfigure}[b]{.3\textwidth}
        \centering
        \includegraphics[width=0.48\textwidth]{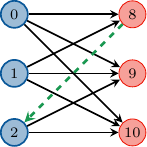}
        \caption{simplified DAG with new edge}
        \label{fig:quantum circuit compilation 3}
    \end{subfigure}
    \caption{(Color online) An illustration of quantum circuit compilation via graph manipulation. Root vertices are marked in blue, and terminal vertices in red. The newly added edge is marked in green.}
    \label{fig:quantum circuit compilation}
\end{figure}

The above idea of dynamic quantum circuit compilation is formally formulated as follows.

\begin{theorem}\label{thm: circuit compilation via graph manipulation}
Let $G = (R,T,E)$ be the simplified DAG of a static quantum circuit. Then compiling the quantum circuit via qubit-reuse is equivalent to adding edges to $G$ (Let $E'$ be the set of added edges) such that
\begin{itemize} 
    \item  $\forall\; (t,r) \in E'$, it has $t \in T$ and $r \in R$;
    \item $\forall\; (t,r) \in E'$, it has $\outdegree{t} = 1$ and $\indegree{r} = 1$;
    \item $G' = (R\cup T, E \cup E')$ is an acyclic graph.
\end{itemize}
The new graph $G'$ will be called the modified DAG in the subsequent discussion.
\end{theorem}

The first condition indicates the reuse of a qubit after another qubit has been measured. The second condition specifies that one qubit is reassigned to a single new qubit, and a reused qubit accommodates only one new qubit. The last condition ensures the equivalence of the compiled circuit. Conversely, any graph manipulation adhering to these conditions corresponds to a dynamic circuit compilation scheme. This is presented in Algorithm~\ref{alg:modified DAG to dynamic circuit} in the supplementary material.

\begin{remark}
\label{rem:ignore single-qubit gates}
Theorem~\ref{thm: circuit compilation via graph manipulation} shows that dynamic quantum circuit compilation only depends on the simplified DAG representation, which is independent of the single-qubit gates in the static quantum circuit to compile. So we will omit all single-qubit gates in the subsequent discussion.
\end{remark}

\section{Determine the reducibility of quantum circuit}
\label{sec:determine reducibility}

In this section, we present three distinct approaches for assessing the reducibility of a static quantum circuit. The first approach is rooted in the DAG representation of the quantum circuit, while the second approach involves a direct and comprehensive analysis of the circuit's structure, guided by a set of critical observations. Departing from these geometric perspectives, the third approach utilizes Boolean matrices to check reducibility from the subcircuits.

\subsection{Approach 1: determine the reducibility from graph}
\label{subsec:reducibility from graph}

It is clear from Theorem~\ref{thm: circuit compilation via graph manipulation} that a quantum circuit is reducible if and only if we can add at least one more edge to its DAG while adhering to the three specified conditions. This idea can be elaborated upon as follows:

\begin{proposition}\label{prop:reducibility from dag}
A static quantum circuit is irreducible if and only if its simplified DAG is a complete bipartite graph.
\end{proposition}

This result shows that the reducibility of a static quantum circuit can be completely determined by its simplified DAG. In particular, we only need to check if the biadjacency matrix of the simplified DAG is an all-one matrix. For this, we can exploit Depth-First Search (DFS) algorithm to identify all paths from roots to terminals for a given DAG. The detailed algorithm of this approach along with its time complexity analysis is given in Algorithm~\ref{alg:reachability through dfs} and Proposition~\ref{prop:reducibility through DFS}.


\begin{algorithm}[!htb]
\small
\caption{Determining reducibility from graph}\label{alg:reachability through dfs}
\LinesNumbered

\KwIn{\\  \begin{tabular}{p{2.5cm}l}
     \textit{StaticCircuit} &  the instruction list of a static quantum circuit
\end{tabular}}

\BlankLine

\KwOut{\\  \begin{tabular}{p{2.5cm}l}
     \textit{True} or \textit{False} & whether the static quantum circuit is reducible
\end{tabular}}

\BlankLine
\BlankLine

Run Algorithm~\ref{alg:Converting static quantum circuit to DAG} in Appendix~\ref{app: further algorithms} to obtain the DAG of \textit{StaticCircuit} and record it as \textit{Digraph}\;

Let \textit{Roots} and \textit{Terminals} be the set of roots and terminals in \textit{Digraph}, respectively\;

Let $n$ be the length of \textit{Roots} and \textit{Terminals}\;

Initialize a $n \times n$ zero matrix $B$\;

\BlankLine

\For{$i=0$ \KwTo $n-1$}{
    \For{$j=0$ \KwTo $n-1$}{
        Apply DFS algorithm to search if there is a path from \textit{Roots}[$i$] to \textit{Terminals}[$j$]\;
        \If{such a path exists}{
            Set the $(i, j)$ entry of matrix $B$ to $1$\;
        }
    }
}

\BlankLine

\uIf{$B$ is all-one matrix}{
\Return \textit{False}\;
}
\lOther{\Return \textit{True}}

\end{algorithm}


\begin{proposition}\label{prop:reducibility through DFS}
For a given static quantum circuit with $n$ qubits and $m$ operations, its reducibility can be efficiently determined by Algorithm~\ref{alg:reachability through dfs} with a worst-case time complexity of $O(m n^2)$.
\end{proposition}

\begin{remark}
Note that this approach can be used to determine the reducibility of quantum circuits with commutable structures.
\end{remark}

\subsection{Approach 2: determine the reducibility from qubit reachability}
\label{subsec:reducibility from circuit sturcture}

It is worth noting that within the DAG representation of a quantum circuit, the connections between roots and terminals on distinct qubits are effectively established by vertices that correspond to two-qubit gates in the circuit. With this in mind, we introduce the concept of qubit reachability within a quantum circuit and present our second approach to determining its reducibility.

\begin{definition}[Reachability between qubits]
\label{def:reachability of qubits}
Given an instruction list of an $n$-qubit quantum circuit acting on the set of qubits $Q = \{q_0,q_1,\cdots,q_{n-1}\}$. Then any two-qubit instruction introduces two reachability relations $q_i \myrightarrow{k} q_j $ and $q_j \myrightarrow{k} q_i$
where $q_i$ and $q_j$ are the qubits upon which the instruction operates and $k$ is the order index of the instruction within the instruction list.
A qubit $q_i$ reaches $q_j$ (or, equivalently, $q_j$ is reachable from $q_i$), denoted as $q_i \to q_j$, if there exists a sequence of relations $q_i \myrightarrow{k_1} q_{l_1},\; q_{l_1} \myrightarrow{k_2} q_{l_2},\; \cdots,\; q_{l_s} \myrightarrow{k_{s+1}} q_{j}$ 
such that $k_1 \leq k_2\leq \cdots \leq k_{s+1}$. Moreover, qubits $q_i$ and $q_j$ are mutually reachable if $q_i$ reaches $q_j$ and vice versa.
\end{definition}

To illustrate this definition, consider the quantum circuit in Figure~\ref{fig:quantum circuit compilation}(a). Here the order indices of each double-qubit instructions are the same as their IDs. For the first CNOT instruction acting on $q_0$ and $q_1$, it introduces two qubit relations $q_0 \myrightarrow{6} q_1$ and $q_1 \myrightarrow{6} q_0$. So we have $q_0 \to q_1$ and $q_1\to q_0$, that is, they are mutually reachable. We can also see that $q_0 \to q_2$ because we have relation $q_0 \myrightarrow{6} q_1$ and $q_1 \myrightarrow{7} q_2$. But the reverse direction $q_2 \to q_0$ does not hold.

The following result demonstrates that qubit reachability constitutes a necessary and sufficient condition for determining circuit reducibility.

\begin{proposition}\label{prop:reducibility from qubit reachability}
A static quantum circuit is irreducible if and only if any two qubits of this quantum circuit are mutually reachable.
\end{proposition}

Compared to the approach in Section~\ref{subsec:reducibility from graph}, the qubit reachability approach does not require the explicit construction of the DAG or the use of the DFS algorithm to derive the simplified DAG. Instead, we can establish qubit reachability by traversing the circuit instructions only once, progressively building up the reachability through a transitive rule. 

For convenience, define the reachable set of $q_i$ by $Q_i := \{q_j \in Q: q_j \to q_i\} \cup \{q_i\}$.
This set is updated with each double-qubit instruction. Taking Figure~\ref{fig:quantum circuit compilation}(a) as an example, before the second CNOT instruction on $q_1$ and $q_2$, the reachable sets are $Q_1 = \{q_0, q_1\}$ and $Q_2 = \{q_2\}$. Subsequently, after this instruction, any qubit in the set $Q_1 \cup Q_2 = \{q_0, q_1, q_2\}$ can reach both $q_1$ and $q_2$. So the reachable sets are updated to $Q_1 = \{q_0, q_1, q_2\}$, and $Q_2 = \{q_0, q_1, q_2\}$.

By this transitive rule, we can effectively determine circuit reducibility. The algorithm for this procedure is outlined in Algorithm~\ref{alg:reducibility qubit reachability}, followed by its time complexity in Proposition~\ref{prop:reducibility through structure}.


\begin{algorithm}[!htb]
\small
\caption{Determining reducibility from qubit reachability}\label{alg:reducibility qubit reachability}
\LinesNumbered
\KwIn{\\  \begin{tabular}{p{2.5cm}l}
     \textit{StaticCircuit} &  the instruction list of a static quantum circuit
\end{tabular}}

\BlankLine

\KwOut{\\  \begin{tabular}{p{2.5cm}l}
     \textit{True} or \textit{False} & whether the static quantum circuit is reducible
\end{tabular}}

\BlankLine
\BlankLine

Let $n$ be the quantum circuit width\;

Initialize a list \textit{ReachableSets} of length $n$, with each element initialized as an empty set\;

\BlankLine

\lFor{$i=0$ \KwTo $n-1$}{add $q_i$ to \textit{ReachableSets}[$i$]}

\ForEach{Instruction {\rm \textbf{in}} StaticCircuit}{
\If{Instruction is a double qubit gate}{
Record the QUBIT values of \textit{Instruction} as $i$ and $j$ respectively\;
Calculate \textit{Union} $=$ \textit{ReachableSets}[$i$] $\cup$ \textit{ReachableSets}[$j$]\;
Set \textit{ReachableSets}[$i$] and \textit{ReachableSets}[$j$] to \textit{Union}\;

}
}

\BlankLine

\ForEach{ReachableSet {\rm \textbf{in}} ReachableSets}{
\If{the length of ReachableSet is less than $n$}{
\Return \textit{Ture}}
}
\lOther{\Return \textit{False}}
\end{algorithm}

\begin{proposition}\label{prop:reducibility through structure}
For a given static quantum circuit with $n$ qubits and $m$ operations, its reducibility can be efficiently determined by Algorithm~\ref{alg:reducibility qubit reachability} with a worst-case time complexity of $O(m n)$.
\end{proposition}

\subsection{Approach 3: determine the reducibility from matrix}
\label{subsec:reducibility Boolean matrix}

The preceding approach by qubit reachability can be further formulated via Boolean matrix manipulation, providing greater convenience for analytical studies, particularly for demonstrating the optimal compilations detailed in Section~\ref{sec:examples}.

Let $\cC_1 \circ \cC_2$ be a composition of quantum circuits $\cC_1$ and $\cC_2$ and let $B(\cC)$ be the biadjacency matrix of the simplified DAG of the quantum circuit $\cC$. Then we have the following relation for quantum circuit composition.

\begin{proposition}\label{prop: biadjacency matrix relation of composite circuit}
Let $\cC_1$ and $\cC_2$ be two static quantum circuits. Then $B(\cC_1 \circ \cC_2) = B(\cC_1) \odot B(\cC_2)$.
\end{proposition}

Note that if an $n$-qubit quantum circuit only contains one two-qubit gate, e.g., a CNOT gate acting on the $i$-th and $j$-th qubit, its biadjacency matrix is given by $B_{ij} = I_n + E^n_{ij} + E^n_{ji}$
where $E^n_{ij}$ is the matrix whose $(i,j)$ entry is one and zero otherwise. Since any quantum circuit can be seen as the composition of subcircuits containing only a single two-qubit gate, as per  Proposition~\ref{prop: biadjacency matrix relation of composite circuit}, the computation of a quantum circuit's biadjacency matrix entails the product of a sequence of matrices in the form of $B_{ij}$.

A more in-depth examination of the Boolean product reveals that there is no need to perform explicit matrix multiplication. That is, for any matrix $A$, we have $A \odot B_{ij} = (a_1, \cdots, a_i \lor a_j, \cdots, a_j \lor a_i, \cdots, a_n)$
where $a_i$ is the $i$-th column of the matrix $A$ and $a_i \lor a_j$ represents entrywise OR of $a_i$ and $a_j$.
In other words, the impact of multiplying a matrix $B_{ij}$ is equivalent to replacing the $i$-th and $j$-th columns of $A$ with $a_i \lor a_j$. This gives the following algorithm for checking the reducibility.

\begin{algorithm}[!htb]
\small
\caption{Determining reducibility from matrix}\label{alg:reducibility through Boolean matrix}
\LinesNumbered
\KwIn{\\  \begin{tabular}{p{2.5cm}l}
     \textit{StaticCircuit} &  the instruction list of a static quantum circuit
\end{tabular}}

\BlankLine

\KwOut{\\  \begin{tabular}{p{2.5cm}l}
     \textit{True} or \textit{False} & whether the static quantum circuit is reducible
\end{tabular}}

\BlankLine
\BlankLine

Let $n$ be the quantum circuit width\;

Initialize $B$ as an $n\times n$ identity matrix.\;

\BlankLine

\ForEach{Instruction {\rm \textbf{in}} StaticCircuit}{
    \If{Instruction is a two-qubit gate}{
        Record the QUBIT values of \textit{Instruction} as $i$ and $j$ respectively\;
        Calculate $B' = B[i] \lor B[j]$, where $B[k]$ is the $k$-th column of matrix $B$\;
        Set $B[i] = B'$ and $B[j] = B'$\;
    }
}

\uIf{$B$ is all-one matrix}{
\Return \textit{False}\;
}
\lOther{\Return \textit{True}}
\end{algorithm}


\begin{proposition}\label{prop:reducibility through matrix}
For a given static quantum circuit with $n$ qubits and $m$ operations, its reducibility can be efficiently determined by Algorithm~\ref{alg:reducibility through Boolean matrix} with a worst-case time complexity of $O(m n)$.
\end{proposition}

\section{Optimal quantum circuit compilation}
\label{sec:mathematical model}

In this section, we introduce the mathematical model that characterizes optimal quantum circuit compilation. This model serves as the foundation for deriving heuristic algorithms and analyzing optimal compilation schemes for specific quantum circuits in the subsequent sections.

\begin{proposition}\label{prop:compilation via BIP}
Finding the optimal dynamic circuit compilation scheme of a static quantum circuit with $n$ qubits is equivalent to solving the following binary integer programming problem with $n^2$ Boolean variables $F_{ij}$: 
\begin{subequations}
\allowdisplaybreaks
\begin{align}
\alpha := \max \;& \sum_{i,j=0}^{n-1} F_{ij} \\
\text{\rm s.t.}\; 
& F_{ij}  \leq \bar{B}_{ij},\, \forall i,j \in \{0, 1, ..., n-1\} \label{eq: submatrix condition}\\
& \sum_{j=0}^{n-1} F_{ij} \le 1, \, \forall i \in \{0, 1, ..., n-1\}  \label{eq: row sum condition} \\
& \sum_{i=0}^{n-1} F_{ij} \le 1, \, \forall j \in \{0, 1, ..., n-1\}  \label{eq: col sum condition} \\
& \begin{pmatrix}
 O_n & B \\
 F & O_n
\end{pmatrix} \text{ nilpotent },\label{eq: nilpotent condition}
\end{align}
\label{eq: exact optimization}
\end{subequations}
where $B$ is the biadjacency matrix of the simplified DAG and $\bar{B} = ^\neg (B^{\top})$, referred to as the candidate matrix, is the logical NOT of the transpose of matrix $B$.  Furthermore, the width of the compiled quantum circuit is given by $n - \alpha$.
\end{proposition}

This result can be seen as a translation of Theorem~\ref{thm: circuit compilation via graph manipulation} from graph manipulation to matrix optimization. Particularly, the variable $F$ here represents a feasible solution in the graph manipulation. The block matrix in Eq.~\eqref{eq: nilpotent condition} is the adjacency matrix of the modified DAG, which can be used to compile the quantum circuit.

\begin{remark}
It is worth noting that condition~\eqref{eq: submatrix condition} is indeed implied by condition~\eqref{eq: nilpotent condition}. However, we explicitly enforce this condition as it proves helpful in the analysis of the optimal solution for quantum circuits in Section~\ref{sec:examples}. Furthermore, there are other equivalent variants of the optimization problem. For instance, the second condition can be omitted, allowing different terminals to connect to the same root. Nevertheless, this modification does not contribute to a reduction in the number of qubits. In such cases, the objective function should be adjusted accordingly.
\end{remark}

\begin{remark}
The difficulty in solving the binary integer programming arises from the presence of the nilpotent constraint. Due to the non-convex nature of the set of nilpotent matrices, the optimization problem in Proposition~\ref{prop:compilation via BIP} is inherently non-convex.
\end{remark}

\section{Heuristic algorithms}\label{sec:heuristic}

The previous section demonstrated that the quantum circuit compilation problem is essentially a binary integer optimization problem with an exponentially increased solution space. While checking the reducibility of a quantum circuit is a polynomial-time task, as analyzed in Section \ref{sec:determine reducibility}, finding the optimal compilation scheme could potentially require exponential time. 

In this section, we introduce several efficient heuristic algorithms designed to address the optimization problem within polynomial time. All algorithms are presented to compile a circuit into the smallest size. However, users retain the flexibility to halt the main loop based on specific criteria, enabling the investigation of tradeoffs among circuit width, depth, and related factors.

\subsection{Algorithm 1: Minimum Remaining Values Heuristic} 

The Minimum Remaining Values (MRV) heuristic is a commonly employed technique in constraint satisfaction problems~\cite{stuart2009artificial}, which designates the variable with the fewest valid values (i.e., the minimum remaining values) as the next one for value assignment. This approach finds extensive application in problems such as Sudoku solvers and map coloring.

\begin{figure}[!htb]
    \centering
    \includegraphics[width=0.8\textwidth]{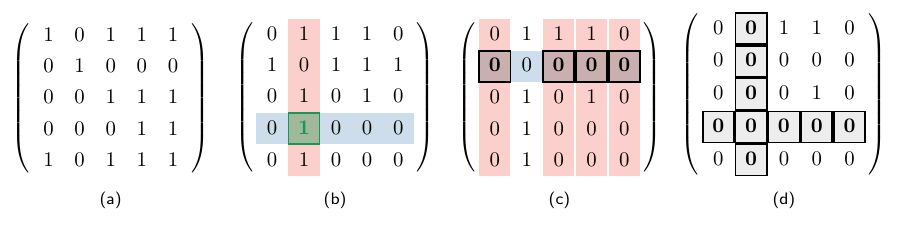}
    \vspace{-3mm}
    \caption{(Color online) An illustration of MRV heuristic algorithm. Suppose the row and column of the matrix are indexed by $\{t_0, \cdots, t_4\}$ and $\{r_0, \cdots, r_4\}$. (a) the biadjacency matrix $B$ of the simplified DAG of a $5$-qubit Bernstein–Vazirani algorithm; (b) the candidate matrix $\bar B = ^\neg (B^{\top})$ whereas the candidate edge is chosen as $(t_3, r_1)$ marked in green; (c) set all entries in $T_3 \times R_1$ to zero , where $T_3 = \{t_1\}$ and $R_1 = \{r_0, r_2, r_3, r_4\}$; (d) set all entries in the $t_3$ row  and all entries in the $r_1$ column to zero.}
    \label{fig:mrv example}
\end{figure}

In the context of our dynamic circuit compilation problem, we can implement the MRV heuristic algorithm as follows. Consider a $5$-qubit Bernstein–Vazirani algorithm, with its biadjacency matrix $B$ provided in Figure~\ref{fig:mrv example}(a), and the candidate matrix $\bar B = ^\neg (B^{\top})$ given in Figure~\ref{fig:mrv example}(b). Let's assume that the rows and columns are indexed by $\{t_0, \cdots, t_4\}$ and $\{r_0, \cdots, r_4\}$. During each iteration, we first identify the terminal $t_i$ with the fewest candidate roots according to the candidate matrix and connect it to a candidate root $r_j$ with the least number of choices of terminals. In this example, the candidate edge is selected as $(t_3, r_1)$, as depicted in green. After adding this edge, we need to update the candidate matrix. Prior to the incorporation of this edge, the set $R_i$ encompasses all roots capable of reaching terminal $t_i$, while the set $T_j$ represents all terminals that are reachable from root $r_j$. In this example, we have $R_1 = \{r_0, r_2, r_3, r_4\}$ and $T_3 = \{t_1\}$. Upon adding this edge, any root $r$ in the set $R_i$ will gain the ability to reach any terminal $t$ in the set $T_j$, which indicates that all edges $(t, r)$ where $t \in T_j$ and $r \in R_i$ are no longer candidate edges. Consequently, we need to update all these entries in the candidate matrix to zero. That is, set entries $(t_1,r_0), (t_1,r_2), (t_1,r_3), (t_1,r_4)$ to zero, as depicted in Figure~\ref{fig:mrv example}(c). Furthermore, to ensure that the added edges do not share common vertices, it is necessary to update all entries in the $t_i$ row and all entries in the $r_j$ column of the candidate matrix to zero. In this example, set all entries in the $t_3$ row and all entries in the $r_1$ column to zero, which is depicted in Figure~\ref{fig:mrv example}(d). 
The complete MRV algorithm is given in Algorithm~\ref{alg:mrv}.

\begin{algorithm}[!htb]
\small
\caption{MRV heuristic algorithm}\label{alg:mrv}
\LinesNumbered

\KwIn{\\  \begin{tabular}{p{3cm}l}
\textit{StaticCircuit} &  the instruction list of a static quantum circuit to compile\\ 
\end{tabular}}

\BlankLine

\KwOut{\\  \begin{tabular}{p{3cm}l}
\textit{DynamicCircuit} & the instruction list of the compiled dynamic quantum circuit\\
\end{tabular}}

\BlankLine

\BlankLine

Run Algorithm~\ref{alg:Converting static quantum circuit to DAG} in Appendix~\ref{app: further algorithms} to get the DAG representation \textit{Digraph} of the circuit\;

Let \textit{Roots} and \textit{Terminals} be the set of roots and terminals, respectively\;

Run Algorithm~\ref{alg:reducibility through Boolean matrix} to get the biadjacency matrix $B$ of the simplified DAG of the circuit\;

Calculate the candidate matrix \textit{CandidateMatrix} as $^\neg (B^{\top})$\;

\BlankLine

Initialize two empty lists \textit{CandidatesNum} and \textit{AddedEdges}\;

\BlankLine

\While{CandidateMatrix is not zero matrix}{

    Calculate the sum of each row of \textit{CandidateMatrix} and record it in \textit{CandidatesNum}\;

    Identify the smallest non-zero element in \textit{CandidatesNum} and record its index as $t$\;

    Identify the non-zero element in the $t$-th row of \textit{CandidateMatrix} with the smallest column sum and record its column index as $r$\;

    Append the edge (\textit{Terminals}[$t$], \textit{Roots}[$r$]) to \textit{AddedEdges}\;

    Let $R$ be the column indices of zero elements in the $t$-th row of \textit{CandidateMatrix}\;

    Let $T$ be the row indices of zero elements in the $r$-th column of \textit{CandidateMatrix}\;

    \ForEach{pair $(u, v)$ where $u \in T$ and $v \in R$}{
        Set the entry $(u, v)$ of \textit{CandidateMatrix} to zero\;
    }

    Set all entries in the $t$-th row and the $r$-th column of \textit{CandidateMatrix} to zero\;
}

\BlankLine

Add \textit{AddedEdges} to \textit{Digraph} and get the \textit{ModifiedGraph}\;

Run Algorithm~\ref{alg:modified DAG to dynamic circuit} in Appendix~\ref{app: further algorithms} to get the compiled circuit \textit{DynamicCircuit}\;

\Return \textit{DynamicCircuit}

\end{algorithm}

\BlankLine

\begin{remark}
Note that the role of root and terminal in Algorithm~\ref{alg:mrv} can be exchanged. That is, we can identify the root with the fewest candidate terminals and connect it to a candidate terminal with least choice of roots. In practice, we can run both algorithms and choose a better result.
\end{remark}

\begin{proposition}\label{prop:time complexity of mrv algorithm}
For a static quantum circuit with $n$ qubits and $m$ operations, Algorithm~\ref{alg:mrv} has a worst-case time complexity of $O(mn + n^3)$.
\end{proposition}

\subsection{Algorithm 2: Greedy Heuristic Algorithm}

Greedy algorithms are time-efficient heuristic strategy that makes a locally optimal choice at each iterative step. This idea can be readily applied to address the dynamic circuit compilation as follows: in each iteration, the algorithm evaluates the potential impact of adding each candidate edge and selects the one that maximizes the possibility to add more edges in subsequent steps. Specifically, for each candidate edge, the algorithm temporarily integrates this edge into the simplified DAG and updates the candidate matrix following the rules outlined in Algorithm \ref{alg:mrv}. The summation of all elements within the updated candidate matrix serves as the score for the candidate edge, which is then stored in the corresponding entry of a matrix called the \emph{score matrix}. After evaluating all candidate edges, the algorithm identifies the candidate edge with the highest score as the optimal choice for inclusion in the current iteration and updates the candidate matrix accordingly before proceeding to the next iteration. In cases where multiple entries share the highest score, one edge is \emph{randomly} selected from among them. An illustrative example of the scoring procedure is depicted in Figure~\ref{fig:greedy example}. The complete greedy algorithm is given in Algorithm~\ref{alg:greedy}. 

\begin{figure}[!htb]
    \centering
    \includegraphics[width=0.8\textwidth]{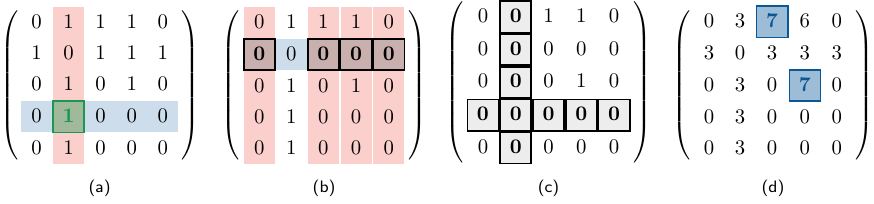}
    \vspace{-2mm}
    \caption{(Color online) An illustration of the greedy heuristic algorithm. Suppose the row and column of the matrix are indexed by $\{t_0, \cdots, t_4\}$ and $\{r_0, \cdots, r_4\}$. (a) the candidate matrix of the simplified DAG of a $5$-qubit Bernstein–Vazirani algorithm $\bar B = ^\neg (B^{\top})$ whereas the candidate edge under evaluation is $(t_3, r_1)$ marked in green; (b) set all entries in $T_3 \times R_1$ to zero , where $T_3 = \{t_1\}$ and $R_1 = \{r_0, r_2, r_3, r_4\}$; (c) set all entries in the $t_3$ row  and all entries in the $r_1$ column to zero, the summation of all elements in the updated candidate matrix serves as the score of the candidate edge $(t_3, r_1)$; (d) the score matrix after the first iteration, where the edge for inclusion in this round is randomly selected from the entries with the highest score (marked in blue).}
    \label{fig:greedy example}
\end{figure}


\begin{algorithm}[!htb]
\small
\caption{Greedy heuristic algorithm}\label{alg:greedy}
\LinesNumbered

\KwIn{\\  \begin{tabular}{p{3cm}l}
\textit{StaticCircuit} &  the instruction list of a static quantum circuit to compile\\
\end{tabular}}

\BlankLine

\KwOut{\\  \begin{tabular}{p{3cm}l}
\textit{DynamicCircuit} & the instruction list of the compiled dynamic quantum circuit\\
\end{tabular}}

\BlankLine

\BlankLine

Run Algorithm~\ref{alg:Converting static quantum circuit to DAG} in Appendix~\ref{app: further algorithms} to get the DAG representation \textit{Digraph} of the circuit\;

Let \textit{Roots} and \textit{Terminals} be the set of roots and terminals, respectively\;

Run Algorithm~\ref{alg:reducibility through Boolean matrix} to get the biadjacency matrix $B$ of the simplified DAG of the circuit\;

Calculate the candidate matrix \textit{CandidateMatrix} as $^\neg (B^{\top})$\;

\BlankLine

Initialize an empty list \textit{AddedEdges}\;

\BlankLine

\While{CandidateMatrix is not zero matrix}{
    Initialize a $n \times n$ zero matrix \textit{ScoresMatrix}\;
    \ForEach{non-zero entry $(i, j)$ in CandidateMatrix}{
        Initialize a matrix $\bar{B}_{i,j}$ as \textit{CandidateMatrix}\;        

        Let $R_i$ be the column indices of zero elements in the $i$-th row of $\bar{B}_{i,j}$\;

        Let $T_j$ be the row indices of zero elements in the $j$-th column of $\bar{B}_{i,j}$\;
    
        \ForEach{pair $(u, v)$ where $u \in T_j$ and $v \in R_i$}{
            Set the entry $(u, v)$ of $\bar{B}_{i,j}$ to zero\;
        }
    
        Set all entries in the $i$-th row and all entries in the $j$-th column of $\bar{B}_{i,j}$ to zero\;
        
        Set the entry $(i, j)$ in \textit{ScoresMatrix} to the sum of all entries in $\bar{B}_{i,j}$ plus one\;
    }

        Identify all entries with the largest score in \textit{ScoresMatrix} as \textit{MaxScore}\;

        Randomly select one entry in \textit{MaxScore} and record its index $(t,r)$\;

        Append the edge (\textit{Terminals}[$t$], \textit{Roots}[$r$]) to \textit{AddedEdges}\;

        Update \textit{CandidateMatrix} to $\bar{B}_{t,r}$
}

\BlankLine

Add \textit{AddedEdges} to \textit{Digraph} and get the \textit{ModifiedGraph}\;

Run Algorithm~\ref{alg:modified DAG to dynamic circuit} in Appendix~\ref{app: further algorithms} to get the compiled circuit \textit{DynamicCircuit}\;

\Return \textit{DynamicCircuit}

\end{algorithm}

\begin{remark}
The scoring rule in the greedy algorithm is flexible and can be replaced with alternative approaches based on specific objectives. One such approach involves evaluating the impact on the circuit depth for adding a candidate edge to the graph by calculating the length of the critical path in the graph. Then the candidate edge can be scored by a predefined cost function that deliberates the tradeoff between resultant circuit width and depth, alongside other pertinent factors.
\end{remark}

\begin{remark}
In scenarios where multiple candidate edges attain the same highest score in an iterative step, a straightforward approach is to choose the one associated with a fixed rule (e.g. the smallest index). However, adopting such a deterministic procedure could inadvertently tie the compilation process to specific qubit labels, potentially restricting the algorithm's performance. To address this constraint, the integration of a random selection process becomes crucial. This stochastic method holds the potential to uncover enhanced solutions by running the algorithm multiple times, a strategy that has proven quite useful in our numerical experiments.
\end{remark} 

\begin{proposition}\label{prop:time complexity of greedy algorithm}
For a static quantum circuit with $n$ qubits and $m$ operations, Algorithm~\ref{alg:greedy} has a worst-case time complexity of $O(mn + n^5)$.
\end{proposition}

\subsection{Algorithm 3: Hybrid Algorithm}

To further enhance the performance of the heuristic algorithms, we introduce a hybrid algorithm by combining the MRV heuristic (also works for other heuristics) and brute force search. Initially, the hybrid algorithm employs a brute force search on a designated subset of terminals, denoted as $T_E \subseteq T$, which exhaustively enumerates all feasible edge additions pertinent to terminals within this subset. Subsequently, the MRV heuristic algorithm is employed to identify edges that can be added to the remaining terminals in $T \backslash T_E$. Let $L$ denote the cardinality of the subset $L = |T_E|$. Notably, at $L=0$ the hybrid algorithm aligns precisely with the the MRV algorithm. As $L$ increases, the hybrid algorithm progressively approximates the characteristics of the brute force search. Upon reaching $L=n$, the hybrid algorithm becomes the brute force search. This hierarchical variation of the hybrid algorithm, characterized by different values of $L$, which we referred as \emph{the hierarchy level}, provides the opportunity to tradeoff between the optimality of the solution and the computational time complexity. Detailed discussions and numerical experiments can be found in Appendix~\ref{sec:hybrid} of the supplementary material.

\section{Analytical and numerical evaluations}
\label{sec:examples}

In this section, we conduct a thorough analysis of quantum circuits with practical relevance, offering optimal compilations for well-known quantum algorithms in quantum computation, ansatz circuits utilized in quantum machine learning, and measurement-based quantum computation crucial for quantum networking. We also perform a comparative analysis against state-of-the-art approaches, demonstrating the superior performance of our methods in both structured and random quantum circuits. A brief summary of the quantum circuits explored in this study is provided in Table~\ref{tab:summary of circuits}. Optimal compilations for the frequently used quantum algorithms and their proofs can be found in Appendix~\ref{app: Optimal compilations for important quantum circuits} of the supplementary material.

\setlength\extrarowheight{1pt}
\begin{table}[!htb]
\footnotesize
\caption{(Color online) Quantum circuits studied in this work. The case in blue indicates irreducible circuits, while the case in red means the compiled width is optimal. Numerical studies are provided for cases in green. 
}
\label{tab:summary of circuits}
\centering
\begin{tabular}{lcc}
\toprule[1pt]
Quantum circuits & Original & Compiled \\
\midrule
Deutsch-Jozsa algorithm & $n$  & \cellcolor{myred!25} $2$ ($1$) \\
Bernstein-Vazirani algorithm & $n$ & \cellcolor{myred!25} $2$ ($1$)\\
Simon's algorithm & $2n$ & \cellcolor{myred!25} $3$ \\
Quantum Fourier transform & $n$ & \cellcolor{myblue!25} $n$ \\
Quantum phase estimation & $n$ & \cellcolor{myblue!25} $n$ \\
Shor's algorithm & $n$ & \cellcolor{myblue!25} $n$ \\
Grover's algorithm & $n$ & \cellcolor{myblue!25} $n$ \\
Quantum counting algorithm & $n$ & \cellcolor{myblue!25} $n$ \\
Linearly entangled circuit with $l$ layers & $n$  & \cellcolor{myred!25} $l+1$ \\
Circularly entangled circuit with $l$ layers & $n$ & \cellcolor{myred!25} $3$ \\
Pairwisely entangled circuit with $l$ layers & $n$ & \cellcolor{myred!25} $2l+1$ \\
Fully entangled circuit & $n$ & \cellcolor{myblue!25} $n$ \\
Diamond-structured quantum circuit & $2n$ & \cellcolor{myred!25} $n + 1$ \\ 
MBQC with cluster state of size $(w,d)$ & $wd$ & \cellcolor{myred!25} $w+1$ \\ 
MBQC with brickwork state of size $(w,d)$ & $wd$ & \cellcolor{myred!25} $w+1$ \\
Quantum ripple carry adder circuit & $n$ ($4$) & \cellcolor{myred!25} $4$ ($3$) \\
Quantum supremacy circuits & - & \cellcolor{mygreen!25} - \\
GRCS circuits & - & \cellcolor{mygreen!25} - \\
QAOA circuits for max-cut & - & \cellcolor{mygreen!25} - \\
Random (IQP) circuits & - & \cellcolor{mygreen!25} - \\
\bottomrule[1pt]  
\end{tabular}
\end{table}

To compare the reducibility of quantum circuits as well as the performance of different compilation methods, we define the \emph{reducibility factor} of a quantum circuit as
\begin{align}
    r = 1 - \frac{n'}{n} \in [0, 1),
\end{align}
where $n$ is the width of the original circuit and $n'$ is the width of the compiled circuit. This factor characterizes the extent to which the circuit width can be reduced by a certain algorithm, which is zero if the circuit is not reducible.

\subsection{Quantum supremacy circuits}

\label{sec: Quantum supremacy circuits}

In this section, we analyze the reducibility factor of quantum circuits used to claim quantum supremacy, including those executed on Sycamore and Zuchongzhi quantum computers. Figure~\ref{fig:supermacy circuit}(a) displays the reducibility factor of these circuits with varying numbers of cycles, determined using the greedy heuristic algorithm (Algorithm~\ref{alg:greedy}).

An interesting observation is that quantum circuits with $70$ qubits and $24$ cycles running on Sycamore~\cite{morvan2023phase}, $56$ qubits and $20$ cycles on Zuchongzhi~\cite{wu2021strong}, and $60$ qubits and $24$ cycles on Zuchongzhi~\cite{zhu2022quantum} are used to claim quantum supremacy, but they all belong to irreducible circuits. This observation reveals a fundamental tradeoff between the technical challenges associated with running deep circuits (a limitation of current quantum computers) and the structural complexity of these circuits (a limitation of classical computers). It highlights the delicate balance required when designing quantum circuits to showcase quantum supremacy in the near term.

Furthermore, another noteworthy point is the quantum circuit with $53$ qubits and $20$ cycles executed on Sycamore~\cite{arute2019quantum}, which possesses a reducibility factor of $0.02$. This result indicates that the circuit can be compiled into a quantum circuit with $52$ qubits. This observation aligns with the historical fact that one qubit on the Sycamore chip is non-functional, thereby breaking the complexity of the circuit and leaving the room for its compilation.

\begin{figure}[!htb]
    \begin{subfigure}[t]{0.48\textwidth}
        \includegraphics[width=\textwidth]{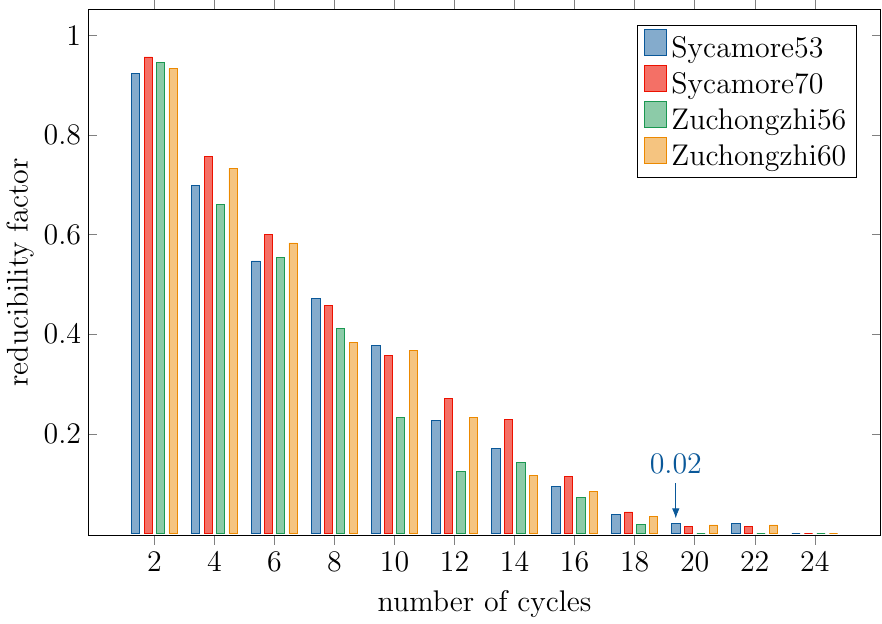}
        \vspace{-5mm}
        \caption*{\scriptsize (a)}
    \end{subfigure}
    \hspace{2mm}
    \begin{subfigure}[t]{0.48\textwidth}
        \includegraphics[width=\textwidth]{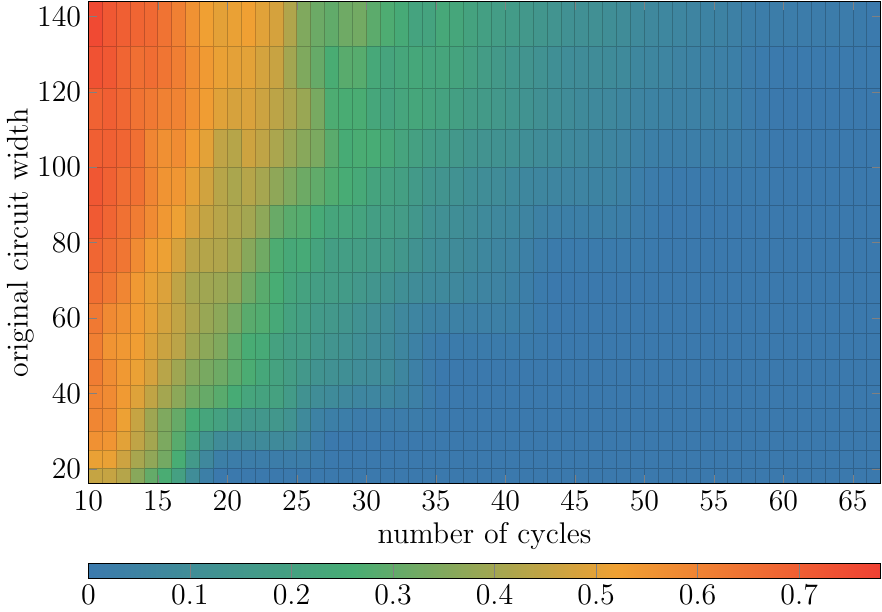}
        \vspace{-5mm}
        \caption*{\scriptsize (b)}
    \end{subfigure}
    \vspace{-1mm}
    \caption{(Color online) (a) The reducibility factor of different quantum supremacy circuits using the greedy heuristic algorithm~\ref{alg:greedy}; (b) The reducibility factor of GRCS circuits located in the GitHub directory `inst/rectangular' with different circuit widths and numbers of cycles using the greedy heuristic algorithm~\ref{alg:greedy}.}
    \label{fig:supermacy circuit}
\end{figure}

In 2018, Google proposed a series of random quantum circuits (\href{https://github.com/sboixo/GRCS}{GRCS})~\cite{boixo2018characterizing}. Due to the hardness of simulation, GRCS is frequently used as benchmark to test the performance of classical simulators. Each instance in GRCS is a random quantum circuit designed for qubits configured in an $n \times m$ lattice. These circuits are composed of multiple cycles of quantum gates. Figure~\ref{fig:supermacy circuit}(b) demonstrates the reducibility factor of GRCS circuits through a single running of the greedy heuristic algorithm~\ref{alg:greedy}. It further validates the earlier observation that the larger the number of cycles (depth), the more difficult for a quantum circuit to reduce.

\subsection{Algorithm benchmarking}
\label{sec: Algorithm benchmarking}

In this section, we conduct a numerical analysis to assess the performance of different algorithms on a range of benchmark circuits, including quantum adders, quantum approximate optimization algorithm (QAOA), random quantum circuits and random IQP circuits. Our primary focus centers around three distinct algorithms for dynamic circuit compilation: the MRV heuristic algorithm (Algorithm~\ref{alg:mrv}), the greedy heuristic algorithm (Algorithm~\ref{alg:greedy}) and the greedy algorithms proposed in~\cite{decross2022qubitreuse} (referred to as DCKF in the subsequent discussion). As the source code for the DCKF algorithms is not publicly available, we have implemented these algorithms based on our understanding of the paper~\cite{decross2022qubitreuse}. Further details regarding our implementation of the DCKF algorithms can be found in Appendix~\ref{sec:dckf implementation} of the supplementary material. For the MRV algorithm, we perform two separate runs, exchanging the roles of roots and terminals within the algorithm for each run. Then, we select the dynamic circuit with the smaller circuit width as the final output. We have also utilized the stochastic nature of our greedy heuristic algorithm by running it multiple times to improve its performance.

\subsubsection{Quantum Ripple Carry Adders}

Quantum adder is a quantum circuit designed for performing addition operation between two bit strings. For example, if we compute `$1+2=3$', then we represent the input string as `01' and `10', and the expected output bit string is `11'. Here we focus on the Quantum Ripple Carry Adders initially proposed in~\cite{chakrabarti2008designing}. 

In Figure~\ref{fig:quantum adder and qaoa}(a), we present the results of the numerical experiments conducted using three different algorithms. It is evident that both MRV and the greedy algorithm successfully find the optimal compilation. In contrast, the results obtained from the DCKF algorithms display linear scaling in the original circuit width, indicating a deficiency in its performance. This limitation is likely linked to a specific implementation aspect. That is, the process of establishing measurement orders within the DCKF algorithms sometimes generates multiple local optima within a single iteration. Consequently, deterministic selections in these scenarios might lead to unfavorable measurement orders, significantly affecting overall performance.  This justifies the reasoning behind our inclusion of randomness within our greedy heuristic algorithm (Algorithm~\ref{alg:greedy}).

\begin{figure}[H]
    \begin{subfigure}[t]{0.48\textwidth}
        \includegraphics[width=\textwidth]{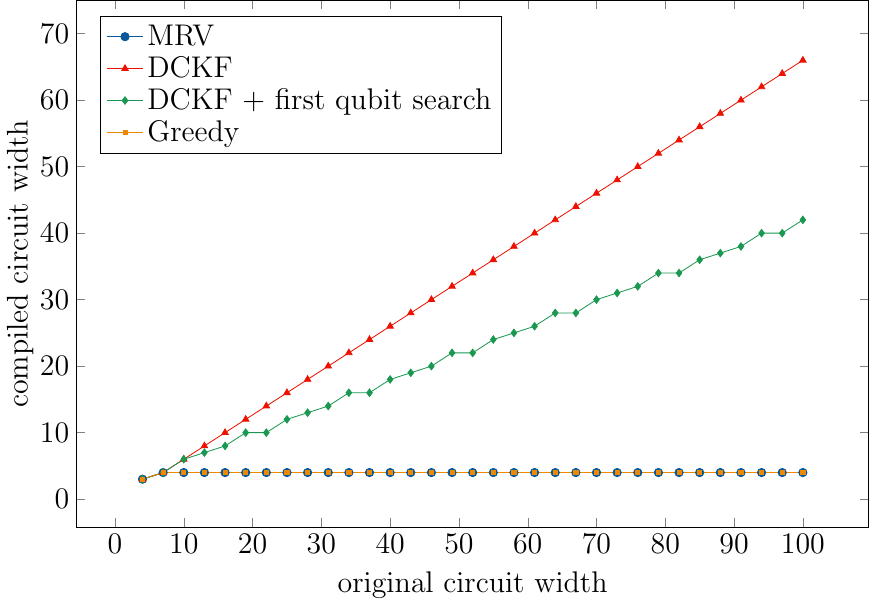}
        \vspace{-5mm}
        \caption*{\scriptsize (a)}
    \end{subfigure}
    \hspace{2mm}
    \begin{subfigure}[t]{0.48\textwidth}
        \includegraphics[width=\textwidth]{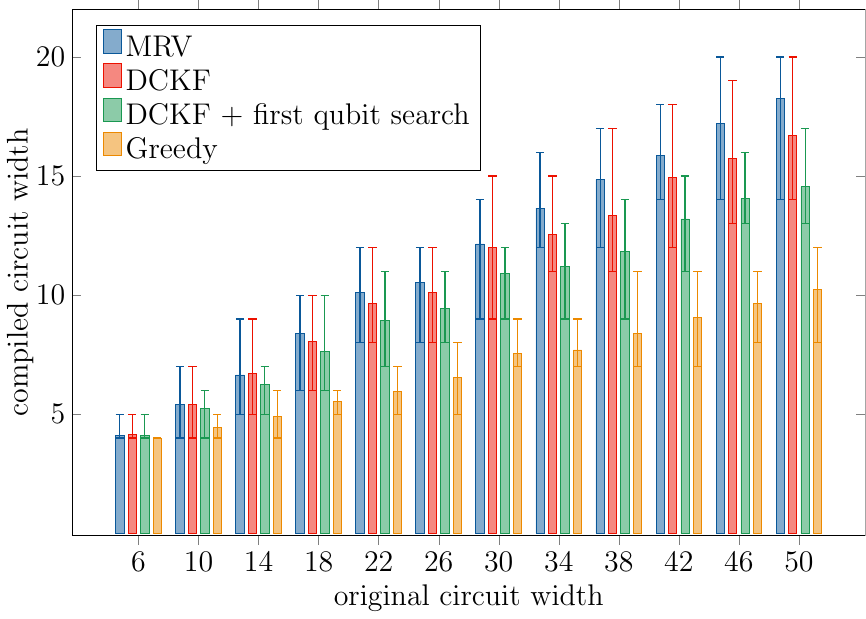}
        \vspace{-5mm}
        \caption*{\scriptsize (b)}
    \end{subfigure}
    \vspace{-1mm}
    \caption{(Color online) (a) Compiled circuit width against the original circuit width of quantum ripple carry adders; (b) Compiled circuit width against the original circuit width of the max-cut QAOA circuits with $p=1$. The plotted error bars correspond to the maximum and minimum compiled width over 20 instances.}
    \label{fig:quantum adder and qaoa}
\end{figure}

\subsubsection{QAOA circuits for max-cut problem}

QAOA \cite{farhi2014quantum} is a quantum algorithm designed to approximately solve classical combinatorial optimization problems and have the potential to run on near-term quantum devices. The QAOA unitary takes the form of the alternative application of a mixing unitary and a problem unitary for $p$ layers.
A max-cut problem is a combinatorial optimization problem in graph theory which involves to find a partition of vertices into two sets, such that the number of edges between the sets is maximized. 
In QAOA circuits designed for solving the max-cut problem, the number of qubits matches the number of vertices in the graph, and the connectivity of two-qubit gates corresponds to the edges in the graph. 

Here, we assess the performance of different algorithms applied to QAOA circuits for solving the max-cut problem on random unweighted three-regular (U3R) graphs with $p = 1$. For each experiment, we ran our greedy heuristic algorithm 10 times and recorded the best result. We evaluated four algorithms for each fixed qubit number on 20 random U3R graphs generated using the \textit{NetworkX} package~\cite{networkx}. The results are presented in Figure~\ref{fig:quantum adder and qaoa}(b). It is evident that the average compiled width achieved by our greedy algorithm is consistently lower than that obtained using the DCKF algorithms for all qubit numbers. Moreover, as the number of qubits increases, this advantage becomes increasingly pronounced.

\subsubsection{Random circuits}

In addition to the previously studied structured circuits, we conducted comprehensive numerical experiments involving random quantum circuits to assess algorithm performance across a broader spectrum of scenarios. These experiments involved fixing the ratio $r = m/n$, where $m$ represents the number of two-qubit gates, and $n$ represents the width of the original circuits. We uniformly and randomly selected a qubit number from the range between 10 and 80 and sampled the desired number of two-qubit gates to construct the circuit.

We evaluated the reducibility factor using both our greedy heuristic and the DCKF algorithms on these random circuits. For each fixed ratio, we sampled 300 random circuits and ran our greedy algorithm 15 times for each instance. The results in Figure~\ref{fig:random circuit}(a) demonstrate that our greedy heuristic (vertical) outperforms the DCKF algorithm (horizontal) in approximately $98.5\%$ of instances. 

To underscore the significance of handling circuits with commutable structures, we further evaluated the reducibility factor using our greedy heuristic and the improved DCKF algorithms on random IQP circuits.  We sampled 300 random IQP circuits for each fixed ratio and ran our greedy algorithm 10 times for each instance. As depicted in Figure~\ref{fig:random circuit}(b), our greedy algorithm outperforms in nearly $100\%$ of instances with an absolute advantage in $98.4\%$ of cases. More numerical analysis can be found in Appendix~\ref{app: Experimental Results Continued} of the supplementary material.

\begin{figure}[H]
    \begin{subfigure}[t]{0.3\textwidth}
        \centering
        \includegraphics[width=\textwidth]{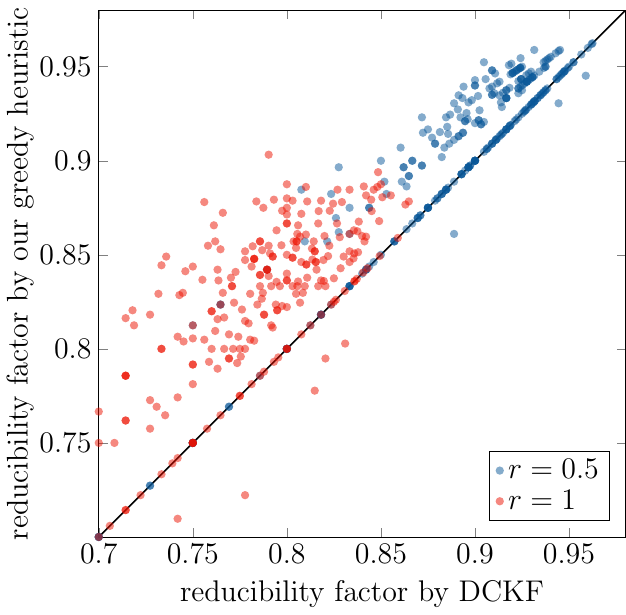}
        \vspace{-6mm}
        \caption*{\tiny \qquad (a)}
    \end{subfigure}
    \hspace{3mm}
    \begin{subfigure}[t]{0.32\textwidth}
        \centering
        \includegraphics[width=\textwidth]{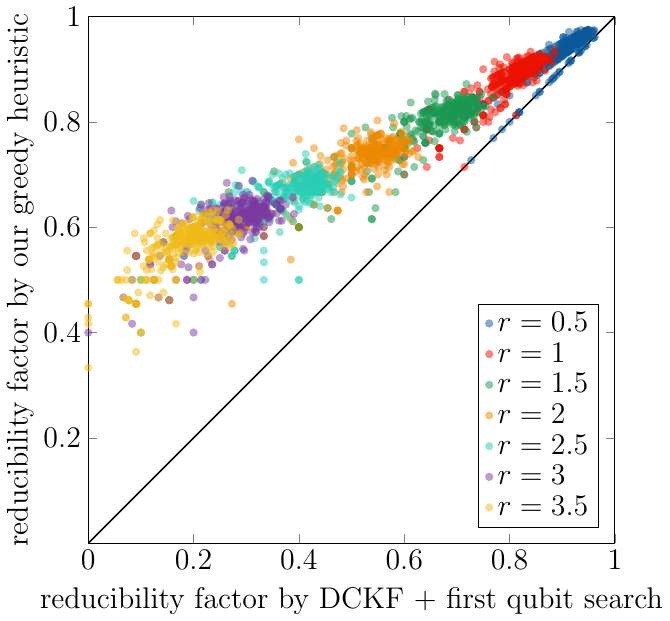}
        \vspace{-6mm}
        \caption*{\tiny (b)}
    \end{subfigure}
    \hspace{-2mm}
    \begin{subfigure}[t]{0.3\textwidth}
        \centering
        \includegraphics[width=\textwidth]{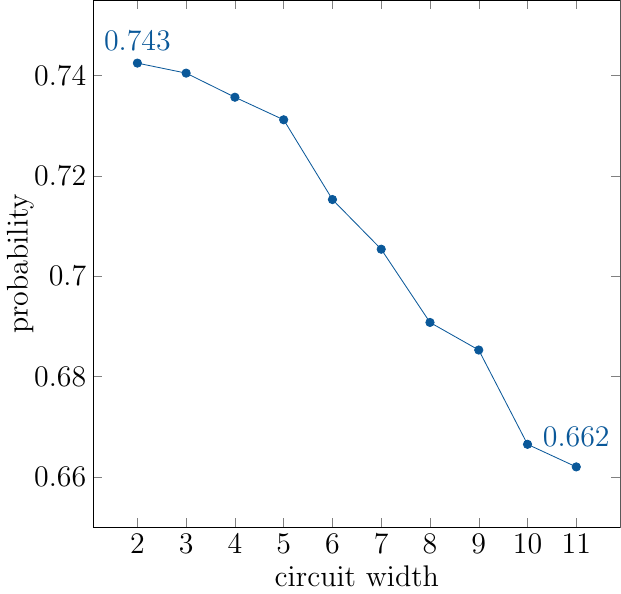}
        \vspace{-6mm}
        \caption*{\tiny \qquad \quad (c)}
    \end{subfigure}
    \vspace{-2mm}
    \caption{(Color online) (a) The reducibility factor of the randomly generated quantum circuits evaluated using our greedy heuristic algorithm~\ref{alg:greedy} and the DCKF algorithm. (b) The reducibility factor of the randomly generated IQP circuits evaluated using our greedy heuristic algorithm~\ref{alg:greedy} and the improved version of DCKF algorithm. The black diagonal line indicates the point at which the reducibility factors produced by both algorithms are equal. (c) The probability of obtaining a correct outcome (all-one string) against the compiled circuit width in the noisy simulation of an 11-qubit Bernstein-Vazirani algorithm.}
    \label{fig:random circuit}
\end{figure}

\subsubsection{Noisy simulation}

Error variability poses a challenge in near-term quantum hardware, making certain qubits perform better than others. By maximizing qubit reuse, we can consistently select qubits with superior performance, thereby enhancing the algorithm's performance. To further demonstrate the practical efficacy of the proposed methods, we design a noisy simulation of an 11-qubit Bernstein-Vazirani (BV) algorithm, specifically targeting the real-world 11-qubit trapped-ion quantum computer reported in~\cite{wright2019benchmarking}. The secret bitstring of the BV algorithm is set to an all-one string. In our simulation, we gradually reduce the number of qubits from 11 to 2 and map the logical qubits onto the physical qubits of the hardware, systematically eliminating a physical qubit with a higher error rate at each step. The resulting probability of obtaining the correct outcome, plotted against the circuit width, is depicted in Figure~\ref{fig:random circuit}(c). The utilization of dynamic quantum circuit compilation enables us to reduce qubit usage by up to $82\%$ while also improving the probability of achieving accurate results by $8\%$. Note that this example is for illustrative purposes, and the advantages of dynamic circuit compilation are expected to be even more prominent when applied to larger-scale algorithms and quantum computers. Detailed information about the noisy simulation is available in Appendix~\ref{app: Experimental Results Continued} of the supplemental material.

\section{Related Work}
\label{sec:related work}

The work~\cite{paler2016wire} studied quantum circuit compilation by wire recycling. They constructed a causal graph to represent the temporal ordering of quantum circuit operations and analyzed the lifetimes of qubits to exploit the potential of recycling wires. They developed two heuristic algorithms based on graph search for wire recycling. However, the proposed method is limited to recycling wires between pre-defined ancilla qubits.

The work~\cite{decross2022qubitreuse} investigated quantum circuit compilation by leveraging the causal structure of the circuits. They formulated the task of minimizing the required number of qubits as a constraint programming and satisfiability (CP-SAT) model. This model incorporates a bunch of constraints and is primarily utilized to numerically benchmark their heuristic algorithms on small-scale problems. In contrast, our approach utilizes a graph manipulation framework that induces straightforward binary integer programming for optimal compilation. This framework has been effectively employed to establish the optimal compilation of numerous quantum circuits. Additionally, \cite{decross2022qubitreuse} proposed greedy heuristic algorithms for approximate compilation. However, our comparative analysis highlights the superior performance of our methods across both structured and random quantum circuits. Notably, our framework successfully addresses an open challenge emphasized in their work, namely, the effective handling of quantum circuits with commutable structures and the ability to conduct compilation at the level of quantum algorithms, regardless of their specific quantum instruction sequences.

The work~\cite{Fei-CaQR} explored the tradeoff between qubit reuse, fidelity, gate count, and circuit duration. They established two conditions for qubit-reuse and designed two versions of compiler-assisted tools with one prioritizing qubit-saving and the other emphasizing SWAP reduction and fidelity improvement. Their empirical demonstrations on quantum hardware showcased notable improvements in qubit usage and circuit fidelity for specific applications, primarily focused on superconducting quantum computers.
Our approach, on the other hand, centers on minimizing the required number of qubits—a scenario well-motivated by trapped ion quantum systems. Additionally, we contribute a more adaptable framework capable of facile extensions to accommodate various optimization objectives. For instance, by fine-tuning the cost function within our greedy algorithm's scoring process, we can explore tradeoffs between circuit width, depth, and other related factors, encompassing their qubit-saving approach.
Furthermore, our study introduces efficient techniques for assessing a quantum circuit's reducibility, serving as a preprocessing step to screen circuits before actual compilation. The optimal compilations identified in our work can serve as benchmarks for other variants of qubit-reuse compilation methods.

The work~\cite{brandhofer2023optimal} introduced a formal SAT-based model for qubit reuse optimization on near-term quantum devices. This model ensured provably optimal solutions concerning quantum circuit depth, number of qubits, or number of swap gates. However, their approach may encounter serious computational challenges and scalability issues as the number of qubits increases. 
\section{Conclusion and Future work}
\label{sec:discussion}

We have conducted a comprehensive investigation into the dynamic circuit compilation problem, introducing the first characterization of this task through graph manipulation and a precise mathematical model for optimal compilation. Our framework primarily targets qubit savings but is general enough to be adapted to other scenarios. The effectiveness of our approach was demonstrated through theoretical analyses of reducibility and optimal compilation for various renowned quantum circuits, as well as numerical evaluation of our heuristic algorithms on a wide range of benchmark circuits. It is worth noting that the dynamic circuit compilation explored in this work offers a complementary strategy to other circuit optimization techniques and can be seamlessly integrated with existing methods. As we approach the point of demonstrating quantum advantage in practical applications, our results shall serve as timely contributions for bridging the gap between theoretical quantum algorithms and their physical implementation on quantum computers with limited resources.

In addition to the practical utility, our work also establishes the connection between dynamic circuit compilation and graph theory. We believe there are many other techniques from graph theory that could be used to extend our framework and be applied in the area of quantum circuit compilation and optimization. An in-depth discussion of related open problems can be found in Appendix~\ref{sec:open problems} of the supplementary material.

\section*{Acknowledgements}
We would like to thank Jingtian Zhao for part of the code implementation in QNET. This work was done when M. Z. and R. S. were research interns at Baidu Research. Y. L. is supported by the National Nature Science Foundation of China (No. 62302346) and supported by ``the Fundamental Research Funds for the Central Universities''.

\bibliographystyle{halpha}
\bibliography{Bib}

\clearpage

\newpage
\appendix
\noindent \textbf{\Large Supplementary Material}

\vspace{0.5cm}

This supplementary material provides a more detailed analysis and proof of the results in the main text. It also contains some further algorithms and their experimental results.

\begin{itemize}
    \item Appendix~\ref{app: further algorithms}: Further algorithms
    \item Appendix~\ref{app: Proofs of results}: Proofs of results
    \item Appendix~\ref{app: Optimal compilations for important quantum circuits}: Optimal compilations for important quantum circuits
    \item Appendix~\ref{sec:dckf implementation}: Implementation of DCKF algorithm
    \item Appendix~\ref{app: Experimental Results Continued}: Experimental results continued
    \item Appendix~\ref{sec:open problems}: Open problems
\end{itemize}

\section{Further Algorithms}
\label{app: further algorithms}

In this section, we provide further algorithms to support our compilation framework.

\subsection{Converting static quantum circuit to DAG}

Algorithm~\ref{alg:Converting static quantum circuit to DAG} below provides an efficient procedure for transforming a quantum circuit from its circuit instructions into the corresponding DAG representation. Particularly, to process circuits with commutable structure, each instruction is associated with a group tag, indicating to which group of gates it belongs. Operations within the same group are commutable, whereas a `None' tag denotes a non-commutable gate. The algorithm traverses the static circuit instructions, generating vertices in the directed graph for each quantum instruction encountered. It then iterates over the qubits involved in the instruction, examining the preceding operation on each qubit. In cases where the preceding operation is non-commutable, a directed edge is established from the previous operation to the current one. However, if the previous operation is commutable, it needs to check whether the previous operation and the current operation belong to the same commutable group. If they do, it traverses the \textit{CasualList} of this qubit in reversed order and identifies the previous commutable group. If this group is `None', a directed edge is created from the first non-commutable operation to the current one. Conversely, if a previous commutable group is identified, the algorithm connects all operations on this qubit belonging to the previous commutable group to the current one. 


\begin{algorithm}[!htb]
\small
\caption{Converting static quantum circuit to DAG}
\label{alg:Converting static quantum circuit to DAG}
\LinesNumbered

\KwIn{\\  \begin{tabular}{p{2.5cm}l}
     \textit{StaticCircuit} & a static quantum circuit instructions
\end{tabular}} 

\BlankLine

\KwOut{\\ \begin{tabular}{p{2.5cm}l} 
\textit{Digraph} & a DAG representation of the static quantum circuit
\end{tabular}}

\BlankLine

Let $n$ be the width of quantum circuit\;

Initialize an empty directed graph \textit{Digraph}\; 

Initialize a list \textit{CausalLists} of length $n$, with each element initialized as an empty list\;

\BlankLine

\ForEach{Instruction {\rm \textbf{in}} StaticCircuit}{
    Add a vertex \textit{Vertex} labeled as ID of \textit{Instruction} to \textit{Digraph}\;
    Identify the group tag of \textit{Instruction} as \textit{Group}\;
    \ForEach{q {\rm \textbf{in}} QUBIT value of instruction}{
        Find the last entry of \textit{CausalLists}[$q$] and record it as \textit{PreVertex}\;
        \If{PreVertex exits}{ 
            Identify the group tag of \textit{StaticCircuit}[\textit{PreVertex}] as \textit{PreGroup}\;
            \If{PreGroup is None}{
                Add a directed edge from \textit{PreVertex} to \textit{Vertex} to \textit{Digraph}\;
            }
            \Else{
                \If{Group is equal to PreGroup}{
                    \ForEach{V {\rm \textbf{in}} reversed(CausalLists{\rm [$q$]})}{
                        Identify the the group tag of \textit{StaticCircuit}[\textit{V}] as \textit{CurrentGroup}\;
                        \If{CurrentGroup is not equal to Group}{
                            Set \textit{PreGroup} to \textit{CurrentGroup}\;
                            \If{PreGroup is None}{
                                Add a directed edge from \textit{V} to \textit{Vertex} to \textit{Digraph}\;
                                \textbf{End else}\;
                            }
                            \textbf{Break foreach}\;
                        }
                    }
                }
                Identify all vertices in \textit{CausalLists}[$q$] belonging to \textit{PreGroup}\;
                Add directed edges from all identified vertices to \textit{Vertex} to \textit{Digraph}\;  
            }
        }
        Append ID of \textit{Instruction} to the end of list \textit{CausalLists}[$q$]\;
    }
}

\BlankLine

\Return {Digraph}

\end{algorithm}

The time complexity of this algorithm is analyzed as follows. Assume that the input static circuit has $n$ qubits and $m$ operations. The outermost `foreach' loops over the static circuit instruction, which takes $O(m)$ times. In cases where the previous operation on a qubit is not commutable, we only add an edge to the graph, which can be done in $O(1)$ times. The primary complexity arises when the previous operation and the current operation belong to the same commutable group, where we need to traverse the \textit{CasualList} of this qubit in reversed order, which takes $O(m / n)$ times in the worst-case ($m / n$ indicates the average number of operations on a qubit). Then adding edges from operations in the previous commutable group to the current operation also takes at most $O(m / n)$ times. Therefore, the overall time complexity of this algorithm can be calculated as $O(m \times (m / n + m /n)) = O(m^2 / n)$.


It's important to recognize that imposing dependencies between commutable operations may limit the opportunities for qubit-reuse. As an example, consider the quantum circuit in Figure~\ref{fig:circuits with commutable structure 1}, where all the $CZ$ gates are commutable. As depicted by Figure~\ref{fig:circuits with commutable structure 2}, ignoring the commutability will impose some unnecessary dependencies (edges marked in red), ultimately limiting the potential for qubit-reuse. The DAG with flexible dependencies is shown in Figure~\ref{fig:circuits with commutable structure 3}, where two qubits can be reused and the circuit can be reduced to 2 qubits.

\begin{figure}[!htb]
    \centering
    \begin{subfigure}[b]{.28\textwidth}
        \centering
        \includegraphics[width=0.7\textwidth]{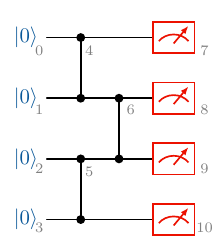}
        \caption{quantum circuit}
        \label{fig:circuits with commutable structure 1}
    \end{subfigure}
    \begin{subfigure}[b]{.34\textwidth}
        \centering
        \includegraphics[width=0.6\textwidth]{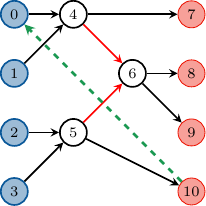}
        \caption{DAG with imposed dependencies}
        \label{fig:circuits with commutable structure 2}
    \end{subfigure}
    \hspace{-0.5em}
    \begin{subfigure}[b]{.34\textwidth}
        \centering
        \includegraphics[width=0.6\textwidth]{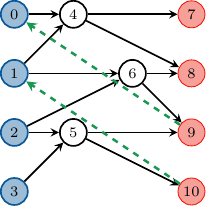}
        \caption{DAG with flexible dependencies}
        \label{fig:circuits with commutable structure 3}
    \end{subfigure}
    \caption{(Color online) An illustration of compiling quantum circuits with commutable structure.}
    \label{fig:circuit independent compilation}
\end{figure}

\subsection{Converting modified DAG to dynamic quantum circuit}

After getting a modified graph, Algorithm~\ref{alg:modified DAG to dynamic circuit} can efficiently convert it to a dynamic quantum circuit.


\begin{algorithm}[!htb]
\small
\caption{Converting modified DAG to dynamic quantum circuit}\label{alg:modified DAG to dynamic circuit}
\LinesNumbered
\BlankLine

\KwIn{\\  \begin{tabular}{p{2.5cm}l}
     \textit{StaticCircuit} & the instruction list of a static quantum circuit\\
     \textit{ModifiedGraph} & a modified DAG\\
     \textit{AddedEdges} & a list of added edges
\end{tabular}} 

\BlankLine

\KwOut{\\ \begin{tabular}{p{2.5cm}l} \textit{DynamicCircuit} & the compiled dynamic quantum circuit instructions
\end{tabular}}

\BlankLine

Initialize empty lists \textit{TopologicalOrder} and \textit{DynamicCircuit}\;

\BlankLine

Apply topological sorting to \textit{ModifiedGraph}, record the result in \textit{TopologicalOrder}\;

Rearrange the order of instructions in \textit{StaticCircuit} according to the order of vertices in \textit{TopologicalOrder}, record the result in \textit{DynamicCircuit}\;

\ForEach{Edge {\rm \textbf{in}} AddedEdges}{
    Let $v_i$ and $v_j$ be the head and tail vertices of \textit{Edge}, respectively\;
    Identify the circuit instructions corresponding to $v_i$ and $v_j$, record the QUBIT values of these two instructions as $q_i$ and $q_j$, respectively\;
    \ForEach{Instruction {\rm \textbf{in}} DynamicCircuit}{
        \ForEach{Qubit {\rm \textbf{in}} QUBIT value of Instruction}{
            \lIf{Qubit is equal to $q_i$}{update \textit{Qubit} to $q_j$}
        }
    }
}
\Return \textit{DynamicCircuit}
\end{algorithm}

\BlankLine

The time complexity of Algorithm \ref{alg:modified DAG to dynamic circuit} is analyzed as follows. Assume that the input static circuit has $n$ qubits and $m$ operations. Initially, the modified DAG is topologically sorted using the Depth-First Search (DFS) algorithm. Denoting the number of vertices and edges in the modified DAG as $|V|$ and $|E|$, it is clear that $|V| = m$. As each operation introduces at most two edges to the graph, $|E|$ scales linearly with $m$, resulting in $|E| = O(m)$. Consequently, the time complexity of topological sorting with DFS is $O(|V| + |E|) = O(m)$~\cite{bang2008digraphs}. Following this, a traversal of the topological order with $m$ vertices is performed to reorder the static circuit instructions. Note that the ID of an instruction and the label of the corresponding vertex are the same. Therefore, for each vertex label $i$, the corresponding instruction is accessed through \textit{StaticCircuit}[$i$] and appended to \textit{DynamicCircuit}, which allows the rearrangement to be completed in $O(m)$ time. Subsequently, for each added edge, we need to traverse the \textit{DynamicCircuit} list and update the qubit indices. For a static circuit with $n$ qubits, at most $n-1$ edges can be added to the DAG, therefore the updating step exhibits a time complexity of $O(mn)$. Consequently, the overall time complexity of Algorithm~\ref{alg:modified DAG to dynamic circuit} is computed as $O(m) + O(m) + O(mn) = O(mn)$.

\subsection{Hybrid Algorithm in details}
\label{sec:hybrid}

Given that the search space is the Cartesian product of candidate roots sets of all terminals in $T_E$, the algorithm initiates by identifying $L$ terminals with the least number of candidate roots, which significantly contracts the search space. Note that for each $t \in T_E$, we add a $\perp$ to the set of candidate roots to represent the situation where $t$ connects no root. Due to the constraint that each terminal can be connected to only one root, solutions with repeated items (except $\perp$) in the calculated search space are rendered unfeasible and are consequently eliminated. Subsequently, for each feasible solution, we check whether the DAG with these additional edges is acyclic. Once acyclicity is verified, Algorithm \ref{alg:reachability through dfs} is employed on the modified DAG to update the candidate matrix for the remaining terminals and roots, followed by Algorithm \ref{alg:mrv} to identify the further added edges. Finally, the solution featuring the highest number of added edges is selected to compile the input static circuit. The complete algorithm is as follows.

\BlankLine

\begin{algorithm}[H]
\small
\caption{Hybrid algorithm with hierarchy level $L$}\label{alg:hybrid}
\LinesNumbered

\KwIn{\\  \begin{tabular}{p{3cm}l}
     $L$                      & the hierarchy level\\
     \textit{StaticCircuit}   & the instruction list of a static quantum circuit to compile\\
\end{tabular}}

\BlankLine

\KwOut{\\  \begin{tabular}{p{3cm}l}
    \textit{DynamicCircuit} & the instruction list of the compiled dynamic quantum circuit\\
\end{tabular}}

\BlankLine

Run Algorithm~\ref{alg:Converting static quantum circuit to DAG} to get the DAG representation \textit{Digraph} of the circuit\;

Let \textit{Roots} and \textit{Terminals} be the set of roots and terminals, respectively\;

Run Algorithm~\ref{alg:reducibility through Boolean matrix} to get the biadjacency matrix $B$ of the simplified DAG of the circuit\;

Calculate the candidate matrix \textit{CandidateMatrix} as $^\neg (B^{\top})$\;

\BlankLine

Initialize empty lists \textit{MaxEdges}, \textit{CandidatesNum}, \textit{IdxTerminals}, \textit{IdxRoots} and \textit{SearchSpace}\;

\BlankLine

Let \textit{CandidatesNum} be the sum of each row in \textit{CandidateMatrix}\;

Let \textit{IdxTerminals} be indices of the first $L$ smallest non-zero values in \textit{CandidatesNum}\;

\ForEach{Index {\rm \textbf{in}} IdxTerminals}{
    Let $R$ be indices of non-zero entries in the \textit{Index}-th row of \textit{CandidateMatrix}\;
    Add an element $\perp$ to $R$, and append $R$ to \textit{IdxRoots}\;
}

Calculate the Cartesian product of sets in \textit{IdxRoots} and record it in \textit{SearchSpace}\;

Delete any elements with repeated items (except $\perp$) from \textit{SearchSpace}\;

\BlankLine

\ForEach{Solution {\rm \textbf{in}} SearchSpace}{
    Initialize an empty list \textit{AddedEdges}\;
    
    \For{$i=0$ {\rm \textbf{to}} $L-1$}{
        \If{\textit{Solution}[$i$] is not $\perp$}{
        Add edge (\textit{Terminals}[\textit{IdxTerminal}[$i$]], \textit{Roots}[\textit{Solution}[$i$]]) to \textit{AddedEdges}\;
        }
        
    }

    \If{Digraph with AddedEdges is acyclic}{
        Run Algorithm \ref{alg:reachability through dfs} with \textit{Digraph} and \textit{AddedEdges} to get a updated simplified DAG and then the updated candidate matrix \textit{CanSubMatrix}\;
        Run Algorithm \ref{alg:mrv} with \textit{CanSubMatrix} to get the added edges \textit{Edges}\;
        Append \textit{Edges} to \textit{AddedEdges}\;
        \If{the length of AddedEdges is larger than the length of MaxEdges}{
            Set \textit{MaxEdges} to \textit{AddedEdges}\;
        }
    }
}

\BlankLine

Add \textit{MaxEdges} to \textit{Digraph} and get the \textit{ModifiedGraph}\;

Run Algorithm \ref{alg:modified DAG to dynamic circuit} with the \textit{ModifiedGraph} to get the compiled circuit \textit{DynamicCircuit}\;

\Return \textit{DynamciCircuit}

\end{algorithm}

\BlankLine

\begin{proposition}\label{prop:time complexity of hybrid algorithm}
For a static quantum circuit with $n$ qubits and $m$ operations, Algorithm~\ref{alg:hybrid} with hierarchy level $L$ has a worst-case time complexity of $O(n^L m^2 (n-L)^2)$.
\end{proposition}
\begin{proof}
The primary complexity of this algorithm arises from the enumeration process (the second `foreach' loop). In the worst-case, each terminal $t \in T_E$ has $(n-1)$ candidate roots, leading to a search space of size at most $(n-1)^L$. Within the search space, each solution undergoes a topological sorting to check whether DAG with additional edges is acyclic. This operation leverages the DFS algorithm, which carries a time complexity of $O(m)$. Following the topological sorting step, Algorithm~\ref{alg:reachability through dfs} is executed on the modified DAG to update both the simplified DAG and the candidate matrix of the remaining $(n-L)$ terminals and roots, which demands $O(m(n-L)^2)$ time as outlined in Proposition~\ref{prop:reducibility through DFS}. Subsequently, the MRV heuristic algorithm is employed to identify edges that can be added between the remaining terminals and roots. Proposition~\ref{prop:time complexity of mrv algorithm} indicates that the MRV algorithm operating on a $(n-L) \times (n-L)$ candidate matrix exhibits a worst-case time complexity of $O(m(n-L) + (n-L)^3)$. Consequently, the total time complexity of the hybrid algorithm is given by: $O(n^L m^2 (n-L)^2) + O(n^L m^2 (n-L)) + O(n^L m (n-L)^3)$, where the dominant factor is $O(n^L m^2 (n-L)^2)$ since $m$ is typically larger than $n$.
\end{proof}

\section{Proofs of results}

\label{app: Proofs of results}

\subsection{Thm. \ref{thm: circuit compilation via graph manipulation}}

\begin{proof}
Dynamic quantum circuit compilation through qubit-reuse seeks to delay certain reset operations until after the measurements, which is illustrated by the addition of a directed edge from a terminal to a root in the DAG representation. However, when incorporating new edges, it is imperative to adhere to the following constraints.
\begin{enumerate}
    \item Resetting a qubit is possible only when all operations on it have been carried out. So the added edge should start from terminals. Moreover, the circuit width is determined by the number of roots in the DAG representation. To reduce the circuit width, the added edges should end at roots. Overall, the directed edges should be added from terminals to roots. This corresponds to the first condition in the asserted result.
    \item A reused qubit can only accommodate one reset operation. So the added edges should have no common tails. Moreover, since the circuit width is determined by the number of roots in the DAG representation, the added edges having common heads will not help. So we can restrict our attention to the case that the added edges share no common heads. This corresponds to the second condition in the asserted result.
    \item Since directed edges represent the execution order of operations, the presence of any cycle in the graph indicates a dependency of past operations on future ones, which violates the causal relation of the quantum circuit. Therefore, the addition of these directed edges must not introduce any cycles, indicating the compiled circuit is still well-defined. This corresponds to the third condition in the asserted result.
\end{enumerate}
Conversely, for any graph manipulation that complies with the specified conditions, we can demonstrate that it corresponds to a dynamic circuit compilation approach. Given that the DAG representation preserves the execution order of quantum operations within the corresponding quantum circuit, employing topological sorting on the modified DAG enables us to establish a viable execution sequence. This sequence is then utilized to rearrange the circuit instruction list. Since our actions solely involve rewiring the original quantum circuit, the resulting compiled circuit maintains its equivalence to the static circuit. Ultimately, it's worth noting that all the conditions mentioned remain independent of whether we are working with a DAG or a simplified DAG representation.

\end{proof}

\subsection{Prop. \ref{prop:reducibility from dag}}
\begin{proof}
If the simplified DAG $G = (R, T, E)$ is a complete bipartite graph, then any terminal is reachable from any root. Within such a graph, the inclusion of any additional edge from a terminal to a root necessarily introduces a directed cycle, thus violating the conditions. Therefore, the circuit width can not be reduced through qubit reuse. Conversely, if the simplified DAG is not complete, it indicates the absence of connections between certain roots and terminals. More specifically, there exists at least one terminal $t$ that is not reachable from a root $r$. In this case, the directed edge $(t, r)$ can be added to the graph, which does not introduce any cycles. Thereby, the corresponding circuit can be reduced to a circuit with a smaller width.
\end{proof}

\subsection{Prop. \ref{prop:reducibility through DFS}}
\begin{proof}
Note that this approach combines three parts.
First, we employ Algorithm \ref{alg:Converting static quantum circuit to DAG} to derive the DAG representation of the given static circuit, which exhibits a time complexity of $O(m)$. Then, we convert the DAG to a simplified DAG by using DFS. For each root-terminal pair $(r, t)$ in the DAG, DFS algorithm is applied to explore the existence of a path from $r$ to $t$. Let $|R|$, $|T|$, $|V|$, and $|E|$ denote the number of roots, terminals, vertices, and edges in the DAG, respectively. It is clear that $|R| = |T| = n$, and $|V| = m$, $|E| = O(m)$ (each instruction introduces at most two edges to DAG). Since running DFS once in such a graph takes $O(|V| + |E|) = O(m)$ steps, and the number of trials we need is $|R| \times |T| = n^2$. So this part takes at most $O(m n^2)$ steps. Finally, after obtaining the reachability relation between all roots and terminals, we get the biadjacency matrix of the simplified DAG, which can be used to conclude the reducibility within $n^2$ steps. Therefore, the overall time complexity required to determine the reducibility of a given static circuit is $O(m) + O(m n^2) + O(n^2)$, with the dominant term being $O(m n^2)$.
\end{proof}

\subsection{Prop. \ref{prop:reducibility from qubit reachability}}
\begin{proof}
Note that the sequence of qubit relations is equivalent to a subset of instructions involving interlacing qubits. This equivalence can further be represented as a directed path in the DAG representation, connecting the $i$-th root to the $j$-th terminal. Thus, any two qubits in the quantum circuit are mutually reachable if and only if a directed path exists in the DAG from any root to any terminal. This condition is equivalent to stating that the simplified DAG is complete. Finally, the proof is concluded by referring to Proposition \ref{prop:reducibility from dag}.
\end{proof}

\subsection{Prop. \ref{prop:reducibility through structure}}
\begin{proof}
To determine the reducibility of a static circuit, it is necessary to analyze the reachability relation for each qubit. These reachability relations are stored in a list named the \emph{ReachableSets}, which has a length of $n$, representing the circuit's width, with the $i$-th entry containing the reachable set $Q_i$. This initial process involves $O(n)$ steps.
Subsequently, a `foreach' loop is executed over the `StaticCircuit' list, which comprises $m$ elements. Throughout this loop, all operations, except the union set calculation, exhibit a constant time complexity of $O(1)$ as they are independent of the input size. However, the computational complexity for uniting two sets, $S_1$ and $S_2$, amounts to $O(|S_1| + |S_2|)$. It is worth noting that the reachable set of a qubit contains a maximum of $n$ elements, thus the union calculation is bounded by $O(n)$. Consequently, the overall time complexity of the `foreach' loop is $O(mn)$.
In the final step, which determines reducibility, another `foreach' loop is executed, iterating over the \emph{ReachableSets}, incurring an additional $O(n)$ steps. In summary, the time complexity of Algorithm \ref{alg:reducibility qubit reachability} is $O(mn) + O(n) + O(n)$, with the dominant factor remaining $O(mn)$.
\end{proof}

\subsection{Prop. \ref{prop: biadjacency matrix relation of composite circuit}}
\begin{proof}
Let $G_m$ be the simplified DAG of circuit $\cC_m$. Then the $(i, j)$ entry of $B(\cC_m)$ equals to one if and only if there is an edge from root $r_i^m$ to terminal $t_j^m$ in graph $G_m$. To study the simplified DAG of the composite circuit $\cC_1 \circ \cC_2$, we can connect all terminals $t_k^1$ in graph $G_1$ with their corresponding roots $r_k^2$ in graph $G_2$. Therefore, a path from $r_i^1$ to $t_j^2$ exists if and only if there is an edge from $r_i^1$ to $t_k^1$ in graph $G_1$, followed by an edge from $r_k^2$ to $t_j^2$ in graph $G_2$ for some intermediate $k$. This condition can be represented as $\vee_{k=1}^n (B(\cC_1)_{ik} \wedge B(\cC_2)_{kj})$, which is exactly the $(i, j)$ entry of the Boolean matrix product $B(\cC_1) \odot B(\cC_2)$.
\end{proof}

\begin{remark}
Consider a quantum circuit $\mathcal{C}$ comprising $m$ two-qubit gates. Denote $B_i$ as the biadjacency matrix of the subcircuits, each containing only a single two-qubit gate. Then the biadjacency matrix of $\cC$ is given by $B = B_1 \cdots B_m$. Since $B_i^\top= B_i$, we get
\begin{align}
    B = (B_mB_{m-1}\cdots B_1)^T.
\end{align}
This corresponds to the concept of a dual circuit as presented in~\cite{decross2022qubitreuse}. In this context, we reverse the sequence of gates and exchange the roles of state preparation and measurement. The equality mentioned above implies that the dual circuit shares the same biadjacency matrix as the original circuit. Therefore, the compilation strategies for both circuits can be the same.  
\end{remark}

\subsection{Prop. \ref{prop:reducibility through matrix}}
\begin{proof}
First, a `foreach' loop is executed over the \textit{StaticCircuit} list, encompassing $m$ elements. Throughout this loop, all operations, except for the calculation of the entrywise OR, exhibit a constant time complexity of $O(1)$, as they are independent of input size. Since each column of $B$ contains $n$ elements, the entrywise OR takes $O(n)$ steps. Consequently, the cumulative time complexity of the `foreach' loop equates to $O(mn)$.
The last part that determines if a Boolean matrix is all-one matrix takes $O(n^2)$ steps. In summary, the overall time complexity is $O(mn) + O(n^2)$, with the dominant factor remaining $O(mn)$ (typically, $m$ is larger than $n$).
\end{proof}

\subsection{Prop. \ref{prop: subcircuit reducibility}}

The following result explores the connection between the reducibility of a quantum circuit and its subcircuits.

\begin{proposition}\label{prop: subcircuit reducibility}
A static quantum circuit is irreducible if it contains an irreducible subcircuit. The reverse direction is not true. That is, there exists an irreducible static quantum circuit such that any of its subcircuit is reducible.
\end{proposition}

\begin{figure}[!htb]
    \centering
    \includegraphics[width=0.4\textwidth]{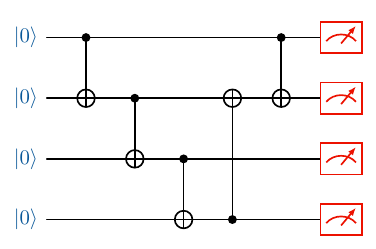}
    \caption{(Color online) An irreducible static quantum circuit such that any of its subcircuit is reducible.}
    \label{fig:counter example}
\end{figure}

\begin{proof}
Note that any quantum circuit $\cC$ can be seen as a composition of a sequence of subcircuits with a single quantum gate $\cC = \cC_1 \circ \cdots \circ \cC_M$. Let $B_i$ be the biadjacency matrix of the simplified DAG of $\cC_i$. The circuit $\cC$ has an irreducible subcircuit implies that there exists $\cC_{i_1}\circ \cdots \circ \cC_{i_m}$ such that $\bigodot_{j=1}^m B_{i_j} = J$. Then this implies $\bigodot_{i=1}^M B_i \geq  \bigodot_{j=1}^m B_{i_j} = J$, where the inequality signifies that each entry of the left matrix is greater than or equal to the corresponding entry in the right matrix. This proves the first statement.
The quantum circuit in Figure~\ref{fig:counter example} is an irreducible circuit with $4$ qubits and $5$ CNOTs. However, the removal of any CNOT in this circuit will result in a reducible circuit.
\end{proof}

\subsection{Prop. \ref{prop:compilation via BIP}}
\begin{proof}
Let $G = (R, T, E)$ be the simplified DAG of the circuit and $R = \{r_0, ..., r_{n-1}\}$, $T = \{t_0, ..., t_{n-1}\}$. Then the adjacency matrix of the directed graph $G$, denoted as $A(G)$, is a $2n \times 2n$ block anti-diagonal matrix, where the first $n$ rows/columns correspond to the $n$ roots, and the last $n$ rows/columns correspond to the $n$ terminals. Since all edges in the original bipartite graph are pointing from roots to terminals, only the $n \times n$ submatrix in the upper right corner has non-zero elements. Therefore, the adjacency matrix can be written in the form:
\begin{equation}
    A(G) = 
    \begin{pmatrix}
    O_n & B\\
    O_n & O_n
    \end{pmatrix}
\end{equation}
where $B$ is the biadjacency matrix of the bipartite graph. 

Note that the absence of directed edge from $r_i$ to $t_j$ in the original bipartite graph $G$ indicates a candidate edge from $t_j$ to $r_i$ in the modified DAG. All these candidate edges are pointing from terminals to roots, therefore can be represented by the $n \times n$ submatrix $\bar{B}$ in the lower left corner of the adjacency matrix $A(G)$. The submatrix $\bar{B}$, can be obtained by:
\begin{equation}
    \bar{B} = ^\neg (B^{\top})
\end{equation}
which is the logical NOT of the transpose of matrix $B$.

We use a $n \times n$ Boolean matrix $F$ to represent all added edges, where the entry $F_{ij} = 1$ if a directed edge from $t_i$ to $r_j$ is added to the graph, and $0$ otherwise. Suppose that $G^{\prime}$ is the modified DAG, then the adjacency matrix of $G^{\prime}$ can be written as:
\begin{equation}
    A(G^{\prime}) = 
    \begin{pmatrix}
        O_n & B \\
        F & O_n
    \end{pmatrix}
\end{equation}

The constraints imposed in Theorem~\ref{thm: circuit compilation via graph manipulation} when adding edges to $G$ can be translated into the following constraints on matrix $F$:
\begin{enumerate}
    \item The objective of maximizing the number of added edges is equivalent to maximizing the sum of all elements in matrix $F$:
    \begin{align}
        \max \sum_{i,j=0}^{n-1} F_{ij} 
    \end{align}
    \item All added edges should be selected from candidate edges, which means elements of matrix $F$ should be selected from elements of matrix $\bar{B}$ and can be expressed as:
    \begin{align}
        F_{ij} \leq \bar{B}_{ij}, \forall i,j \in \{0, 1, ..., n-1\}
    \end{align}
    \item The added edges must not share common vertices, which implies the sum of each row/column of $F$ can not exceed one:
    \begin{align}
        & \sum_{j=1}^{n} F_{ij} \le 1, \, \forall i \in \{0, 1, ..., n-1\}\\
        & \sum_{i=1}^{n} F_{ij} \le 1, \, \forall j \in \{0, 1, ..., n-1\}.
    \end{align}
    \item After incorporating these edges, the graph $G^{\prime}$ should remain acyclic. According to~\cite[Theorem 9-17]{deo2017graph}, this requirement is equivalent to the adjacency matrix of $G'$ being nilpotent.
\end{enumerate}
Finally, since the optimal value $\alpha$ gives the maximum number of terminal-root pairs in the modified DAG, the remaining number of roots is $n-\alpha$, which is the width of the compiled circuit. This completes the proof.
\end{proof}

\subsection{Prop. \ref{prop:time complexity of mrv algorithm}}
\begin{proof}
The MRV heuristic algorithm comprises three key steps. Initially, Algorithm~\ref{alg:Converting static quantum circuit to DAG} and Algorithm~\ref{alg:reducibility through Boolean matrix} are applied to derive the DAG and the biadjacency matrix of the simplified DAG of the input static circuit and calculate the candidate matrix, which exhibits a time complexity of $O(m) + O(mn) + O(n^2)= O(mn)$. Subsequently, our MRV heuristic is applied on the candidate matrix of size $n \times n$ to explore a feasible assignment. Since the maximum number of added edges is $n-1$, the `while' loop iterates at most $n-1$ times. Within the loop, several operations are adapted to identify a candidate edge, where the row summation of the candidate matrix takes $O(n^2)$ time, the candidate terminal identification takes $O(n)$ times, while the candidate root identification involves traversing all candidate roots and summing of the corresponding columns, resulting in a worst-case time complexity of $O(n^2)$. Following this, the candidate matrix needs to be updated. Given that sets $R$ and $T$ contain a maximum of $n-1$ elements, this encompasses the update of at most $(n-1)^2$ entries, which can be completed in $O(n^2)$ time. Consequently, the time complexity of the entire `while' loop is $O(n^3) + O(n^2) + O(n^3) + O(n^3) = O(n^3)$. Upon obtaining a list of added edges, Algorithm \ref{alg:modified DAG to dynamic circuit} is employed to compile the input static circuit, which requires $O(mn)$ time. In summary, the total time complexity of the MRV heuristic algorithm is calculated as: $O(mn) + O(n^3) + O(mn) = O(mn + n^3)$.
\end{proof}

\subsection{Prop. \ref{prop:time complexity of greedy algorithm}}
\begin{proof}
Similar to the MRV heuristic algorithm, both the initial step to get the candidate matrix and the final step to compile the input static circuit has a time complexity of $O(mn)$, while the primary `while' loop iterates at most $n-1$ times. However, the main difference between the greedy heuristic and the MRV heuristic lies in the candidate edge identification process. In the worst-case scenario, where the biadjacency matrix is an identity matrix and the candidate matrix has $(n^2 - n)$ non-zero entries. Each of these entries triggers the update of the candidate matrix and the summation of a matrix's elements, which both exhibit a time complexity of $O(n^2)$. As a result, the total time complexity of the scoring step becomes $O(n^4)$. Subsequently, both the identification of the entry with maximum scoring and the corresponding candidate matrix update can be completed within $O(n^2)$ time. In summary, the total time complexity of the greedy heuristic algorithm \ref{alg:greedy} is evaluated as $O(mn) + O(n^5) + O(n^3)$, with the dominating term being $O(mn) + O(n^5)$.
\end{proof}

\section{Optimal compilations for important quantum circuits}
\label{app: Optimal compilations for important quantum circuits}

\subsection{Frequently used quantum algorithms}

\subsubsection{Deutsch-Jozsa algorithm}
The Deutsch-Jozsa algorithm~\cite{deutsch1992rapid} was the pioneering example of a quantum algorithm that outperforms classical counterparts. It exemplifies the advantages of using quantum computing for specific problems. The algorithm aims to determine whether a Boolean function $f: \{0,1\}^n \to \{0,1\}$ is balanced or constant. This algorithm is represented by the following $(n+1)$-qubit quantum circuit. The initial state has the first $n$ qubits set to $\ket{0}$, and the last qubit initialized as $\ket{1} = X \ket{0}$. A Hadamard gate is subsequently applied to each qubit. Following this, a quantum oracle $U_f$ maps $\ket{x}\ket{y}$ to $\ket{x} \ket{y \oplus f(x)}$. Finally, Hadamard gates are reapplied to the first $n$ qubits, and all qubits are measured in the computational basis.

\begin{figure}[H]
    \centering
    \includegraphics[scale=0.8]{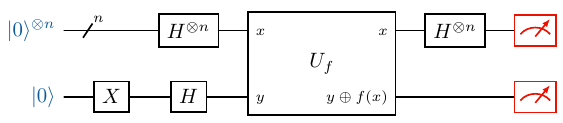}
    \caption{(Color online) The Deutsch-Jozsa algorithm for determing an $n$-bit Boolean function.}
    \label{fig:Deutsch-Jozsa algorithm}
\end{figure}

If the Boolean function $f$ is constant, the quantum oracle $U_f$ can be implemented using only single-qubit gates, making it trivially reducible to a quantum circuit with $1$ qubit. In the case where the Boolean function $f$ is balanced, multiple quantum circuit implementations are possible. For instance, the quantum oracle $U_f$ can be realized using a quantum circuit depicted in Figure~\ref{fig:Deutsch-Jozsa algorithm oracle}, that is, regardless of the single-qubit gates, applying a CNOT gate for each of the first $n$ qubits, with the $(n+1)$-th qubit as the target.

\begin{figure}[H]
    \centering
    \includegraphics[scale=0.8]{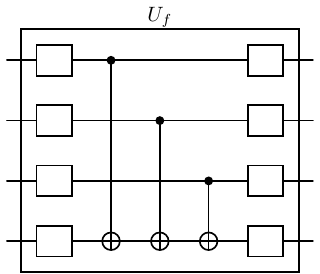}
    \caption{A quantum circuit implementation of a balanced oracle for determining a $3$-bit Boolean function. The empty box represents a single-qubit gate.}
    \label{fig:Deutsch-Jozsa algorithm oracle}
\end{figure}

The following result gives an optimal compilation of the Deutsch-Jozsa algorithm.

\begin{proposition}[Deutsch-Jozsa]\label{prop:Deutsch-Jozsa}
The quantum circuit of the Deutsch-Jozsa algorithm for determining an $n$-bit Boolean function $f$  contains $n+1$ qubits ($n \geq 2$). This circuit is always reducible. If $f$ is constant, the quantum circuit can be reduced to a dynamic quantum circuit with $1$ qubit. Otherwise, it can be reduced to a dynamic quantum circuit with $2$ qubits. In this case, the corresponding optimal solution for optimization~\eqref{eq: exact optimization} can be taken at $\sum_{i=0}^{n-2} E^{n+1}_{i,i+1}$.
\end{proposition}
\begin{proof}
If the function is constant, the quantum circuit only contains single-qubit gate. So it can be trivially reduced to a quantum circuit with $1$. Now we prove the other case.
For the quantum circuit implementation in Figure~\ref{fig:Deutsch-Jozsa algorithm oracle}, the biadjacency matrix of the simplified DAG is given by
\begin{align}
    B = \sum_{i=0}^{n} \;\;\sum_{j= i}^{n} E^{n+1}_{i,j} + \sum_{i=0}^{n-1} E^{n+1}_{n,i}.
\end{align}
So the candidate matrix is
\begin{align}
    C = \sum_{i=0}^{n-2} \;\;\sum_{j=i+1}^{n-1} E^{n+1}_{i,j}.
\end{align}
For the optimal compilation, we need to select the maximum number of $1$ elements in such a matrix under the conditions stated in Proposition \ref{prop:compilation via BIP}. Without the nilpotent condition, the maximum number of $1$ elements that simultaneously satisfies conditions~\eqref{eq: submatrix condition},~\eqref{eq: row sum condition} and~\eqref{eq: col sum condition} is $n-1$. Let us choose
\begin{align}
    F = \sum_{i=0}^{n-2} E^{n+1}_{i,i+1},
\end{align} 
with the total sum of $n-1$. Then we can easily check that the adjacency matrix with such a selection is indeed nilpotent, making it an optimal solution in~\eqref{eq: exact optimization}. Finally, by Algorithm~\ref{alg:modified DAG to dynamic circuit}, the compiled circuit has a circuit width $2$.
\end{proof}

\subsubsection{Bernstein-Vazirani algorithm}
The Bernstein-Vazirani algorithm~\cite{bernstein1993quantum} can be considered an extension of the Deutsch-Jozsa algorithm, demonstrating the advantages of using a quantum computer as a computational tool for more complex problems than the Deutsch-Jozsa problem. Instead of distinguishing between two different classes of functions, it is designed to recover a string encoded within a function. Specifically, when provided with an oracle implementing a function $f: \{0,1\}^n \to \{0,1\}$ in which $f(x)$ is promised to be the dot product between $x$ and a secret string $s \in \{0,1\}^n$ modulo $2$, the objective is to determine the value of $s$. The Bernstein-Vazirani algorithm employs the same $(n+1)$-qubit quantum circuit as depicted in Figure~\ref{fig:Deutsch-Jozsa algorithm}. Furthermore, the quantum oracle can be realized by applying a CNOT gate to the corresponding qubit with the last qubit as the target for each bit in the string $s$ that equals one. For example, if the bit string is $s = 101$, the quantum oracle is implemented as

\begin{figure}[H]
    \centering
    \includegraphics[scale=0.8]{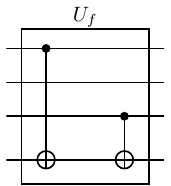}
    \caption{A quantum circuit implementation of an oracle for $f(x) = (x_0 + x_2)\mod 2$.}
    \label{fig:Bernstein-Vazirani algorithm oracle}
\end{figure}

The following result gives an optimal compilation of the Bernstein-Vazirani algorithm.

\begin{proposition}[Bernstein-Vazirani]\label{prop:Bernstein-Vazirani}
The quantum circuit of the Bernstein-Vazirani algorithm for determining an $n$-bit Boolean function $f$ contains $n+1$ qubits ($n \geq 2$). This circuit is always reducible. If $f$ is constant, the quantum circuit can be reduced to a dynamic quantum circuit with $1$ qubit. Otherwise, it can be reduced to a dynamic quantum circuit with $2$ qubits. In this case, the corresponding optimal solution for optimization~\eqref{eq: exact optimization} can be taken at $\sum_{i=0}^{n-2} E^{n+1}_{i,i+1}$.
\end{proposition}
\begin{proof}
If the secret string is all zero, then the function is constant and the quantum oracle contains no multi-qubit gates. This makes the circuit trivially reducible to one qubit. Otherwise, the secret string contains at least one non-zero element. So the quantum oracle has at least one CNOT gate. Without loss of generality, we can consider the case in which the secret string is all one. This is because the quantum circuit for all the other cases is a subcircuit of this one. In this case, the quantum circuit structure is the same as the Deutsch-Jozsa algorithm when omitting the single-qubit gates. So the rest of proof follows exactly the same as the proof of Proposition~\ref{prop:Deutsch-Jozsa}.
\end{proof}

\subsubsection{Simon's algorithm}
Simon's algorithm~\cite{simon1997power} marked a significant milestone as the first quantum algorithm to exhibit an exponential speed-up when compared to the best classical algorithm for a specific problem. This breakthrough served as a foundational inspiration for a family of quantum algorithms centered around the quantum Fourier transform, including the renowned Shor's factoring algorithm. In the problem Simon's algorithm addresses, we are provided with an oracle that is guaranteed to have either a one-to-one mapping (it maps a unique output to every input) or a two-to-one mapping (it maps two different inputs to a single unique output). The nature of this two-to-one mapping is determined by a secret bitstring $s$, where the function $f(x)$ equals $f(y)$ if and only if $y = x \oplus s$. The primary objectives are twofold: first, to decide whether the function $f$ is one-to-one or two-to-one; and second, in the event it is determined to be two-to-one, to unveil the secret bitstring $s$. The quantum circuit of Simon's algorithm is given by

\begin{figure}[H]
    \centering
    \includegraphics[scale=0.8]{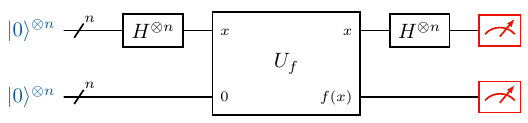}
    \caption{(Color online) The Simon's algorithm for determining an $n$-bit Boolean function.}
    \label{fig:Simon's algorithm}
\end{figure}

There are many possible ways of implementing a desired two-to-one function. Here we can consider a specific choice that
\begin{align}
    f(x) = \begin{cases}
        x, & x_j = 0,\\
        x \oplus s, & x_j = 1,
    \end{cases}
\end{align}
where $x_j$ is the $j$-th bit of $x$ and $j$ is the smallest index where the bit of $s$ equals to one. For example, if $s = 011$, then $j = 1$. In this case, the quantum oracle $U_f$ can be implemented by first performing a CNOT gate between $i$-th and $n+i$-th qubits for every $i \in \{0,1,\cdots, n-1\}$ and then performing a CNOT gate on the $j$-th qubit and the $n+k$-th qubit whenever the $k$-th bit of $s$ equals to one. For instance, given secret string $s = 011$, the quantum oracle $U_f$ can be implemented by

\begin{figure}[H]
    \centering
    \includegraphics[scale=0.8]{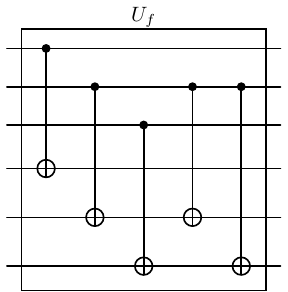}
    \caption{A quantum circuit implementation of an oracle for secret string $s = 011$.}
    \label{fig:Simon's algorithm oracle}
\end{figure}

\begin{proposition}[Simon]\label{prop:Simon}
The quantum circuit of the Simon's algorithm for determining an $n$-bit function $f$ contains $2n$ qubits ($n \geq 2$). This circuit is always reducible. It can be reduced to a dynamic quantum circuit with $3$ qubits. In this case, the corresponding optimal solution for optimization~\eqref{eq: exact optimization} can be taken at $E_{0,0}^2 \otimes \sum_{i=1}^{n-2} E_{i,i+1}^n + E_{1,1}^2 \otimes \sum_{i=0}^{n-2} E^{n}_{i,i+1}$.
\end{proposition}
\begin{proof}
Note that any quantum circuit is a subcircuit of the one given by a secret string with all-one values. So we can restrict our consideration in the later case, where a CNOT is applied between $i$-th and $i+n$-th qubits for every $i\in \{0,1,\cdots, n-1\}$ and a CNOT is applied between $0$-th and $i+n$-th qubits for every $i\in \{0,1,\cdots, n-1\}$. In this case, the biadjacency matrix of the simplified DAG is given by
\begin{align}
    B = (E_{0,0}^2 + E_{1,0}^2) \otimes \left(I_n + \sum_{i=1}^{n-1} E^n_{i,0}\right) + (E_{0,1}^2 + E_{1,1}^2) \otimes \sum_{i=0}^{n-1} \sum_{j=i}^{n-1} E^n_{i,j}.
\end{align}
So the candidate matrix is
\begin{align}
    C = (E_{0,0}^2 + E_{0,1}^2) \otimes \left( \sum_{i=1}^{n-2} \sum_{j=i+1}^{n-1} E^n_{i,j} +  \sum_{i=1}^{n-1} \sum_{j=0}^{i-1} E^n_{i,j} \right) + (E_{1,0}^2 + E_{1,1}^2) \otimes \sum_{i=0}^{n-2} \sum_{j=i+1}^{n-1} E^n_{i,j}.
\end{align}
Without the nilpotent condition, we can choose
\begin{align}\label{eq: feasible solution simon}
    F = E_{0,0}^2 \otimes \sum_{i=1}^{n-2} E_{i,i+1}^n + E_{1,1}^2 \otimes \sum_{i=0}^{n-2} E^{n}_{i,i+1},
\end{align} 
with the total sum of $2n-3$. Then we can easily check that the adjacency matrix with such a selection is indeed nilpotent, making it a feasible solution in~\eqref{eq: exact optimization}. Finally, by Algorithm~\ref{alg:modified DAG to dynamic circuit}, the compiled circuit has circuit width $3$. Note that if we take $n = 2$, the quantum circuit is a subcircuit for any quantum circuit with larger $n$. In the case of $n = 2$, we can easily enumerate all possible compilation schemes and conclude that it requires at least $3$ qubits in the compiled quantum circuit. This implies that any larger circuit will also require at least $3$ qubits and concludes the optimality of our feasible solution in~\eqref{eq: feasible solution simon}.
\end{proof}

\subsubsection{Quantum Fourier transform}\label{sec: Quantum Fourier transform}
The Fourier transform is a fundamental concept in classical computing with various applications, including signal processing, data compression, and complexity theory. In quantum computing, the quantum Fourier transform (QFT) serves as the quantum analog of the discrete Fourier transform and operates on the amplitudes of a quantum wavefunction. The QFT plays a pivotal role in numerous quantum algorithms, with its most notable appearances in Shor's factoring algorithm and quantum phase estimation. The quantum Fourier transform acts on a quantum state $\ket{x} = \sum_{j=0}^{2^n-1} x_j \ket{j}$ and maps it to another quantum state $\ket{y} = \sum_{k=0}^{2^n-1} y_k \ket{k}$ where $y_k = \frac{1}{\sqrt{2^n}} \sum_{j=0}^{2^n-1} x_j e^{2\pi i jk/2^n}$. Let $R_k = \ket{0}\bra{0} + e^{2\pi i/2^k} \ket{1}\bra{1}$. Then the quantum Fourier transform can be implemented by the following quantum circuit. For every $j\in \{0,1,\cdots,n-1\}$, apply a Hadamard gate on the $j$-th qubit and then apply a controlled-$R_k$ gate between the $j$-th qubit and the $(j+k-1)$-th qubit for every $k\in \{2,\cdots, n-j\}$. Finally, apply SWAP gates to reverse the order of the qubits. For instance, the quantum circuit for $n = 4$ is illustrated below.

\begin{figure}[!htb]
    \centering
    \includegraphics[scale=0.8]{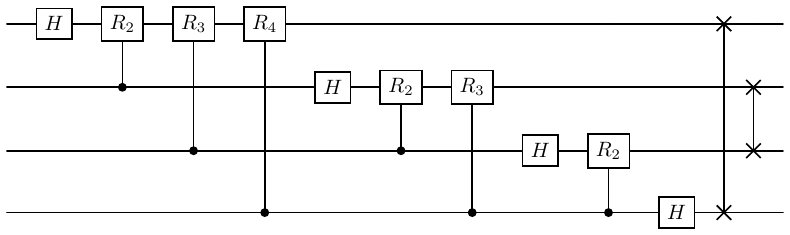}
    \caption{A quantum circuit implementation of the quantum Fourier transform for $n = 4$.}
    \label{fig:Quantum Fourier transform}
\end{figure}

Since there is a two-qubit gate on any two qubits in the quantum circuit of the quantum Fourier transform, this quantum circuit is clearly irreducible by Proposition~\ref{prop:reducibility from qubit reachability}.

\subsubsection{Quantum phase estimation}\label{sec: Quantum phase estimation}
The Fourier transform plays a crucial role in a general procedure known as phase estimation, which serves as a fundamental component in many quantum algorithms. Consider a unitary operator $U$ that possesses an eigenvector $\ket{u}$ with an associated eigenvalue of $e^{2\pi i \theta}$, where the precise value of $\theta$ remains unknown. The objective of the phase estimation algorithm is to estimate the value of $\theta$. To conduct this estimation, we assume the availability of oracles capable of preparing the state $\ket{u}$ and performing controlled-$U^{2^j}$ operations, with $j$ being a non-negative integer. The quantum phase estimation procedure involves two registers. The first register comprises $t$ qubits, initially set to $\ket{0}$. The choice of the value of $t$ depends on two factors: the level of precision required for estimating $\theta$, and the desired probability of a successful phase estimation procedure. The second register begins in the state $\ket{u}$ and accommodates a number of qubits sufficient to store $\ket{u}$. Phase estimation unfolds in two stages. First, the circuit commences by applying a Hadamard transform to the first register, followed by a series of controlled-$U$ operations on the second register, each involving $U$ raised to successive powers of two. The second stage of phase estimation involves the application of the inverse quantum Fourier transform. An example is provided in Figure~\ref{fig:Quantum phase estimation}.

\begin{figure}[H]
    \centering
    \includegraphics[scale=0.8]{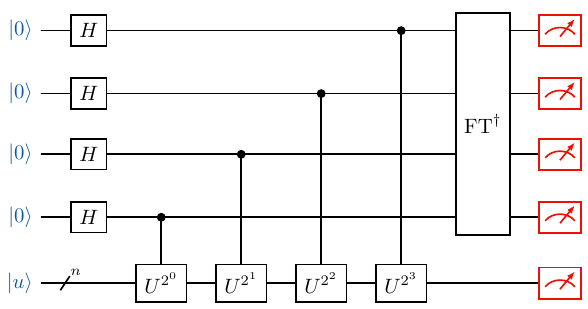}
    \caption{(Color online) A quantum circuit implementation of the quantum phase estimation for $t = 4$.}
    \label{fig:Quantum phase estimation}
\end{figure}

The original quantum circuit for quantum phase estimation necessitates the utilization of $t+n$ qubits. It is important to observe that the reducibility of a quantum circuit is equivalent to the reducibility of its dual circuit, which is essentially the circuit with a reversed gate sequence. In the context of the phase estimation algorithm, the dual circuit is structured as follows: it begins with the quantum Fourier transform applied to the first $t$ qubits, followed by a sequence of controlled-$U$ operations conducted on each qubit within the initial $t$ qubits and the subsequent $n$ qubits. The quantum Fourier transform gives the biadjacency matrix as
\begin{align}
B_{-1} = \begin{pmatrix}
    J_t & O \\
    O & I_n
\end{pmatrix}
\end{align}
For every controlled operation acting on the $k$-th qubit and the last $n$ qubits, the biadjacency matrix is given by
\begin{align}
    B_k = I_{t+n} + \sum_{\substack{i\neq j\\ i,j \in \{k, t,t+1,\cdots,t+n-1\}}} (E_{i,j}^{t+n} + E_{j,i}^{t+n}).
\end{align}
According to Proposition~\ref{prop: biadjacency matrix relation of composite circuit},  the biadjacency matrix for the dual circuit of quantum phase estimation is given by
\begin{align}
    B_{-1}\bigodot \left(\bigodot_{k=0}^{t-1} B_k\right),
\end{align}
which is an all-one matrix $J_{t+n}$. So the quantum circuit for quantum phase estimation is irreducible.

\subsubsection{Shor's algorithm}\label{sec: Shor's algorithm}
Shor's algorithm~\cite{shor1994algorithms} is renowned for its ability to factor integers in polynomial time. This algorithm holds particular significance because the best-known classical algorithm for factoring requires more than polynomial time, and RSA (Rivest–Shamir–Adleman), the widely-used cryptographic protocol, relies on the assumption that factoring large integers is computationally infeasible. The factoring problem is equivalent to order-finding problem, that is, given positive integers $x$ and $N$ with $x <N$ with no common factors, find the smallest positive integer $r$ such that $x^r = 1 \mod N$. The quantum algorithm for order-finding is just the phase estimation algorithm applied to the unitary operator $U\ket{y} = \ket{xy \mod N}$. Since the quantum circuit for quantum phase estimation is irreducible in general, the quantum circuits for order-finding and consequently Shor's algorithm are also irreducible. 

\subsubsection{Grover's algorithm}\label{sec: Grover's algorithm}
Grover's algorithm~\cite{grover1996fast}, also known as the quantum search algorithm, refers to a quantum algorithm for unstructured search that finds with high probability the unique input to a black box function that produces a particular output value, using just $O({\sqrt  {N}})$ evaluations of the function, where $N$ is the size of the function's domain.
Grover's algorithm consists of three main algorithms steps as illustrated in Figure~\ref{fig:Grover's algorithm}: state preparation, the oracle $U_f$, and the diffusion operator $U_s$. The first step, state preparation, involves creating the search space, encompassing all possible cases where the solution might exist. The oracle serves the crucial role of marking the correct answers we seek within the search space. It identifies the desired solutions, making them distinguishable for measurement in the final steps of the algorithm. The oracle and diffusion operator are iteratively applied in the quantum circuit, amplifying the presence of the marked solutions until they stand out for measurement.

\begin{figure}[!htb]
    \centering
    \includegraphics[scale=0.8]{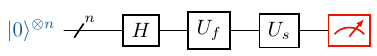}
    \caption{(Color online) The Grover's algorithm.}
    \label{fig:Grover's algorithm}
\end{figure}

In particular, the diffusion operator can be implemented by a multi-level controlled-CNOT gate wrapped by single-qubit gates on both sides. An illustration of the diffusion operator with $n = 4$ qubits is given below. Since the quantum circuit involves an $n$-qubit gate, it is clearly irreducible by Proposition~\ref{prop: biadjacency matrix relation of composite circuit}.

\begin{figure}[!htb]
    \centering
    \includegraphics[scale=0.8]{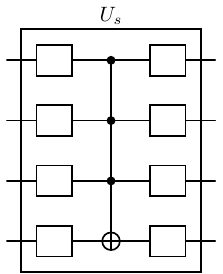}
    \caption{A quantum circuit implementation of the diffusion operator for $4$-qubit Grover's algorithm. Each empty box represents a single-qubit gate.}
    \label{fig:Grover's algorithm diffusion operator}
\end{figure}

\subsubsection{Quantum counting algorithm}
\label{sec: Quantum counting}

Quantum counting algorithm~\cite{brassard1998quantum} is a quantum algorithm designed to efficiently determine the number of solutions to a given search problem. While Grover's algorithm focuses on finding a specific solution within the oracle, the Quantum counting algorithm provides us with the ability to count how many solutions are present. The quantum circuit underlying the quantum counting algorithm essentially performs quantum phase estimation on the Grover iterator $U_sU_f$. As discussed in Section~\ref{sec: Quantum phase estimation}, the quantum circuit for quantum counting is inherently irreducible due to the irreducibility of quantum phase estimation.

\subsubsection{Quantum Ripple Carry Adder}
Here we focus on the Quantum Ripple Carry Adders initially proposed in~\cite{chakrabarti2008designing}. These circuits are composed of $3k + 1$ qubits where the additional qubit is used to store the carry bit and rely on CNOT and Toffoli gates as foundational components. The implementation of a 2-bit quantum adder circuit is illustrated in Figure~\ref{fig:2-bit quantum adder implementation}.

\begin{figure}[H]
    \centering
    \includegraphics[scale=0.8]{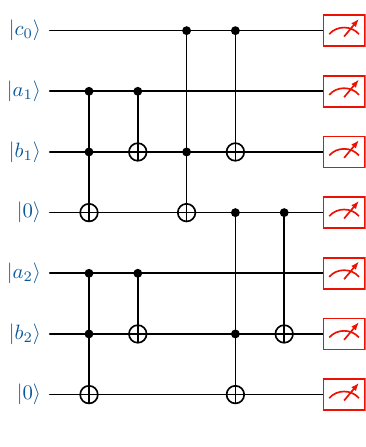}
    \vspace{-3mm}
    \caption{(Color online) The implementation of a 2-bit quantum ripple carry adder circuit to calculate $a_2 a_1 + b_2 b_1$.}
    \label{fig:2-bit quantum adder implementation}
\end{figure}

The optimal compilation of quantum ripple carry adder with $n = 3k+1$ qubits is given in the following proposition.

\begin{proposition}[Quantum ripple carry adder.]\label{prop:adder}
A quantum ripple carry adder circuit with $n$ qubits is always reducible. The corresponding optimal solution for optimization~\eqref{eq: exact optimization} can be taken at $\sum_{i=0}^{n-5} E_{i, i+4}^n$ for general cases and $E_{1, 0}$ for $n=4$. Consequently, the circuit can be compiled to an equivalent dynamic circuit with $4$ qubits for $n \geq 5$ and $3$ qubits for $n=4$. 
\end{proposition}
\begin{proof}
For a quantum ripple carry adder circuit with $4$ qubits, the candidate matrix only has one non-zero entry, therefore the optimal compilation gives $3$ qubits. For a quantum adder circuit with $n = 3k + 1 > 4$ qubits, the biadjacency matrix of the simplified DAG is given by
\begin{align}
    B = \sum_{i=0}^3 \sum_{j=0}^{n-1} E_{i, j}^n + \sum_{m=1}^{k-1} \sum_{i=3m+1}^{3m+3} \sum_{j=3m}^{n-1} E_{i, j}^n - \sum_{j=0}^{k-1} \sum_{i=0}^{3j} E^n_{i, 3j+1}.
\end{align}
The candidate matrix can be further computed as
\begin{align}
    C = \sum_{i=1}^{k} \sum_{j=0}^{n - 3i - 1} E_{n - 3i, j}^n + \sum_{m=0}^{k-2} \sum_{i=3m}^{3m+2} \sum_{j=3m+4}^{n-1} E_{i, j}^n.
\end{align}

We need to select as many as 1 elements within the conditions stated in Proposition~\ref{prop:compilation via BIP}. Without taking into account the nilpotent condition, we can choose
\begin{align}
    F = \sum_{i=0}^{n-5} E_{i, i+4}^n + E_{n-3, l}
\end{align}
where $l \in [0, n-4]$. However, upon closer examination, it becomes evident that each of these candidate edges $E_{n-3, l}$ conflicts with the edge $E_{n-5, n-1}$ due to the nilpotent condition. Therefore, our selection is constrained to
\begin{align}
    F = \sum_{i=0}^{n-5} E_{i, i+4}^n
\end{align}
we can verify that the adjacency matrix with such a selection is indeed nilpotent, making it the optimal solution in~\eqref{eq: exact optimization}. Finally, by Algorithm~\ref{alg:modified DAG to dynamic circuit}, the compiled circuit has circuit width $4$. This can be understood as the quantum ripple carry adder employs Toffoli gates, therefore requiring a minimum of three qubits, while an additional qubit serving as a carry-over between distinct adders.

\end{proof}

\subsection{Ansatzs in quantum machine learning}
\subsubsection{Frequently used entanglement structures}
\label{subsubsec:entanglement reducibility}

Parameterized Quantum Circuits (PQCs) serve as foundational components in quantum machine learning. These circuits usually consist of alternating rotation and entanglement layers. In the rotation layers, single-qubit rotation gates are applied to all qubits. Subsequently, the entanglement layers employ a predefined strategy to create entanglement among the qubits using two-qubit gates. Each of these layers can be repeated multiple times to achieve the desired level of entanglement. Several commonly used entanglement layers include the following~\cite{sim2019expressibility}:

\begin{itemize}
    \item A linearly entangled layer of $n$ qubits means that each qubit $i$ is entangled with the subsequent qubit $i+1$ for $i \in \{0, 1, \cdots, n-2\}$. An example with $4$ qubits is given in Figure~\ref{fig:linearly entangled};
    
    \item A circularly entangled layer of $n$ qubits is a linearly entangled layer with an additional entangling gate between qubit $n-1$ and qubit $0$. An example with $4$ qubits is given in Figure~\ref{fig:circularly entangled};
    
    \item A pairwisely entangled layer of $n$ qubits consists of two sub-layers. In the first sub-layer, qubit $i$ is entangled with qubit $i+1$ for all even values of $i$, and in the second sub-layer, qubit $i$ is entangled with qubit $i+1$ for all odd values of $i$. An example with $4$ qubits is illustrated in Figure~\ref{fig:pairwisely entangled};

    \item A fully entangled layer of $n$ qubits means that each qubit $i$ for $i \in \{0, 1,\cdots, n-1\}$ is entangled with qubit $j$ for $j \in \{i, i+1, ..., n-1\}$. An example with $4$ qubits is illustrated in Figure~\ref{fig:fully entangled}.
\end{itemize}

\begin{figure}[ht]
    \centering
    \begin{subfigure}[b]{.38\textwidth}
    \centering
        \includegraphics[scale=0.8]{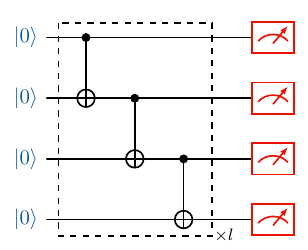}
        \caption{Linearly entangled layer.}
        \label{fig:linearly entangled}
    \end{subfigure}
    \begin{subfigure}[b]{.6\textwidth}
    \centering
        \includegraphics[scale=0.8]{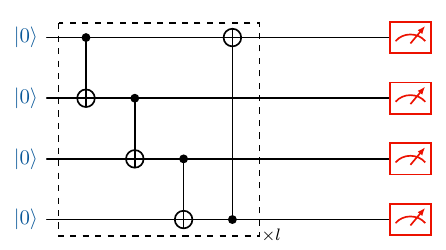}
        \caption{Circularly entangled layer.}
        \label{fig:circularly entangled}
    \end{subfigure}
    \begin{subfigure}[b]{.38\textwidth}
    \centering
        \includegraphics[scale=0.8]{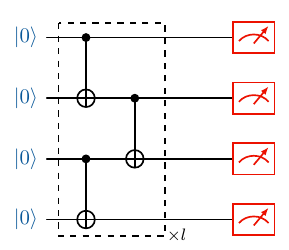}
        \caption{Pairwisely entangled layer.}
        \label{fig:pairwisely entangled}
    \end{subfigure}
    \begin{subfigure}[b]{.6\textwidth}
    \centering
        \includegraphics[scale=0.8]{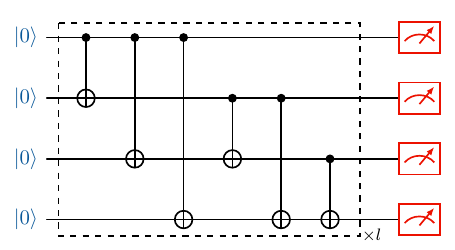}
        \caption{Fully entangled layer.}
        \label{fig:fully entangled}
    \end{subfigure}
    \caption{(Color online) Examples of four entanglement structures widely used in quantum machine learning. All single-qubit gates are omitted. Each layer in the dashed box can repeat $l$ times.}
    \label{fig:entanglement structures}
\end{figure}

Utilizing Proposition~\ref{prop: biadjacency matrix relation of composite circuit} and Theorem~\ref{prop:compilation via BIP}, we present the analysis of the reducibility and the optimal compilation for quantum circuits employing the four entanglement structures.

\begin{proposition}[Linearly entangled quantum circuit]\label{prop:linearly entangled}
A linearly entangled quantum circuit with $n \geq 2$ qubits and $l \geq 1$ linear layers is reducible if and only if $l \leq n-2$. The corresponding optimal solution for optimization~\eqref{eq: exact optimization} can be taken at $\sum_{i=0}^{n-l-2} E^{n}_{i,i+l+1}$. Consequently, the circuit can be compiled to an equivalent dynamic circuit with $l+1$ qubits.
\end{proposition}
\begin{proof}
For a linearly entangled quantum circuit with layer $l = 1$, the biadjacency matrix of the simplified DAG is given by
\begin{align}
    B = \sum_{i=0}^{n-1} \;\;\sum_{j=\max(i-1,0)}^{n-1} E^{n}_{i,j}.
\end{align}
For any integer $l \geq 1$, we can easily check that 
\begin{align}
    B^{\odot l} = \sum_{i=0}^{n-1} \;\;\sum_{j=\max(i-l,0)}^{n-1} E^{n}_{i,j}.
\end{align}
Therefore, $B^{\odot l}$ is an all-one matrix if and only if $l \geq n-1$. By Proposition~\ref{prop: biadjacency matrix relation of composite circuit}, the biadjacency matrix for the multi-layer quantum circuit is exactly given by $B^{\odot l}$. So the circuit is irreducible if and only if $l \geq n-1$, or equivalently, it is reducible if and only if $l \leq n-2$.

Moreover, we can compute the candidate matrix for the circuit with $l$ layers as
\begin{align}
    C = \sum_{i=0}^{n-l-2} \;\;\sum_{j=i+l+1}^{n-1} E^{n}_{i,j}.
\end{align}

For the optimal compilation, we need to select the maximum number of $1$ elements in such a matrix under the conditions in Proposition \ref{prop:compilation via BIP}. Without the nilpotent condition, the maximum number of $1$ elements that simultaneously satisfies conditions~\eqref{eq: submatrix condition},~\eqref{eq: row sum condition} and~\eqref{eq: col sum condition} is $n-l-1$. Let us choose
\begin{align}
    F = \sum_{i=0}^{n-l-2} E^{n}_{i,i+l+1},
\end{align} 
with the total sum of $n-l-1$. Then we can easily check that the adjacency matrix with such a selection is indeed nilpotent, making it an optimal solution in~\eqref{eq: exact optimization}. Finally, by Algorithm~\ref{alg:modified DAG to dynamic circuit}, the compiled circuit has circuit width $l+1$.
\end{proof}

\begin{proposition}[Circularly entangled quantum circuit]\label{prop:circularly entangled}
A circularly entangled quantum circuit with $n \geq 2$ qubits and $l \geq 1$ circular layers is reducible if and only if $n \geq 4$ and $l = 1$. The corresponding optimal solution for optimization~\eqref{eq: exact optimization} can be taken at $\sum_{i=1}^{n-3} E^{n}_{i,i+2}$. Consequently, the circuit can be compiled to an equivalent dynamic circuit with $3$ qubits.

\end{proposition}
\begin{proof}
The proof is similar to the proof of Proposition~\ref{prop:linearly entangled}.
For a circularly entangled quantum circuit with layer $l = 1$, the biadjacency matrix of the simplified DAG is given by
\begin{align}
    B = \sum_{i=0}^{n-1} \;\;\sum_{j=\max(i-2,0)}^{n-1} E^{n}_{i,(j+1\Mod n)}.
\end{align}
which is an all-one matrix if $n \leq 3$. Moreover, for any $n\geq 4$ we can check that $B^{\odot 2}$ becomes an all-one matrix. So the circuit is reducible if and only if $n\geq 4$ and $l = 1$.

Moreover, we can compute the candidate matrix for the circuit with $l=1$ layer as
\begin{align}
    C = \sum_{i=1}^{n-3} \sum_{j = i+2}^{n-1} E^{n}_{i,j}.
\end{align}

For the optimal compilation, we need to select the maximum number of $1$ elements in such a matrix under the conditions stated in Proposition \ref{prop:compilation via BIP}. Without the nilpotent condition, the maximum number of $1$ elements that simultaneously satisfies conditions~\eqref{eq: submatrix condition},~\eqref{eq: row sum condition} and~\eqref{eq: col sum condition} is $n-l-1$. Let us choose
\begin{align}
    F = \sum_{i=1}^{n-3} E^{n}_{i,i+2}.
\end{align}
with the total sum of $n - 3$. Then we can easily check that the adjacency matrix with such a selection is indeed nilpotent, making it an optimal solution in~\eqref{eq: exact optimization}. Finally, by Algorithm~\ref{alg:modified DAG to dynamic circuit}, the compiled circuit has circuit width $3$.
\end{proof}

\begin{proposition}[Pairwisely entangled quantum circuit]\label{prop:pairwisely entangled}

A pairwisely entangled quantum circuit with $n \geq 2$ qubits and $l \geq 1$ pairwise layers is reducible if and only if $l \leq \lceil n/2 \rceil - 1$. In this case, a feasible solution for optimization~\eqref{eq: exact optimization} can be taken at $\sum_{i=0}^{n-(2l+2)} E^{n}_{i,i+2l+1}$ which is optimal when $l > (n-2)/4$. Consequently, the circuit can be compiled to an equivalent dynamic circuit with $2l+1$ qubits.

\end{proposition}
\begin{proof}
The proof is also similar to the proof of Proposition~\ref{prop:linearly entangled}.
For a pairwisely entangled quantum circuit with layer $l = 1$, the biadjacency matrix of the simplified DAG is given by
\begin{align}
    B = \sum\limits_{\substack{i,j=0 \\ |i-j| \leq 1}}^{n-1} E^{n}_{i,j} + \sum\limits_{\substack{i=0 \\ i \text{ even}}}^{n-3} E^{n}_{i,i+2} + \sum\limits_{\substack{i=0 \\ i \text{ odd}}}^{n-3} E^{n}_{i+2,i}
\end{align}
For any integer $l\geq 1$, we can check that
\begin{align}
    B^{\odot l} = \sum\limits_{\substack{i,j=0 \\ |i-j| \leq 2l-1}}^{n-1} E^{n}_{i,j} + \sum\limits_{\substack{i=0 \\ i \text{ even}}}^{n-(2l+1)} E^{n}_{i,i+2l} + \sum\limits_{\substack{i=0 \\ i \text{ odd}}}^{n-(2l+1)} E^{n}_{i+2l,i}.
\end{align}
Therefore, $B^{\odot l}$ is an all-one matrix if and only if $l \geq \lceil n/2 \rceil$. By Proposition~\ref{prop: biadjacency matrix relation of composite circuit}, the biadjacency matrix for the multi-layer quantum circuit is exactly given by $B^{\odot l}$. So the circuit is irreducible if and only if $l \geq \lceil n/2 \rceil$, or equivalently, it is reducible if and only if $l \leq \lceil n/2 \rceil - 1$.

Moreover, we can compute the candidate matrix for the circuit with $l$ layers as $C = C_L + C_R$ where
\begin{align}
    C_L & = \sum\limits_{\substack{i,j=0 \\ i-j \geq 2l+1}}^{n-1} E^{n}_{i,j} + \sum\limits_{\substack{i=0 \\ i \text{ odd}}}^{n-(2l+1)} E^{n}_{i+2l,i}\\
    C_R & = \sum\limits_{\substack{i,j=0 \\ j-i \geq 2l+1}}^{n-1} E^{n}_{i,j} + \sum\limits_{\substack{i=0 \\ i \text{ even}}}^{n-(2l+1)} E^{n}_{i,i+2l}
\end{align}
Note that $C_L$ and $C_R$ have non-zero values only at the lower-left and upper-right corners of the matrix $C$. If $l >  (n-2) /4 $, the row and column indices of the non-zero entries of $C_L$ and $C_R$ have no intersection. So any candidate edge in $C_R$ will contradict to any candidate edge in $C_L$ because of the acyclic condition. If the optimal solution takes any value in $C_R$ (or $C_L$), then all candidate edges are taken from $C_R$ (or $C_L$). 

For the optimal compilation, we need to select the maximum number of $1$ elements in such a matrix under the conditions stated in Proposition~\ref{prop:compilation via BIP}. Without the nilpotent condition, the maximum number of $1$ elements that simultaneously satisfies conditions~\eqref{eq: submatrix condition},~\eqref{eq: row sum condition} and~\eqref{eq: col sum condition} is $n-(2l+1)$. Let us choose
\begin{align}
    F = \sum_{i=0}^{n-(2l+2)} E^{n}_{i,i+2l+1}.
\end{align}
with the total sum of $n - (2l+1)$. Then we can easily check that the adjacency matrix with such a selection is indeed nilpotent, making it an optimal solution in~\eqref{eq: exact optimization}.  Finally, by Algorithm~\ref{alg:modified DAG to dynamic circuit}, the compiled circuit has circuit width $2l+1$. Note that when $l \leq (n-2)/4$, the above choice of $F$ is also a feasible solution. This completes the proof.
\end{proof}

\begin{proposition}[Fully entangled quantum circuit]\label{prop:fully entangled}
A fully entangled quantum circuit with $n \geq 2$ qubits and $l \geq 1$ full layers is always irreducible.
\end{proposition}

\begin{proof}
It is clear that for fully entangled quantum circuit, its biadjacencdiamond-structuredy matrix of the simplified DAG is always an all-one matrix. Therefore, the circuit is irreducible. Alternatively, the irreducibility can be easily seen from Proposition~\ref{prop:reducibility from qubit reachability}.
\end{proof}

\subsubsection{Diamond-structured quantum circuits}
\label{sec: Diamond-structured circuits}

The Hartree-Fock method is a widely used approach in quantum chemistry for determining the electronic structure and energy of molecules and atoms. In~\cite{google2020hartree} the authors performed a VQE simulation of the binding energy of $H_6, H_8, H_{10}, H_{12}$ hydrogen chains where they used a variational ansatz based on basis rotations to prepare the Hartree-Fock state. Figure~\ref{fig:Diamond-structured quantum circuit} illustrates an example of the diamond-structured basis rotation circuit used for the $H_6$ chain. 

\begin{figure}[H]
    \centering
    \includegraphics[scale=0.8]{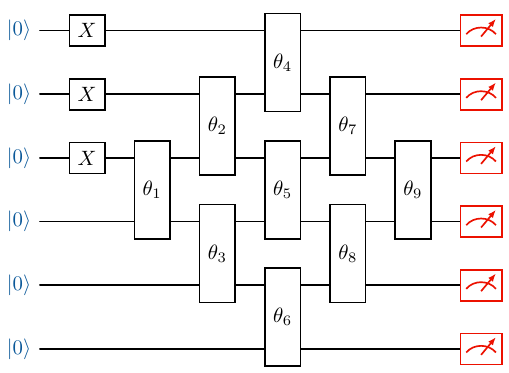}
    \caption{(Color online) The diamond-structured quantum circuit for the $H_6$ chain. Each box with a rotation angle $\theta$ represents a double qubit rotation gate.}
    \label{fig:Diamond-structured quantum circuit}
\end{figure}

The subsequent proposition provides the optimal compilation of the diamond-structured circuit designed to prepare the Hartree-Fock state.

\begin{proposition}[Diamond-structured quantum circuit]\label{prop:diamond-structured circuit}
A diamond-structured circuit for linear chain of $2n$ hydrogen atoms with $2n$ qubits is always reducible. The corresponding optimal solution for optimization~\eqref{eq: exact optimization} can be taken at $\sum_{i=0}^{n-2} E^{2n}_{i, i+n+1}$. Consequently, the circuit can be compiled to an equivalent dynamic circuit with $n+1$ qubits. 
\end{proposition}
\begin{proof}
For a diamond-structured circuit with $2n$ qubits, the biadjacency matrix of the simplified DAG is given by
\begin{align}
    B = \sum_{i=0}^{n-1} \sum_{j=0}^{i+n} E^{2n}_{i, j} + \sum_{i=n}^{2n-1} \sum_{j = i - n}^{2n-1} E^{2n}_{i, j}
\end{align}
We can further compute the candidate matrix for the circuit as $C = C_L + C_R$ where
\begin{align}
    C_L & = \sum_{i= n + 1}^{2n-1} \sum_{j=0}^{i - n - 1} E^{2n}_{i, j} \\
    C_R & = \sum_{i=0}^{n - 2} \sum_{j= i + n + 1}^{2n-1} E^{2n}_{i, j}
\end{align}

It is worth noting that the matrices $C_L$ and $C_R$ contain non-zero entries only at the lower-left and upper-right corners. Since there is no overlap in the row and column indices of the non-zero entries between $C_L$ and $C_R$, any candidate edge in $C_L$ would conflict with any candidate edge in $C_R$ because of the acyclic condition. Consequently, if the optimal solution selects any edge from $C_R$ (or $C_L$), it necessitates that all candidate edges are drawn exclusively from $C_R$ (or $C_L$).

To achieve the optimal compilation, our objective is to maximize the number $1$ elements in such a matrix under the conditions stated in Proposition~\ref{prop:compilation via BIP}. In the absence of the nilpotent condition, we have the flexibility to select
\begin{align}
    F = \sum_{i=0}^{n-2} E^{2n}_{i, i+n+1}
\end{align}
with the total sum of $n - 1$. Then we can easily check that the adjacency matrix with such a selection is indeed nilpotent, making it an optimal solution in~\eqref{eq: exact optimization}. Finally, by Algorithm~\ref{alg:modified DAG to dynamic circuit}, the compiled circuit has circuit width $n+1$, which completes the proof.
\end{proof}

\subsection{Measurement-based quantum computation}
Measurement-based quantum computation (MBQC) constitutes an alternative quantum computation model~\cite{briegel2009measurement}. This model guides the computation process by measuring a portion of the qubits within an entangled state, while the remaining unmeasured qubits evolve correspondingly. The MBQC computation can be divided into three primary steps.
The initial step involves preparing a resource state, which is a highly entangled many-body quantum state. This state can be pre-generated offline and is independent of specific computational tasks. Subsequently, the second step entails sequentially performing single-qubit measurements on each qubit of the prepared resource state. These measurements might be adaptive, meaning that later measurement choices can depend on the outcomes of prior measurements.
In the third step, classical data processing is applied to the measurement outcomes in order to derive the necessary computational results. An illustration of the resource state is provided in Figure~\ref{fig:mbqc}(a), where the grid signifies a widely used quantum resource state referred to as the cluster state. Within this grid, each vertex represents a qubit (typically initialized as a plus state $\ket{+} = (\ket{0}+\ket{1})/\sqrt{2}$), while edges connecting the vertices symbolize controlled-phase gates operating on the linked qubits.

\begin{figure}[ht]
    \centering
    \includegraphics[width=0.8\textwidth]{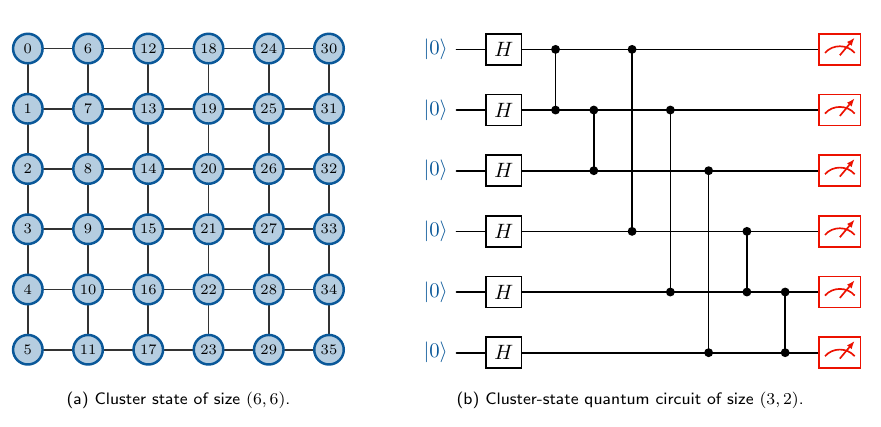}
    \vspace{-4mm}
    \caption{(Color online) An example of cluster state and cluster-state quantum circuit.}
    \label{fig:mbqc}
\end{figure}

Measurement-based quantum computation plays a pivotal role in blind quantum computation, a cloud-based quantum computing scheme that enables users to conduct calculations privately in a quantum network without disclosing any information to the server~\cite{fitzsimons2017private}. This capability holds great promise for the future quantum internet. Nonetheless, the standard approach of quantum computation in cloud servers typically operates in the quantum circuit model. Integrating an MBQC algorithm into this model necessitates a considerable number of qubits. Specifically, MBQC applied to a cluster state of size $(w, d)$—where $w$ signifies the number of rows and $d$ signifies the number of columns—can be translated directly into a quantum circuit of size $w d$. We refer to this circuit as a \emph{cluster-state quantum circuit} of size $(w, d)$. As shown in Figure~\ref{fig:mbqc}(b), an initial Hadamard gate is applied to each qubit to prepare a plus state. Subsequently, controlled-phase gates are employed based on the topology of the grid. Finally, a measurement is taken on each qubit. It is important to note that measurements are typically applied sequentially to the qubits column by column. Measurements on the same column are independent, while measurements on later columns may depend on the results of the previous columns.

The reducibility of cluster-state quantum circuits depends on the order of the controlled-phase gates. The following result is based on the specific order in which we implement the controlled-phase gates. First, we implement the gates in the first column, then those between the first and second columns, and so on. Additionally, the qubits in the first column of a cluster state can be initialized in an entangled quantum state.

The following result demonstrates that this circuit can be effectively reduced to a dynamic quantum circuit with $w + 1$ qubits, which is independent of the parameter $d$. This reduction significantly reduces the resource requirements for running MBQC algorithms on quantum computers that are primarily designed for circuit-based operations.

\begin{proposition}[Cluster-state quantum circuit.]\label{prop:mbqc}
A cluster-state quantum circuit of size $(w, d)$ is reducible for any $d \geq 2$. The corresponding optimal solution for optimization~\eqref{eq: exact optimization} can be taken at $\sum_{i=0}^{wd - (w+2)} E^{wd}_{i,i+w+1}$ in general. Consequently, the circuit can be compiled to an equivalent dynamic circuit with $w+1$ qubits. 
\end{proposition}
\begin{proof}
For a cluster-state quantum circuit of size $(w, d)$, its biadjacency matrix of the simplified DAG is given by a block matrix:
\begin{align}
    B = \sum_{i=1}^{d-1} E^{d}_{i,i-1}  \ox I_{w} + \sum_{i=0}^{d - 1} \sum_{j=0}^{d - i - 1} E^{d}_{i,i+j} \ox D_j
\end{align}
where each block
\begin{align}
    D_k = \sum_{i=0}^{w - 1} \sum_{j = \max(i-k-1, 0)}^{w-1} E^{w}_{i,j}.
\end{align}
We can further compute the candidate matrix as
\begin{align}
    C = \sum_{i=0}^{d-1} \sum_{j = 0}^{i} E^{d}_{i,i-j} \ox G_j + \sum_{i=0}^{d-2} E^{d}_{i,i+1} \ox (J_w - I_w) +  \sum_{i=0}^{d-1}\sum_{j=i+2}^{d-1} E^{d}_{i,j} \ox J_w,
\end{align}
where
\begin{align}
    G_k = \sum_{i=0}^{w-1}\sum_{j=i+k+2}^{w-1} E^{w}_{i,j}.
\end{align}
We need to select as many as $1$ elements within the conditions stated in Proposition \ref{prop:compilation via BIP}. Without the nilpotent condition, we can choose
\begin{align}
    F = \sum_{i=0}^{wd - (w+2)} E^{wd}_{i,i+w+1}.
\end{align}
which simultaneously satisfies conditions~\eqref{eq: submatrix condition},~\eqref{eq: row sum condition} and~\eqref{eq: col sum condition}. Then we can easily check that the adjacency matrix with such a selection is indeed nilpotent, making it a feasible solution in~\eqref{eq: exact optimization}. Finally, by Algorithm~\ref{alg:modified DAG to dynamic circuit}, the compiled circuit has circuit width $w+1$. Since the quantum state in the first column of qubits can be assigned to an arbitrary global quantum state in general, this requires at least $w$ qubits. Moreover, no matter which qubit in the first column to measure first, we need at least one more qubit to apply for the controlled-phase gate between the qubit to measure and the adjacent qubit on its right. This makes the circuit has at least $w+1$ qubits, indicating the optimality of our compilation.
\end{proof}

\begin{remark}\label{rmk: brickwork}
Similar compilations also work for measurement-based quantum computation with other graph states, such as the brickwork state used in~\cite{broadbent2009universal}.
\end{remark}

\section{Implementation of DCKF algorithm}
\label{sec:dckf implementation}

Since the DCKF algorithm described in~\cite{decross2022qubitreuse} is not publicly available, we implemented this algorithm based on our interpretation of their results. We observed that the causal cone of an output qubit $q_i$ in their algorithm corresponds exactly to the set of roots that can reach the terminal of $q_i$. We effectively obtained this information by identifying all non-zero indices in the $i$-th column of the biadjacency matrix of the quantum circuit.

Our initial step involved executing Algorithm~\ref{alg:reducibility through Boolean matrix} to obtain the biadjacency matrix and compute causal cones for all outputs. Subsequently, we applied their greedy strategy to determine the measurement order and identify how qubits are reused. It is important to note that, in order to complete a measurement, all qubits within its causal cone that have not yet been measured should be activated, while registers of all previously measured qubits are available for reuse.

Therefore, when a qubit $q_i$ needed to be activated, we first identified all unoccupied registers from the existing ones. If an available register existed, $q_i$ was loaded onto the first available one. Conversely, if all registers were occupied, a new register was initialized to accommodate $q_i$. Additionally, after the measurement of a qubit, the register it occupied was recycled for subsequent use. Upon the completion of all measurements in the input static circuit, the sequence in which registers were occupied was translated into a list of edges, which were subsequently incorporated into the DAG representation of the circuit. Finally, the compilation process was completed by applying Algorithm~\ref{alg:modified DAG to dynamic circuit}. The full implementation of the DCKF algorithm is provided in Algorithm~\ref{alg:dckf implementation}.

\BlankLine

\begin{algorithm}[!htb]
\small
\caption{DCKF Algorithm}
\label{alg:dckf implementation}
\LinesNumbered

\KwIn{\\  \begin{tabular}{p{3.5cm}l}
     \textit{StaticCircuit}  &  the instruction list of a static quantum circuit\\ 
\end{tabular}}

\BlankLine

\KwOut{\\  \begin{tabular}{p{3.5cm}l}
    \textit{DynamicCircuit}  &  the instruction list of the compiled dynamic quantum circuit\\ 
\end{tabular}}

\BlankLine

Let $n$ be the width of the static quantum circuit\;

Run Algorithm~\ref{alg:Converting static quantum circuit to DAG} to get the DAG representation \textit{Digraph} of the circuit\;

Let \textit{Roots} and \textit{Terminals} be the set of roots and terminals, respectively\;

Run Algorithm~\ref{alg:reducibility through Boolean matrix} to get the biadjacency matrix $B$ of the simplified DAG of the circuit\;

\BlankLine

Initialize empty lists \textit{CausalCones}, \textit{Register}, \textit{Occupation}, \textit{AddedEdges} and an empty set $C_q$\;

Initialize a set \textit{UnmeasuredQubits} as $\{0, 1, ..., n-1\}$\;

\BlankLine

\For{$i=0$ {\rm \textbf{to}} $n-1$}{
    Let $C$ be the set of row indices of non-zero entries in the $i$-th column of $B$\;
    Append $C$ to \textit{CausalCones}\;
}

\For{$t=0$ {\rm \textbf{to}} $n-1$}{
    Initialize \textit{ConeSize} as $n + 1$\;
    \ForEach{Qubit {\rm \textbf{in}} UnmeasuredQubits}{
        Calculate \textit{Union} $= C_q$ $\cup$ \textit{CausalCones}[\textit{Qubit}]\;
        \If{the length of Union < ConeSize}{
            Set \textit{ConeSize} to the length of \textit{Union}\;
            Set \textit{NextMeasure} to \textit{Qubit}\;
        }
    }
    Calculate \textit{ActiveQubits} as \textit{CausalCones}[\textit{NextMeasure}] $-$ \textit{MeasureOrder}\;
    \ForEach{Qubit {\rm \textbf{in}} ActiveQubits}{
        \If{Qubit not {\rm \textbf{in}} Register}{
            Record all indeices of None in \textit{Register} in \textit{Available}\;
            \uIf{Available is not empty}{
                Record \textit{Available}[$0$] as \textit{Address}\;
                Record the last element in \textit{Occupation}[\textit{Address}] as $t$\;
                Append an edge (Terminals[$t$], Roots[\textit{Qubit}]) to \textit{AddedEdges}\;
                Append \textit{Qubit} to \textit{Occupation}[\textit{Address}]\;
            }
            \Else{
                Set \textit{Address} to the length of \textit{Register}\;
                Append a list [\textit{Qubit}] to \textit{Occupation}\;
            }
            Set \textit{Register}[\textit{Address}] to \textit{Qubit}\;
        }
    }
    Identify the index of \textit{NextMeasure} in \textit{Register} as $m$ and set \textit{Register}[$m$] to None\;
    Append \textit{NextMeasure} to \textit{MeasureOrder}\;
    Remove \textit{NextMeasure} from \textit{UnmeasuredQubits}\;
    Calculate $C_q = C_q$ $\cup$ \textit{CausalCones}[\textit{NextMeasure}]\;
}

Add \textit{AddedEdges} to \textit{Digraph} and get the \textit{ModifiedGraph}\;

Run Algorithm~\ref{alg:modified DAG to dynamic circuit} to get the compiled circuit \textit{DynamicCircuit}\;

\Return \textit{DynamicCircuit}

\end{algorithm}

\section{Experimental Results Continued}
\label{app: Experimental Results Continued}

\subsection{Experimental evaluation of hybrid algorithm}

To demonstrate the tradeoff between solution optimality and computational time complexity, we assessed the performance and runtime of the hybrid algorithm (Algorithm~\ref{alg:hybrid}) at various hierarchy levels (0, 1, 2, 3) for max-cut QAOA circuits on U3R graphs with $p=1$. The numerical experiments were conducted on a MacBook Pro (2020) with an Intel i5 processor and 16GB memory.
Figure~\ref{fig:hybrid mrv qaoa p1} depicts the average compiled circuit width and algorithm runtime over 20 random graphs for each qubit number. As expected, the results show that as the hierarchy level increases, we attain better compilations. However, the algorithm runtime also increases proportionally. Therefore, the hybrid algorithm shall be a suitable choice for scenarios where minimizing the compiled circuit width is of top priority, provided that the runtime remains acceptable.

\begin{figure}[!htb]
    \centering
    \begin{subfigure}{0.48\textwidth}
        \includegraphics[width=0.9\textwidth]{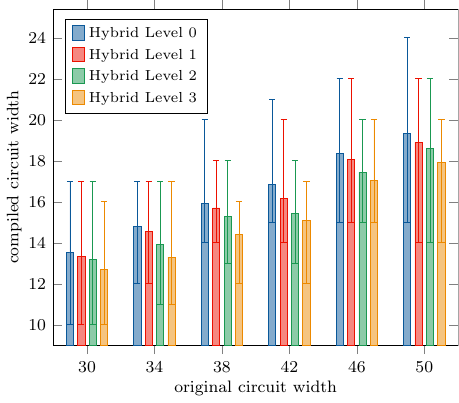}
    \end{subfigure}
    \begin{subfigure}{0.495\textwidth}
        \includegraphics[width=0.9\textwidth]{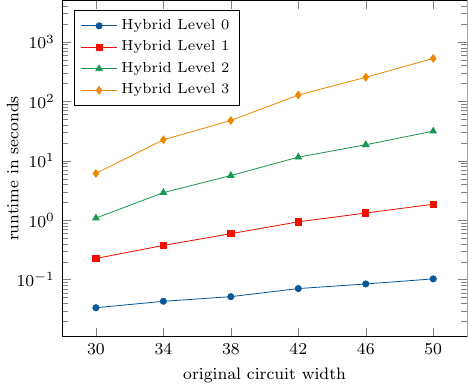}
    \end{subfigure}
    \vspace{-2.5mm}
    \caption{(Color online) Compiled circuit width (left) and algorithm runtime in seconds (right) as a function of the original circuit width of max-cut QAOA with $p=1$, compiled by hybrid algorithm~\ref{alg:hybrid} with different hierarchy levels. The ploted error bars in the left figure correspond to the maximum and the minimum compiled width over 20 evaluated instances. The runtime in the right figure represents the average runtime.}
    \label{fig:hybrid mrv qaoa p1}
\end{figure}

\subsection{Experimental evaluation of random circuits}

Figure~\ref{fig:further result on random circuit}(a) depicts the complete results of the reducibility factor of random circuits. Notably, for all evaluated cases, our greedy heuristic outperforms the DCKF algorithm in approximately $87.6\%$ of cases, with a strict advantage over $49.6\%$ of random circuits. Additionally, a comparative analysis of the performance between our greedy heuristic and the improved DCKF algorithm with the first qubit search technique on random circuits is depicted in Figure~\ref{fig:further result on random circuit}(b). In cases where the reducibility factor is relatively large, our algorithm showcases superiority in approximately $91.2\%$ of instances, maintaining a strict advantage in over $36\%$ of random circuits.

\begin{figure}[!htb]
    \begin{subfigure}[t]{0.48\textwidth}
        \centering
        \includegraphics[width=0.9\textwidth]{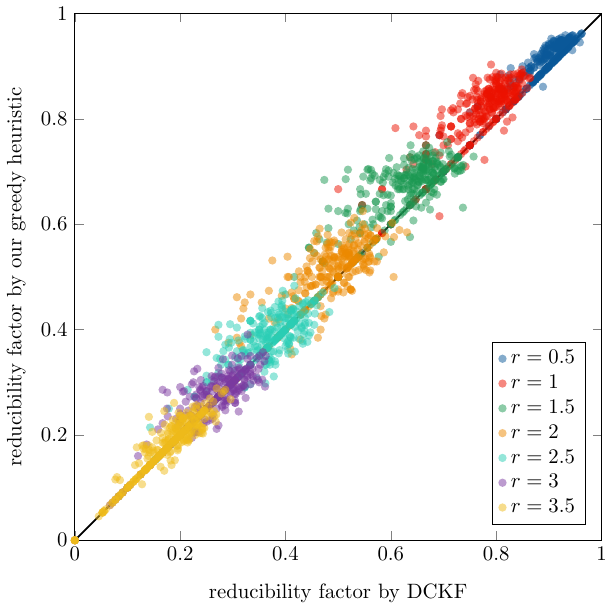}
        \vspace{-1mm}
        \caption*{\scriptsize (a)}
    \end{subfigure}
    \hspace{1mm}
    \begin{subfigure}[t]{0.48\textwidth}
        \centering
        \includegraphics[width=0.9\textwidth]{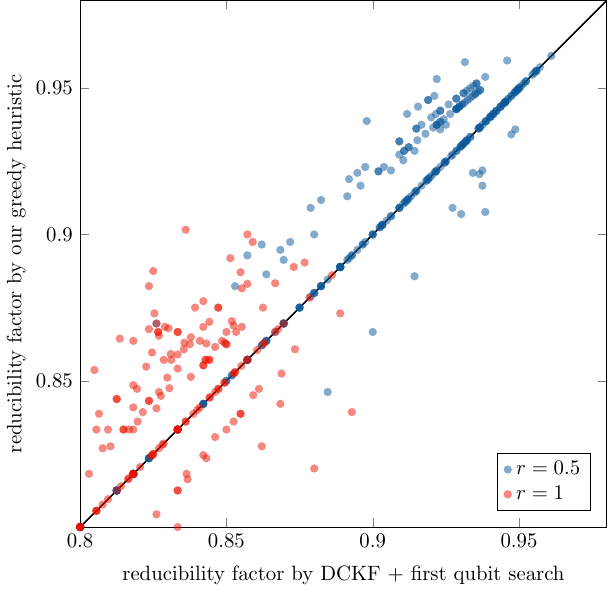}
        \vspace{-1mm}
        \caption*{\scriptsize (b)}
    \end{subfigure}
    \vspace{-1.5mm}
    \caption{(Color online) (a) The reducibility factor of the random circuits evaluated using our greedy heuristic algorithm~\ref{alg:greedy} and the DCKF algorithm; (b) The reducibility factor of the random circuits evaluated using our greedy heuristic algorithm~\ref{alg:greedy} and the improved DCKF algorithm (DCKF + first qubit search).}
    \label{fig:further result on random circuit}
\end{figure}

\subsection{Noisy simulation}

To further demonstrate the practical efficacy of the proposed methods, we design a noisy simulation of an 11-qubit Bernstein-Vazirani (BV) algorithm, specifically targeting the real-world 11-qubit trapped-ion quantum computer reported in~\cite{wright2019benchmarking}. The logical-physical qubit mapping used and the probability of getting the correct outcome in each circuit are listed in Table~\ref{tab:qubit mapping} below. For dynamic circuit compilation, we always reuse the logical qubit $q_0$ for other qubits. This example is for an illustrative purpose under certain assumptions. We assume that CNOT gates admit the same performance as the native X-X gates. In addition, we utilize depolarizing noise as the noise model to simulate the reported error rate, which may not perfectly replicate real-world conditions. Nonetheless, the simulation can be more fine-grained for interested readers and the experiments can be readily conducted on quantum hardware when available. 

\setlength\extrarowheight{1pt}
\begin{table}[!htb]
\scriptsize	
\caption{The mapping of logical qubits ($q$) to physical qubits ($Q$) and the probability of obtaining the correct outcome in each circuit.}
\label{tab:qubit mapping}
\centering
\begin{tabular}{ccccccccccccc}
\toprule[1pt]
              & $q_0$ & $q_1$ & $q_2$ & $q_3$ & $q_4$ & $q_5$ & $q_6$ & $q_7$ & $q_8$ & $q_9$    & $q_{10}$ & probability \\
\midrule
BV\_11        & $Q_3$ & $Q_0$ & $Q_2$ & $Q_4$ & $Q_5$ & $Q_6$ & $Q_7$ & $Q_8$ & $Q_9$ & $Q_{10}$ & $Q_1$    & $66.2\%$    \\
BV\_10        & $Q_3$ & $Q_0$ & $Q_4$ & $Q_5$ & $Q_6$ & $Q_7$ & $Q_8$ & $Q_9$ & $Q_{10}$ & $Q_1$ & -        & $66.7\%$    \\
BV\_9         & $Q_3$ & $Q_0$ & $Q_4$ & $Q_6$ & $Q_7$ & $Q_8$ & $Q_9$ & $Q_{10}$ & $Q_1$ & -     & -        & $68.5\%$    \\
BV\_8         & $Q_3$ & $Q_0$ & $Q_4$ & $Q_6$ & $Q_7$ & $Q_8$ & $Q_9$ & $Q_1$ &  -       & -     & -        & $69.1\%$    \\
BV\_7         & $Q_3$ & $Q_0$ & $Q_4$ & $Q_7$ & $Q_8$ & $Q_9$ & $Q_1$ &  -    &  -       & -     & -        & $70.5\%$    \\
BV\_6         & $Q_3$ & $Q_0$ & $Q_4$ & $Q_7$ & $Q_9$ & $Q_1$ &  -    &  -    &  -       & -     & -        & $71.5\%$    \\
BV\_5         & $Q_3$ & $Q_0$ & $Q_7$ & $Q_9$ & $Q_1$ &  -    &  -    &  -    &  -       & -     & -        & $73.1\%$    \\
BV\_4         & $Q_3$ & $Q_0$ & $Q_7$ & $Q_1$ &  -    &  -    &  -    &  -    &  -       & -     & -        & $73.6\%$    \\
BV\_3         & $Q_3$ & $Q_0$ & $Q_1$ &  -    &  -    &  -    &  -    &  -    &  -       & -     & -        & $74.1\%$    \\
BV\_2         & $Q_3$ & $Q_1$ &  -    &  -    &  -    &  -    &  -    &  -    &  -       & -     & -        & $74.3\%$    \\
\bottomrule[1pt]  
\end{tabular}
\end{table}

\section{Open Problems}
\label{sec:open problems}
In this section, we provide some detailed discussion of several open problems related to our work.

\subsection{Compiling dynamic quantum circuit}
\label{subsec:compiling dynamic quantum circuit}

The most significant open problem pertains to extending the approach outlined in the main text to dynamic quantum circuits. Specifically, we can consider to begin with a dynamic quantum circuit and seek to compile it into a circuit with a smaller width. However, this extension introduces greater complexities, as it necessitates determining whether a given dynamic quantum circuit can be further reduced in terms of qubit count. This problem is intricately linked to the optimality of quantum circuit compilation. Specifically, consider the decision problem of whether a given static circuit can be compiled into a dynamic circuit of a given size $k$. Then the optimization version of the problem reduces to this decision problem through a binary search running logarithmically in the size of the original circuit. Moreover, the task of determining if a dynamic quantum circuit can be further reduced is equivalent to questioning if the corresponding static circuit can be reduced to a circuit of size smaller than the current compilation. Combining these arguments, we can conclude that determining the reducibility of a dynamic circuit is as hard as finding the optimal compilation of a static circuit.

Therefore, it is crucial to emphasize that the applicability of our reducibility checking approaches, as proposed in Propositions~\ref{prop:reducibility through DFS},~\ref{prop:reducibility from qubit reachability} and~\ref{prop:reducibility through matrix}, is currently limited to static circuits. We provide an illustrative example below to highlight the challenges and considerations involved in compiling a dynamic quantum circuit. Examining the dynamic circuit illustrated in Figure~\ref{fig:reducible dynamic circuit} and the corresponding DAG representation in Figure~\ref{fig:reducible dynamic dag} it becomes apparent that any terminal vertex is reachable from any root vertex. Consequently, it shall be concluded that this dynamic circuit is irreducible.

\begin{figure}[H]
    \centering
    \begin{subfigure}[b]{.45\textwidth}
        \centering
        \includegraphics[scale=0.8]{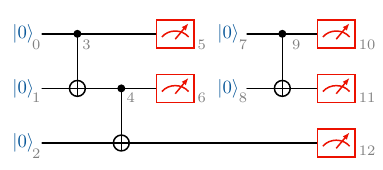}
        \caption{An `irreducible' dynamic circuit.}
        \label{fig:reducible dynamic circuit}
    \end{subfigure}
    \begin{subfigure}[b]{.45\textwidth}
        \centering
        \includegraphics[scale=0.8]{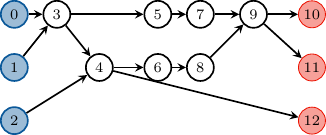}
        \vspace{2.5mm}
        \caption{DAG representation of the dynamic circuit}
        \label{fig:reducible dynamic dag}
    \end{subfigure}
    \caption{(Color online) A dynamic quantum circuit and its DAG representation.}
    \label{fig:reducible dynamic circuit and its dag}
\end{figure}

Nevertheless, further exploration reveals that the dynamic circuit can actually be compiled into the dynamic circuit shown in Figure~\ref{fig:minimum dynamic circuit}, with the number of qubits reduced by one.

\begin{figure}[ht]
    \centering
    \includegraphics[width=0.5\textwidth]{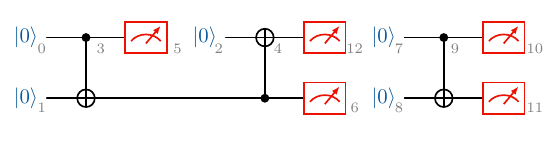}
    \vspace{-6mm}
    \caption{(Color online) The dynamic quantum circuit reduced from the dynamic circuit in Figure~\ref{fig:reducible dynamic circuit}.}
    \label{fig:minimum dynamic circuit}
\end{figure}

Interestingly, the dynamic circuit depicted in Figure~\ref{fig:minimum dynamic circuit} can be regarded as a dynamic circuit compilation outcomes of the static circuit in Figure~\ref{fig:reducible static circuit}. As shown in Figure~\ref{fig:different schemes}, these two compilation schemes can be implemented by adding dashed orange edges and dashed green edges to the DAG representation of the static circuit in Figure~\ref{fig:reducible static circuit}, respectively. Furthermore, when dealing with the compilation of dynamic quantum circuits, special consideration is imperative for classically controlled gates as these may introduce additional dependencies within the DAG of quantum circuits. Therefore, a reasonable starting point in the compilation of a dynamic circuit is to initiate the conversion of the dynamic circuit into an equivalent static circuit. Following this conversion, the methodology in the main text can be employed to explore the existence of an enhanced compilation scheme in comparison to the original dynamic circuit.

\begin{figure}[!htb]
    \centering
    \begin{subfigure}[b]{.4\textwidth}
        \centering
        \includegraphics[width=0.7\textwidth]{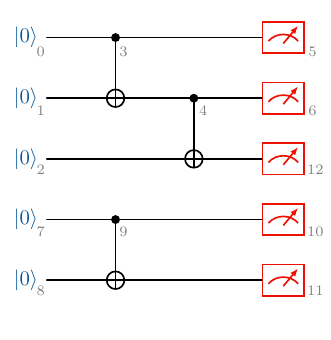}
        \vspace{-5mm}
        \caption{The original reducible static circuit.}
        \label{fig:reducible static circuit}
    \end{subfigure}
    \begin{subfigure}[b]{.44\textwidth}
        \centering
        \includegraphics[width=0.6\textwidth]{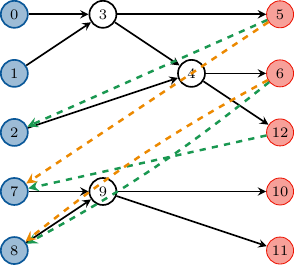}
        \caption{Different compilation schemes.}
        \label{fig:different schemes}
    \end{subfigure}
    \caption{(Color online) A static quantum circuit that can be reduced to both the dynamic circuit in Figure~\ref{fig:reducible dynamic circuit} and~\ref{fig:minimum dynamic circuit}. The compilation scheme indicated by the dashed orange edges results in the dynamic circuit shown in Figure~\ref{fig:reducible dynamic circuit}, while the compilation scheme by the dashed green edges leads to the circuit in Figure~\ref{fig:minimum dynamic circuit}.}
    \label{fig:different compilation schemes on static circuit}
\end{figure}

The complete algorithm for converting a dynamic circuit to a static one is provided in Algorithm~\ref{alg:expand algorithm}. The central concept involves assigning a new quantum register for each reset operation in the dynamic circuit. To elaborate, we traverse the dynamic circuit's instructions, and whenever we encounter a reset operation, we allocate a new qubit based on the current circuit width. All operations performed on the reset qubit after this reset operation are subsequently transferred to the newly allocated qubit. Through this process, the dynamic circuit can be transformed into an equivalent static circuit, essentially representing the reverse of our dynamic circuit compilation procedure. It is worth noting that Algorithm~\ref{alg:expand algorithm} can be adapted to accommodate dynamic circuits that feature classically controlled quantum gates with slight modification.

\BlankLine

\begin{algorithm}[!htb]
\small
\caption{Converting dynamic circuit to static circuit}
\label{alg:expand algorithm}
\LinesNumbered

\KwIn{\\  \begin{tabular}{p{3.5cm}l}
     \textit{DynamicCircuit}  & the instruction list of a dynamic quantum circuit\\ 
\end{tabular}}

\BlankLine

\KwOut{\\  \begin{tabular}{p{3.5cm}l}
     \textit{StaticCircuit}      & the instruction list of the compiled quantum circuit\\
\end{tabular}}

\BlankLine

Let $n$ be the dynamic quantum circuit width\;

Initialize two empty lists \textit{StaticCircuit} and \textit{Measurements}\;

\BlankLine

\ForEach{instruction {\rm \textbf{in}} DynamicCircuit}{
    \uIf{instruction is a measurement}{
    Append \textit{instruction} to \textit{Measurement}\;
    }
    \uElseIf{instruction is a reset}{
    Record the value in QUBIT of \textit{instruction} as \textit{ResetQubit}\;
    Initialize an empty list \textit{PostResets}\;
    Record all instructions subsequent to \textit{instruction} in \textit{DynamicCircuit} to \textit{PostResets}\;
    Remove all instructions subsequent to \textit{instruction} from \textit{DynamicCircuit}\;
        \ForEach{PostInstruction {\rm \textbf{in}} PostResets}{
            \ForEach{Qubit {\rm \textbf{in}} QUBIT of PostInstruction}{
                \lIf{Qubit is equal to ResetQubit}{
                Update \textit{Qubit} to $n$}
            }
        }
    Update $n$ to $(n+1)$\;
    Append \textit{PostReset} to \textit{DynamicCircuit}\;
    }
    \ElseIf{instruction is other single-/two-qubit gate}{
    Append \textit{instruction} to \textit{StaticCircuit}\;
    }
}

Append \textit{Measurement} to \textit{StaticCircuit}\;

\Return \textit{StaticCircuit}

\end{algorithm}

\BlankLine

\subsection{NP-hardness of the dynamic circuit compilation problem}

We have noted a remarkable similarity between the graph optimization problem explored in this work and the well-established Maximum Acyclic Subgraph (MAS) problem in graph theory. The MAS problem for a graph $(V, E)$ involves identifying a subgraph $(V, E')$ with $E' \subseteq E$, such that it contains no cycles and has the maximum number of edges. Our problem can be reformulated as a constrained version of the MAS problem. Let $(V, E)$ be the simplified DAG and $\bar{E}$ be the set of candidate edges. We can first incorporate all candidate edges into the simplified DAG and try to identify the maximum acyclic subgraph in the resultant graph $(V, E \cup \bar{E})$ under the additional constraint that only edges in set $\bar{E}$ can be removed. Suppose the solution to this problem yields a subgraph $(V, E \cup \bar{E}')$ where $\bar{E}' \subseteq \bar{E}$, then the maximum matching in the graph $(V, \bar{E}')$ corresponds to a feasible solution to our problem. It is worth noting that the MAS problem has been well-established as an NP-hard problem~\cite{karp1972reducibility}. Therefore, we expect that the optimal compilation of a quantum circuit may similarly belong to the class of NP-hard. However, the formal proof of its NP-hardness remains open. Additionally, several algorithms have been proposed in the literature to find approximate solutions for the MAS problem~\cite{hassin1994approximations, cvetkovic2020maximal}, which may offer valuable insights for seeking approximate solutions to the dynamic circuit compilation.

\end{document}